\pdfoutput=1
\documentclass[acmtog]{acmart}
\acmSubmissionID{571}
\AtBeginDocument{%
  \providecommand\BibTeX{{%
    \normalfont B\kern-0.5em{\scshape i\kern-0.25em b}\kern-0.8em\TeX}}}

\definecolor{MyRed}{rgb}{0.65,0.07,0.09}

\definecolor{MyGreen}{rgb}{0.18,0.55,0.09}

\usepackage{wrapfig}
\usepackage{diagbox}
\usepackage{verbatim}
\usepackage{amssymb}
\usepackage{latexsym}
\usepackage{lineno}
\usepackage{xspace}
\usepackage{color}
\graphicspath{{figures/}}
\usepackage{amsmath,bm}
\usepackage{psfrag}
\usepackage{mathtools}
\newtheorem{theorem}{Theorem}[section]
\theoremstyle{definition}

\usepackage{multirow, tabularx}
\newcolumntype{Y}{>{\centering\arraybackslash}X}
\usepackage{capt-of}%
\usepackage{array, boldline, makecell, booktabs}

\newcommand{\Mod}[1]{\ (\mathrm{mod}\ #1)}
\usepackage{array}

\usepackage{pifont}% http://ctan.org/pkg/pifont

\usepackage{psfrag}
\usepackage{makecell}
\usepackage{algorithm}
\usepackage{threeparttable,booktabs}
\usepackage{graphicx}
\usepackage{subcaption}
\usepackage{flushend}

\usepackage{algpseudocode}
\newtheorem{remark}{Remark}

\usepackage{xcolor}
\usepackage{graphicx}

\newcommand{\revv}[1]{{#1}}
\newcommand{\revvdel}[1]{}
\newcommand{\revvv}[1]{{#1}}

\newcommand{\junweirev}[1]{{#1}}

\newcommand{\WrapFig}[5]
{
\begin{wrapfigure}{r}{{#4}\textwidth}
	\vspace{-.5\baselineskip}
	\centering
	\includegraphics[height={#5}in]{#1}
	\caption{#2}
	\vspace{-.25\baselineskip}
	\label{#3}
\end{wrapfigure}
}

\newif\ifverbose
\verbosetrue

\ifverbose

\newcommand{\vb}[1]{\textcolor{red}{#1}}
\else

\newcommand{\vb}[1]{}
\fi

\usepackage{url}
\usepackage{xcolor}
\definecolor{newcolor}{rgb}{.8,.349,.1}

\usepackage{enumitem}
\usepackage{svg}
\svgpath{{./img/}} % <- using \svgpath to avoid warning

\usepackage[normalem]{ulem}

\setcopyright{rightsretained}
\acmJournal{TOG}
\acmYear{2024} \acmVolume{43} \acmNumber{4} \acmArticle{76} \acmMonth{7}\acmDOI{10.1145/3658180}

\begin{document}

\author{Junwei Zhou}
\email{zhoujw@umich.edu}
\affiliation{
\institution{University of Michigan}
\country{USA}
}

\author{Duowen Chen}
\email{dchen322@gatech.edu}
\affiliation{
\institution{Georgia Institute of Technology}
\country{USA}
}

\author{Molin Deng}
\email{mdeng47@gatech.edu}
\affiliation{
\institution{Georgia Institute of Technology}
\country{USA}
}

\author{Yitong Deng}
\email{yitongd@stanford.edu}
\affiliation{
\institution{Stanford University}
\country{USA}
}

\author{Yuchen Sun}
\email{yuchen.sun.eecs@gmail.com}
\affiliation{
\institution{Georgia Institute of Technology}
\country{USA}
}

\author{Sinan Wang}
\email{wsn1226@connect.hku.hk}
\affiliation{
\institution{University of Hong Kong}
\country{China}
}

\author{Shiying Xiong}
\email{shiying.xiong@zju.edu.cn}
\affiliation{
\institution{Zhejiang University}
\country{China}
}

\author{Bo Zhu}
\email{bo.zhu@gatech.edu}
\affiliation{
\institution{Georgia Institute of Technology}
\country{USA}
\state{(Junwei Zhou's work was done during a remote intern with Georgia Institute of Technology).}
}

% \footnote{}

\title{Eulerian-Lagrangian Fluid Simulation on Particle Flow Maps}

\begin{abstract}
%\revv{We introduce Particle Flow Maps (PFM), a novel fluid simulation method that significantly enhances computational efficiency compared to the existing Neural Flow Maps (NFM) \cite{deng2023fluid} approach while maintaining high simulation accuracy.} 
\revv{We propose a novel Particle Flow Map (PFM) method to enable accurate long-range advection for incompressible fluid simulation.} The foundation of our method is the observation that a particle trajectory generated in a forward simulation naturally embodies a perfect flow map. Centered on this concept, we have developed an Eulerian-Lagrangian framework comprising four essential components: Lagrangian particles for a natural and precise representation of bidirectional flow maps; a dual-scale map representation to accommodate the mapping of various flow quantities; a particle-to-grid interpolation scheme for accurate quantity transfer from particles to grid nodes; and a hybrid impulse-based solver to enforce incompressibility on the grid. The efficacy of PFM has been demonstrated through various simulation scenarios, highlighting the evolution of complex vortical structures and the details of turbulent flows. Notably, compared to NFM, PFM reduces computing time by up to 49 times and memory consumption by up to 41\%, while enhancing vorticity preservation as evidenced in various tests like leapfrog, vortex tube, and turbulent flow.
\end{abstract}



\begin{teaserfigure}
 \centering
 \includegraphics[width=.99\textwidth]{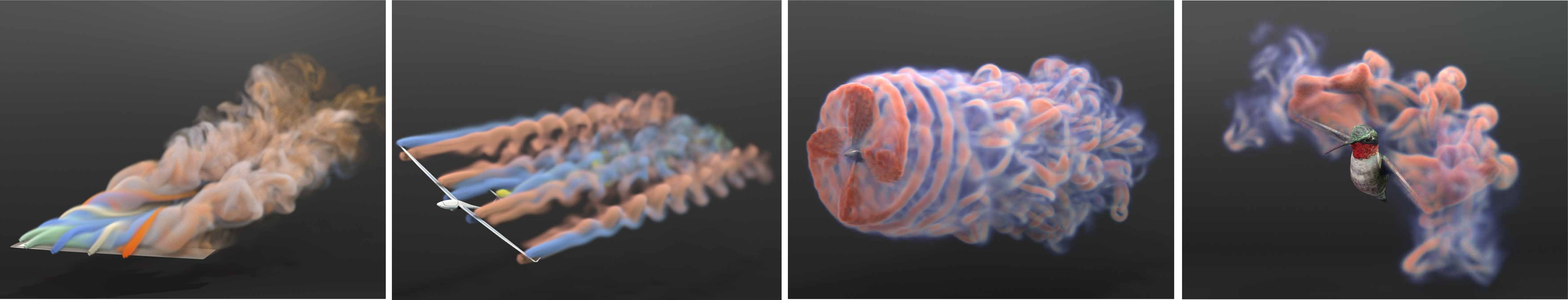}
 \caption{\textcolor{black}{Our novel particle flow map (PFM) methods can preserve vorticity by the state-of-\revvv{the}-art standard and excel in memory usage and computation time among methods that can give similar quality. Four simulation snapshots are shown using our PFM method including a delta wing plane at an angle-of-attack (\textit{left}), a flat wing plane (\textit{middle left}), a rotating propeller (\textit{middle right}) and a flapping hummingbird (\textit{right}).}}
 \label{fig:teaser}
\end{teaserfigure}

\maketitle

\section{Introduction}
%The ability to precisely resolve fluid advection is essential in simulating turbulent and vortical fluid flows, integral to a myriad of real-world phenomena and applications in various disciplines. 
%Consequently, enhancing numerical advection accuracy and its conservative characteristics has become a focal point in both computer graphics and computational physics, as evidenced by an extensive range of prior research encompassing high-order numerical schemes \cite{mullen2009energy, jameson1981numerical, cho2018dual, selle2008unconditionally, kim2006advections}, adaptive data structures \cite{rasmussen2004directable, losasso2004simulating, losasso2006spatially, aanjaneya2017power, ando2020practical, setaluri2014spgrid, mcadams2010parallel, irving2006efficient, houston2006hierarchical, chentanez2011real, zhu2013new}, and geometric algorithms \cite{mullen2011discrete, elcott2007stable}. 

\revv{Preserving vortical flow details during fluid simulations has garnered extensive attention in the fields of computer graphics and computational physics. Over recent years, researchers have developed various strategies to address this challenge, notably through advancements on three fronts: the development of conservative numerical schemes, such as the Affine Particle-in-Cell (APIC) method \cite{jiang2015affine}; vorticity-adapted geometric representations, for instance, the covector \revvv{fluids} (CF) approach \cite{nabizadeh2022covector}; and the creation of accurate, long-range flow maps, exemplified by the bidirectional mapping (BiMocq) method \cite{qu2019efficient}.}

\revv{A common foundational idea underpinning these three categories of works, which motivates the further design of structural preserving numerical schemes, is the development of a \textit{long-range}, \textit{bidirectional} flow map capable of accurately transporting \textit{vorticity-related physical quantities} between time frames.} 
\revv{The recent work on the \textit{Neural Flow Maps} method (NFM) \cite{deng2023fluid} further exemplifies this methodological philosophy (i.e., geometric representation + conservative numerical scheme + bidirectional mapping) by constructing accurate flow maps that transport fluid impulses across a long-range spatiotemporal domain. This method, requiring the online training of a neural network to compress the 4D spatiotemporal velocity field, facilitates a time-reversed marching scheme capable of achieving state-of-the-art vorticity preservation.}

%However, the employment of a time-reversed marching scheme for reconstructing flow maps in this approach, despite its accuracy, necessitates successive neural network training, significantly increasing computational costs and consequently diminishing its attractiveness for performance-sensitive applications. 

%Such a conundrum serves as the core motivation for this work --- to seek low-cost algorithms to create accurate flow maps. 
\revv{Along this line of research, we explore low-cost algorithms to enable long-range, bidirectional flow maps. The key problem we aim to address lies in finding the most fitting discrete representation for flow map geometries in the spatiotemporal domain.}  
In particular, while \citet{deng2023fluid} champions a neural representation that can be integrated to recover flow maps at Eulerian grid points, we believe that Lagrangian particles make for the most natural representation that promises to generate accurate flow maps with a drastically reduced cost. To see this, one can consider that the initial and final positions of a forward-simulated particle, trivially known during the simulation process, form a near-perfect sample of the actual bidirectional flow maps created by the fluid flow. In other words, accurate flow maps are obtained \textit{for free} in Lagrangian simulations.

However, although particles trivially offer accurate flow map samples, the fact that the final positions of these sampled Lagrangian paths are not aligned with the grid nodes leads to a core issue. 
%With these particles as samples, we can already accurately compute fluid quantities like impulse at their final positions, which are scattered, non-grid locations, but we eventually need to obtain these values at grid points to perform essential computations like Poisson solving and velocity projection. 
To counter this issue, ``virtual particle'' approaches, including the ones \revvv{employed by} NFM \cite{deng2023fluid}\revv{,} have been proposed, which trace virtual particles backward in time so that the samples will have their final positions aligned with grid points. Although virtual (backward) and real (forward) particles offer equally accurate samples of the flow map, the only difference being whether or not the endpoints align with the grid nodes, the amounts of computation that go into both are night-and-day different. While real particles are given as byproducts of the simulation, requiring no additional memory or computation, virtual particles require both a long history buffer and \revv{$O(n)$} number of substeps\revv{, where $n$ is the length of the flow map,} since solutions from past timesteps cannot be reused.

Our work takes a different path from the "virtual particle" approach. Instead of devoting excessive computational resources only to recover fluid impulse at grid points, we make use of the free samples of forward-simulated particles, compute accurate impulse at their non-grid end locations, and transfer these values onto grid points. The two perspectives\revv{, "virtual
particle" versus real particle,} are symmetric in construction: virtual particles combine accurate flow maps with a lossy grid-to-particle interpolation process since their \textit{initial} positions are not grid-aligned; real particles combine accurate flow maps with a lossy particle-to-grid process since their \textit{end} positions are not grid-aligned. Hence, from the accuracy standpoint, both methods are equivalent, whereas our proposed approach is more desirable by a large margin from the efficiency standpoint. Moreover, \revv{inspired by the conservative transfer approaches suggested by APIC \cite{jiang2015affine} and Taylor-PIC \cite{nakamura2023taylor},} we propose a novel adaptive flow-map scheme to refine the accuracy of the particle-to-grid process by employing flow maps to transport not only fluid quantities but also their gradients. %We devise both a long-range flow map and a short-range flow map to facilitate the evolution of these values on the particles. %The involvement of these evolved gradients in the particle-to-grid process substantially diminishes errors, ensuring the accurate transfer of precise fluid quantities from particles to grid points.

On this foundation, we introduce Particle Flow Maps (PFM), a simulation method that achieves accurate flow maps on particles with substantially reduced computational demands compared to NFM. By adopting particles as direct representatives of flow maps, we avoid the backtrack substeps and the training process of the neural buffer. This efficiency enables PFM to operate around \textit{10-20 times faster} on average and up to \textit{49 times faster} than NFM while delivering comparable or superior simulation outcomes. Moreover, by eliminating the neural buffer, PFM achieves a \textit{memory savings of 29\%-41\% over NFM}. We validated our approach in multiple examples, including leapfrogging vortices, vortex-tube reconnections, and solid-driven vortices. In each case, PFM consistently exhibits comparable or superior vortex preservation, energy conservation, and visual intricacy \revvv{to those} achieved by NFM, all while demanding significantly less computational and memory costs than NFM. 
%These computational advantages, highlighted by its rapid computing time, low memory cost, and better vorticity-preserving capability, establish PFM as an exceptionally practical solution for a wide spectrum of fluid simulation applications.

The key contributions of our work can be summarized as follows:
\begin{enumerate}
    \item We introduced a particle representation for long-range, bi-directional flow maps.
    \item We proposed a two-scale flow-map scheme, defined on a single particle trajectory, for mapping flow quantities with different reinitialization requirements.
    \item We developed a particle-to-grid interpolation scheme to transfer flow maps from particles to grid nodes.
    \item We presented a comprehensive Eulerian-Lagrangian solver based on the impulse fluid model, achieving state-of-the-art vorticity preservation capabilities while maintaining both low memory cost and high computational speed.
\end{enumerate}

\section{Related Work}
%Building upon the fluid simulation framework \citet{stam1999stable} proposed, many attemps where made to reduce the dissipation error caused by semi-lagrangian advection \cite{robert1981stable, sawyer1963semi}. Efforts to diminish viscosity artifacts, which stem from the simplistic advection scheme, have included the development of \textit{more advanced advection schemes}, the \textit{adoption of hybrid representations for fluids}, and \textit{the introduction of gauge variables}. The subsequent paragraphs will review the past researches conducted in each of these areas.

\paragraph{Advection Schemes}

The dissipation error inherent in the Stable Fluid framework \cite{stam1999stable} not only leads to unrealistic behavior in viscous fluid surfaces but also \revvv{causes} significant vorticity dissipation, especially when vortices are the primary focus. Various approaches have been adopted to address this issue. High-order interpolation schemes were proposed by \citet{losasso2006spatially} and \citet{nave2010gradient}, while \citet{kim2006advections} and \citet{selle2008unconditionally} introduced error compensation methods. The accuracy of backtracking was enhanced by \citet{jameson1981numerical, cho2018dual}, offering improvements over single-step semi-Lagrangian tracing. \citet{mullen2009energy} developed an unconditionally stable time integration scheme to combat numerical dissipation. Specifically focusing on vorticity, \citet{fedkiw2001visual} and \citet{zhang2015restoring} suggested grid-based approaches to conserve vorticity.
Subsequently, the advection-reflection method \cite{zehnder2018advection, narain2019second} was introduced as a simple yet effective modification to the original advection-projection scheme, resulting in significant advancements. More recently, an advection scheme based on the Kelvin circulation theorem \cite{nabizadeh2022covector} has shown state-of-the-art results.
\revv{The following methods} can also be viewed as advecting a material element, which includes point elements \cite{Chern2016,yang2021clebsch,Xiong2022Clebsch}, line elements \cite{Xiong2022Vortex}, and surface elements \cite{feng2022impulse}, rather than focusing solely on the components of the velocity field.
%
%The flow map technique emerges as an alternate advection scheme, effectively preserving energy and vorticity demonstrated by \citet{deng2023fluid}. 
%Our method follows this track, and we shall further review this method in the next paragraph.

\paragraph{Flow Map Methods}
% before gpt
% \duowen{Flow map originally regarded as the method of characteristic mapping (MCM) was first introduced by \citet{wiggert1976numerical}. Such method reduces numerial dissipation by utilizing a long range mapping to track fluid quantities which reduces the number of interpolations required. Later, \citet{tessendorf2011characteristic} adopted such idea to the graphics community. Researches \cite{hachisuka2005combined, sato2017long, sato2018spatially, tessendorf2015advection} following this line of work uses virtual particles to track flow map, but suffers from intensive computational cost in regard of memory and time. Bidirectional flow map was proposed by \citet{qu2019efficient} with inspiration from BFECC \cite{kim2006advections} to improve accuracy of the mapping. Subsequently, \citet{nabizadeh2022covector} further combines impulse fluid \cite{cortez1996impulse} with flow map. In NFM \cite{deng2023fluid}, authors proposed the concept of \textit{perfect flow map} and used a neural network to compress the storage of velocity fields needed to reconstruct the flow map with accuracy. Such methods all make drastic usage of virtual particles to track flow map and uses different ways to reduce the cost and error caused by such choice. We take a different path by directly track flow map on particles but only interpolate them to grid when needed. Our research shows this choice is more natural in terms of tracking flow map and greatly reduces computational cost without losing accuracy.}
The concept of a flow map, initially known as the method of characteristic mapping (MCM), was first introduced by \citet{wiggert1976numerical}. This method diminishes numerical dissipation with a long-range mapping to track fluid quantities, thereby reducing the frequency of interpolations required. This idea was later adapted for the graphics community by \citet{tessendorf2011characteristic}. Subsequent research \cite{hachisuka2005combined, sato2017long, sato2018spatially, tessendorf2015advection} in this area utilized virtual particles to track the flow map yet faced substantial computational demands in terms of memory and time.
A bidirectional flow map approach was developed by \citet{qu2019efficient}, drawing inspiration from BFECC \cite{kim2006advections}, to enhance the mapping's accuracy. Building on this, \citet{nabizadeh2022covector} combined the impulse fluid model \cite{cortez1996impulse} with the flow map concept. In NFM \cite{deng2023fluid}, the authors introduced the concept of a \textit{perfect flow map} and employed a neural network to efficiently compress the storage of velocity fields required for reconstructing the flow map with high precision. These methods extensively use virtual particles to trace the flow map and employ various techniques to minimize the costs and errors associated with this approach.
A community focusing on flow data visualization studies particle trajectory as a visualization tool instead. One prominent example is the Lagrangian coherent structures (LCS) \cite{haller2000lagrangian, sun2016detection, macmillan2021most} and, specifically, finite time Lyapunov exponent (FTLE) \cite{leung2013backward, leung2011eulerian}. Recently, \citet{kommalapati2021machine} applied Deep Learning techniques to achieve super-resolution visualization of LCS. However, none of these data structures have been leveraged to facilitate simulation.
Moreover, utilizing features of Lagrangian elements to help shrink storage and speed up running time has proved its effectiveness in 3D Gaussian Splatting \cite{kerbl20233d} by reducing the training time of NeRF\cite{mildenhall2021nerf}.
%Our research diverges from these methods and taking the idea of using Lagrangian particle trajectories for tracking the flow map and only interpolating them onto a grid when necessary. 
%This approach proves to be more intuitive for flow map tracking and significantly reduces computational costs without compromising accuracy.

\paragraph{Eulerian-Lagrangian Methods} 
% before gpt
% \duowen{
% To take advantage of the fast convergence rate for grid based Possion solver, e.g. MGPCG \cite{mcadams2010parallel}, a hybrid Eulerian-Lagrangian representation is needed. Such choice is commonly made in graphics community due to the significant reduction of viscosity with Lagrangian representation. Since the seminal work of PIC \cite{harlow1962particle} / FLIP \cite{brackbill1986flip} was introduced to graphics community \cite{zhu2005animating}, Eulearian-Lagrangian hybrid representation is widely used for simulating fluids \cite{gao2009simulating, hong2008adaptive, lee2009interchangeable, zhu2010creating, raveendran2011hybrid, ando2011particle, deng2022moving}. Later, MPM, considered as a generalization for PIC/FLIP, was used to simulate various fluid behaviors including snow \cite{stomakhin2013material}, phase change \cite{stomakhin2014augmented}, foam \cite{yue2015continuum, ram2015material}, magnetized flow \cite{sun2021material}, fluid-structure interaction \cite{fei2017multi, fei2018multi, yan2018mpm, han2019hybrid, fang2020iq} and sediment flow \cite{tampubolon2017multi, gao2018animating}. Subsequent research further improves accuracy for transferring between Lagrangian and Eulerian representations with momentum conservation \cite{jiang2015affine, fu2017polynomial}, discontinuous velocity \cite{hu2018moving} and volume conservation \cite{qu2022power}. With the merit of previous exploration, our method further explores the usage of using hybrid representation in the domain of flow map advection scheme. }
To leverage the rapid convergence rate of grid-based Poisson solvers, such as MGPCG \cite{mcadams2010parallel}, a hybrid Eulerian-Lagrangian representation is needed. Such a choice is commonly made in the graphics community due to the significant viscosity reduction it offers with Lagrangian representation. Since the seminal work of PIC \cite{harlow1962particle} and FLIP \cite{brackbill1986flip} was introduced to graphics community \cite{zhu2005animating}, hybrid Eulerian-Lagrangian representations are widely used in fluid simulation \cite{gao2009simulating, hong2008adaptive, zhu2010creating, raveendran2011hybrid, ando2011particle, deng2022moving}. MPM, which can be seen as \revv{a} generalization of PIC/FLIP, has been utilized to simulate a variety of \revv{continuum} behaviors including snow \cite{stomakhin2013material}, phase changes \cite{stomakhin2014augmented}, foam \cite{yue2015continuum, ram2015material}, magnetized flow \cite{sun2021material}, fluid-structure interactions \cite{fei2017multi, fei2018multi, yan2018mpm, han2019hybrid, fang2020iq}, and sediment flow \cite{tampubolon2017multi, gao2018animating}. Further research has enhanced the accuracy of transferring between Lagrangian and Eulerian representations, focusing on momentum conservation \cite{jiang2015affine, fu2017polynomial}, handling discontinuous velocities \cite{hu2018moving}, \revv{energy dissipation \cite{fei2021revisiting}} and maintaining volume conservation \cite{qu2022power}. 
%Building upon these advancements, our method further explores the use of hybrid representation for flow map advection schemes.

\paragraph{Impulse Fluid and Gauge Methods}
% before gpt
% \duowen{
% As flow map is particularly useful while considering \textit{impulse variable}, we briefly mention this topic below. \textit{Impulse variable}, first introduced in \cite{buttke1992lagrangian}, leads to an alternative way of writing incompressible Navier-Stoke Equation by allowing a gauge variable and a gauge transformation \cite{oseledets1989new, roberts1972hamiltonian, buttke1993velicity}. Different gauge freedoms are studied for purpose of surface turbulance\cite{buttke1993velicity, buttke1993turbulence}, numerical stability\cite{weinan2003gauge} and fluid-structure interaction \cite{cortez1996impulse, summers2000representation}. More recent work from \citet{saye2016interfacial, saye2017implicit} use gauge freedom to solve interfacial discontinuity between density and viscousity for free surface flow and fluid-structure coupling. The concept of \textit{gauge freedom} were revisited in computer graphics by \citet{feng2022impulse}, \citet{yang2021clebsch} and \citet{nabizadeh2022covector}. However, such methods suffered from impedance of the advection schemes they use. NFM \cite{deng2023fluid} solves such issue by presenting a state-of-art neural hybrid advection scheme using flow map, but is limited by neural buffer storage and training time. In contrast, our method takes a more natural approach for representing flow map using Lagrangian particles and greatly reduces the cost induced by NFM but maintaining the state-of-art quality.}

Flow maps are particularly relevant when considering the concept of the \textit{impulse variable}. This term was first introduced in \cite{buttke1992lagrangian} and offers an alternative formulation of the incompressible Navier-Stokes Equation through the use of a gauge variable and gauge transformation \cite{oseledets1989new, roberts1972hamiltonian, buttke1993velicity}. Various aspects of gauge freedom have been explored, including its application to surface turbulence \cite{buttke1993velicity, buttke1993turbulence}, numerical stability \cite{weinan2003gauge}, and fluid-structure interactions \cite{cortez1996impulse, summers2000representation}. More recent research by \citet{saye2016interfacial, saye2017implicit} has utilized gauge freedom to address interfacial discontinuities in density and viscosity in free surface flows and fluid-structure coupling.
The concept of \textit{gauge freedom} was revisited in the field of computer graphics by \citet{feng2022impulse}, \citet{yang2021clebsch}, and \citet{nabizadeh2022covector}. However, these methods encountered limitations due to the advection schemes' constraints. NFM \cite{deng2023fluid} addresses this issue by introducing a state-of-the-art neural hybrid advection scheme using flow maps. Despite its advancements, NFM is constrained by the requirements for neural buffer storage and extended training time.
%
%In contrast, our method adopts a more intuitive approach by representing flow maps using Lagrangian particles. This significantly reduces the computational costs associated with NFM, while maintaining its high-quality results.
% \textcolor{red}{
% \paragraph{3D Gaussian Splatting to NeRF is ours to NFM} Utilizing features of Lagrangian elements to help shrinking storage and speeding up running time has proved its effectiveness in 3D Gaussian Splatting \cite{kerbl20233d}. 
% Since, the seminal work of NeRF\cite{mildenhall2021nerf}, many of its variations presents. However, only recently 3D Gaussian Splatting \cite{kerbl20233d} was able to remove the neural component of NeRF and use only 3D Gaussian splats to achieve state-of-art result. NFM \cite{deng2023fluid} using neural buffer for compression gives the unpresented ability of maintaining vorticity and structure of vortices. However it also suffers from the common issue similar to all neural method --- long training time and high GPU memory usage. By using Lagrangian elements, our method successfully reduced the cost of NFM and preserved its state-of-art result.
% }

\section{Physical Model}
\paragraph{Naming convention} 
We adopt the following notation conventions: lowercase letters for scalars (e.g., $t, n, w$), bold letters for vectors (e.g., $\bm X, \bm x, \bm \phi, \bm \psi$), and calligraphic font for matrices (e.g., $\mathcal{F}, \mathcal{T}$). Subscripts are used in two contexts: in a continuous setting, subscripts (typically \revv{$a, b, c$}) denote specific time instants; for example, \revv{$\mathcal{T}_b$} represents the backward map Jacobian at time instant \revv{$b$}. \junweirev{In a discrete setting, subscripts (typically \revv{$i, j, k, p$}) indicate primitive indices, such as \revv{$\bm x_p$} signifying the position of the \revv{$p$}-th particle.} The notation $[:,:]$ is used to specify a time period over which a map holds, for instance, \revv{$\mathcal{T}_{[a,c]}$} denotes the backward map Jacobian from time \revv{$c$} to time \revv{$a$}. \junweirev{However, mathematically, if \revv{$a < c$}, the backward map Jacobian from time \revv{$c$} to time \revv{$a$} should be denoted as \revv{$\mathcal{T}_{[c, a]}$}. For simplicity, in this paper, we adopt a convention where time in subscripts progresses from smaller to larger values, thus \revv{$\mathcal{T}_{[a, c]}$} is used to denote the backward map Jacobian from time \revv{$c$} to time \revv{$a$}. Similarly, the backward map from time \revv{$c$} to time \revv{$a$} is represented as \revv{$\bm \psi_{[a, c]}$}. In addition, when we refer to the time period $[0, t]$, we will simplify it as $t$. For instance, $\mathcal{T}_{[0, t]} = \mathcal{T}_t$.} We have summarized these important notations in Table~\ref{tab: notation_table}. We refer to Table~\ref{tab:grid_particle_quantity} for variables specific to discrete settings. 

\newcolumntype{z}{X}
\newcolumntype{s}{>{\hsize=.25\hsize}X}
\begin{table}
\caption{Summary of important notations used in the paper.}
\centering\small
\begin{tabularx}{0.47\textwidth}{scz}
\hlineB{3}
Notation & Type & Definition\\
\hlineB{2.5}
\hspace{12pt}$\bm X$ & vector & material point position at initial state\\
\hlineB{1}
\hspace{12pt}$\bm x$ & vector & material point position at terminated state\\
\hlineB{1}
\hspace{12pt}$t$ & scalar & time\\
\hlineB{1}
\hspace{12pt}$\bm \phi$ & vector & forward map\\
\hlineB{1}
\hspace{12pt}$\bm \psi$ & vector & backward map\\
% \hlineB{1}
% \hspace{1pt}$\phi(X,t)$ & vector & position induced from $X$ by $\phi$ at time $t$ \\
% \hlineB{1}
% \hspace{2.5pt}$\psi(x,t)$ & vector & position backtraced from $x$ by $\psi$ at time $t$\\
\hlineB{1}
\hspace{12pt}$\mathcal{F}$ & matrix & forward map Jacobian\\
\hlineB{1}
\hspace{12pt}$\mathcal{T}$ & matrix & backward map Jacobian\\
\hlineB{1}
\hspace{12pt}$\bm \nabla\mathcal{T}$ & matrix & backward map Hessian\\
\hlineB{1}
\hspace{12pt}$\bm{u}$ & vector & velocity\\
\hlineB{1}
\hspace{12pt}$\bm{m}$ & vector & impulse\\
\hlineB{1}
\hspace{12pt}$\bm \nabla \bm{m}$ & vector & impulse gradients\\
% \hlineB{1}
% \hspace{12pt}$\bm \nabla\varphi$ & vector & an arbitrary gradient between $\bm{u}$ and $\bm{m}$\\
\hlineB{1}
\hspace{12pt}$w$ & scalar & interpolation weights \\
\hlineB{1}
\hspace{12pt}$\bm \nabla w$ & \revv{vector} & interpolation weights gradients \\
\hlineB{1}
\hspace{12pt}$n$ & \revv{integer} & reinitialization steps\\
% \hlineB{1}
% \hspace{12pt}$N(p)$ & / & set of particles in the neighborhood of particle $p$ \\
\hlineB{3}
% \hspace{12}
\end{tabularx}
\vspace{5pt}
% \caption{Summary of important notations used in the paper.}
\label{tab: notation_table}
\end{table}

\subsection{Mathematical Foundation}
% Additionally, we use the superscripts $L$ and $S$ for quantities associated with the long-range and short-range maps, respectively. For initial values, we add a subscript $0$. And superscript $T$ denotes transpose. For instance, $\bm m^L_{p0}$ represents initial impulse on particles of the long-range map, while $\mathcal{T}^{S,T}_p$ signifies the transpose of backward map Jacobian on particles of the short-range map.

\paragraph{Flow Map Preliminaries}
We define a velocity field $\bm u(\bm x,t)$ in the fluid domain $\Omega$ which specifies the velocity at a given location $\bm x$ and time $t$. Consider a material point $\bm X \in \Omega$ at time $t=0$, We define the forward flow map $\bm \phi(:,t):\Omega\rightarrow\Omega$ as
\begin{equation}
    \label{eq:phi_def}
    \begin{dcases}
    \frac{\partial \bm \phi(\bm X,t)}{\partial t}=\bm u[\bm \phi(\bm X,t),t],\\
    \bm \phi(\bm X,0) = \bm X, \\
   \bm  \phi(\bm X,t) = \bm x,
    \end{dcases}
\end{equation}
which traces the trajectory of the point, moving from its initial position $\bm X$ at time $0$ to its location at time $t$, represented by $\bm x$. Its inverse mapping $\bm \psi(:,t):\Omega\rightarrow\Omega$ is defined as
\begin{equation}
    \label{eq:psi_def}
    \begin{dcases}
    \bm \psi(\bm x,0) = \bm x, \\
   \bm  \psi(\bm x,t) = \bm X.
    \end{dcases}
\end{equation}
which maps $\bm x$ at $t$ to $\bm X$ at $0$.

% % \begin{equation}
% % \bm \psi[\bm \phi(\bm X, t),t ] =  \bm X, 
% % \end{equation}
% \begin{equation}
%     \begin{dcases}
%     \frac{\partial \bm \psi(\bm x,t)}{\partial t}=\bm u(\psi(\bm x,t),t),\\
%     \bm \psi(\bm x,t) = \bm x, \\
%     \bm \psi(\bm x,0) = \bm X.
%     \end{dcases}
%     \label{eq:flow_map}
% \end{equation}
% \begin{equation}
%     \label{eq:phi_def}
%     \begin{dcases}
%     \bm \psi(\bm x,0) = \bm x, \\
%    \bm  \psi(\bm x,t) = \bm X.
%     \end{dcases}
% \end{equation}

In order to characterize infinitesimal changes in the flow map and its inverse mapping, we compute their Jacobian matrices as:
\begin{equation}
    \begin{dcases}
        \mathcal{F}(\bm \phi, t)= \frac{\partial \bm \phi (\bm x, t)}{\partial \bm x}, \\
        \mathcal{T}(\bm x, t) = \frac{\partial \bm \psi (\bm x, t)}{\partial \bm x}.
    \end{dcases}
    \label{eq:FT}
\end{equation}
As proven in Appendix \ref{sec:Jacobian_prov}, the evolution of $\mathcal{F}$ and $\mathcal{T}$ satisfies the following equations:
\begin{equation}
\begin{dcases}
\frac{D \mathcal{F}}{D t}=\bm \nabla\bm{u}\mathcal{F},\\
\frac{D \mathcal{T}}{D t}=-\mathcal{T}\bm \nabla\bm{u}. \label{eq:T_mapping}
\end{dcases}
\end{equation}

We further define the backward map Hessian as:
\begin{equation}
    \bm \nabla \mathcal{T} = \frac{\partial \mathcal{T}}{\partial \bm x}.
\end{equation}

\paragraph{Perfect Flow Map}
A perfect flow map \cite{deng2023fluid} refers to a \revv{bidirectional} map satisfying the following two \revv{remarks}:
\begin{remark}
\label{remark:roundtrip_psi_phi}
After undergoing a backward map, denoted as $\bm \psi_t \equiv \bm \psi(:,t)$, followed by a forward map $\bm \phi_t \equiv \bm \phi(:,t)$, a point is anticipated to return to its original position. This principle holds true when the sequence is reversed. In essence,
\begin{equation}
    \begin{dcases}
        \bm X =\bm \psi_t\circ\bm \phi_t(\bm X), \\
        \bm x= \bm \phi_t \circ\bm \psi_t (\bm x).
    \end{dcases}
    \label{eq:roundtrip_psi}
\end{equation}

\revvdel{
\begin{proof}
Derived from the expressions in \eqref{eq:phi_def} and \eqref{eq:psi_def}, we obtain the following equations
\begin{equation}
\begin{dcases}
\bm \psi_t\circ\bm \phi_t(\bm X)=
\bm \psi[\bm \phi(\bm X, t),t]=\bm \psi[\bm \phi(\bm X, t),t]=\bm \psi(\bm x,t) = \bm X\\
\bm \phi_t \circ\bm \psi_t (\bm x)=\bm \phi[\bm \psi(\bm x, t),t]=\bm \phi[\bm \psi(\bm x, t),t]=\bm \phi(\bm X,t) = \bm x
\end{dcases}
\end{equation}
\end{proof}
}

\end{remark}

\begin{remark}
\label{remark:roundtrip_FT}
A coordinate frame that undergoes deformation specified initially by a forward map Jacobian $\mathcal{F}_t(\bm X)\equiv\mathcal{F}(\bm X,t)$ and then by its backward map Jacobian $\mathcal{T}_t(\bm x)\equiv\mathcal{T}(\bm x,t)$ should remain unchanged. The same principle holds when the order of application is reversed. In matrix notation,
\begin{equation}
    \begin{dcases}
        \bm I = \mathcal{F}_t (\bm X)\mathcal{T}_t(\bm x),\\
        \bm I = \mathcal{T}_t(\bm x)\mathcal{F}_t(\bm X).
    \end{dcases}
    \label{eq:roundtrip_T}
\end{equation}

\revvdel{
\begin{proof}
By taking the gradients of the first and second lines of \eqref{eq:roundtrip_psi} with respect to $\bm X$ and $\bm x$ respectively, and employing the chain rule, we obtain
\begin{equation}
\begin{dcases}
\bm I = \frac{\partial \bm \psi(\bm \phi, t) }{\partial \bm \phi} \frac{\partial \bm \phi(\bm X, t) }{\partial \bm X} \\
\bm I = \frac{\partial \bm \phi(\bm \psi, t) }{\partial \bm \psi} \frac{\partial \bm \psi(\bm x, t) }{\partial \bm x} \\
\end{dcases}
\label{eq:IIXx}
\end{equation}
Substituting equations \eqref{eq:phi_def}, \eqref{eq:psi_def}, and \eqref{eq:FT} into the expressions for \eqref{eq:IIXx} results in the equation \eqref{eq:roundtrip_T}.
\end{proof}
}

\end{remark}

\subsection{Impulse Fluid on Flow Maps}
The Euler equations can be written in the form of the impulse as follows
\begin{equation}
\begin{dcases}
 \frac{D \bm m }{D t} =-\left(\bm \nabla \bm u\right)^T \bm m,\\
\nabla^2 \varphi = \bm \nabla \cdot \bm m,\\
\bm u = \bm m - \bm \nabla \varphi,
\end{dcases}
\label{eq:impulse_euler_eq}
\end{equation}
Here, \revv{$\varphi$} serves as a gauge variable employed to depict the gradient field distinction between the impulse $\bm m$ and the divergence-free velocity field $\bm u$.

As shown in \citet{cortez1996impulse,nabizadeh2022covector,deng2023fluid}, the evolution of impulse $\bm m$ can be written as a flow map from time $0$ to time $t$ as:
\begin{equation}
    \label{eq:evolve_imp}
    \bm m(\bm x,t)=\mathcal{T}^{T}_t(\bm x)\,\bm{m}(\bm \psi(\bm x),0),
\end{equation}
% where we use $\mathcal{T}_t(\bm x)$ to denote $\mathcal{T}(\bm x,t)$. 
%To solve Equation~\eqref{eq:impulse_euler_eq}, we can first compute backward map Jacobian $\mathcal{T}_t$, and then utilize $\mathcal{T}_t$ to evolve impulse $\bm m$ from time $0$ to $t$. 
Similarly, we describe the evolution of the impulse gradient, $\bm \nabla \bm m$, using the flow map as follows:
\begin{equation}
    \label{eq:evolve_grad_imp}
    \bm \nabla \bm m(\bm x,t) = \mathcal{T}^{T}_t\,\bm \nabla_{\bm \psi} \bm{m}(\bm \psi(\bm x),0)\,\mathcal{T}_t + \bm \nabla \mathcal{T}^{T}_t\,\bm{m}(\bm \psi(\bm x),0),
\end{equation}
% where $\bm \nabla \mathcal{T}_t$ is defined as
where $\bm \nabla_{\bm \psi}$ denotes the gradient with respect to the variable $\bm \psi$. For a mathematical proof for Equations~\ref{eq:evolve_imp} and \ref{eq:evolve_grad_imp}, we direct readers to the Appendix \ref{sec:mdm}.

Consequently, by utilizing the backward map Jacobian $\mathcal{T}_t$ and the backward map Hessian $\bm \nabla \mathcal{T}_t$, we can obtain the impulse $\bm m(\bm x,t)$ and its gradients $\bm \nabla \bm m(\bm x,t)$ at time $t$. This is achieved by evolving both the initial impulse $\bm{m}(\bm \psi(\bm x),0)$ and the initial impulse gradients $\bm \nabla_{\bm \psi} \bm{m}(\bm \psi(\bm x),0)$ from time $0$ to $t$.

\section{Particle Flow Map}
\subsection{Particle Trajectory is a Perfect Flow Map}
\begin{comment}
\begin{figure}[t]
  \centering
  \includegraphics[inkscapelatex=false, width = .15\textwidth]{example-image}
  \caption{Illustration of a particle flow map.}
  \label{fig:pfm}
\end{figure}
\end{comment}

%\paragraph{Definition}

\WrapFig{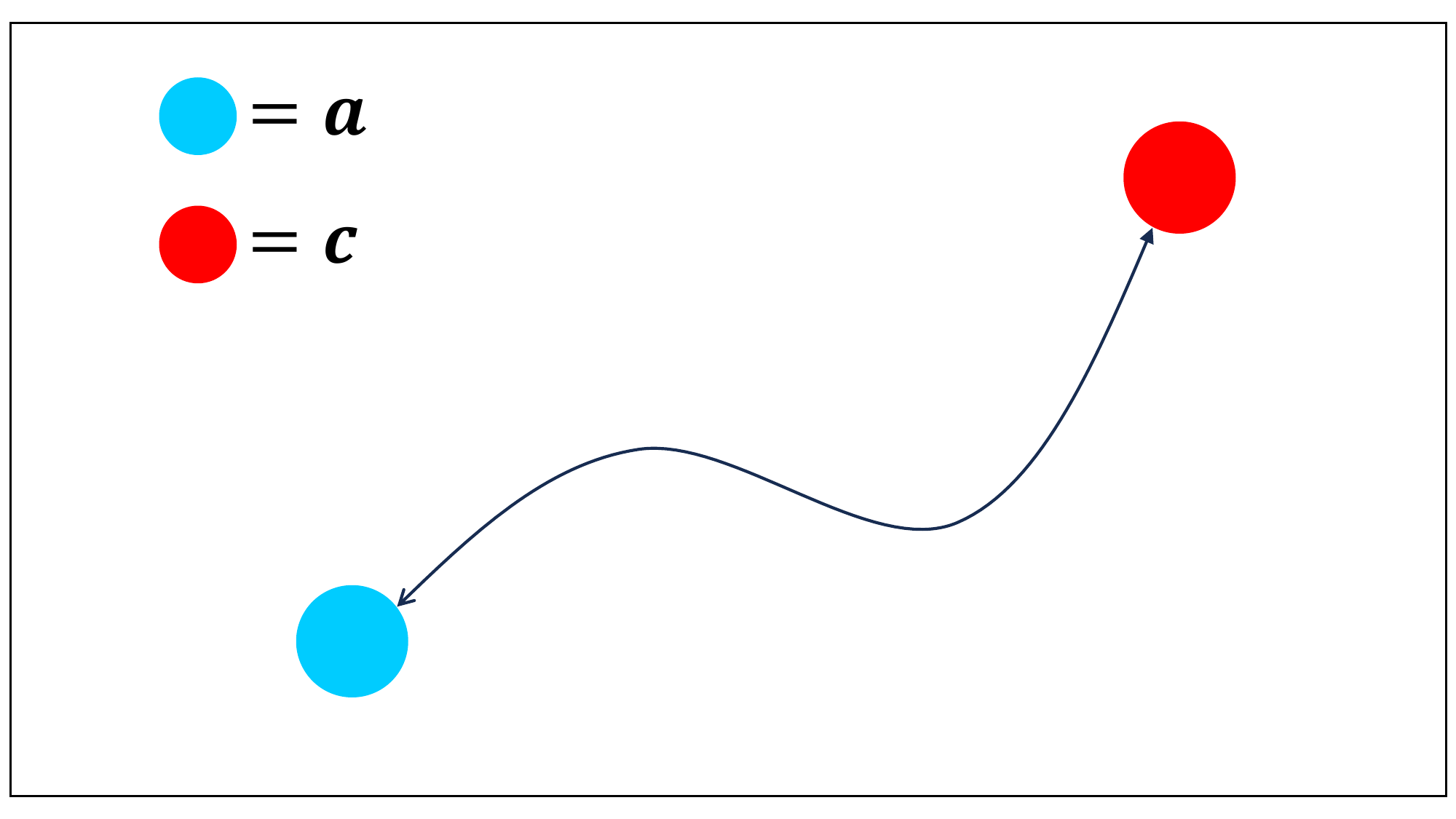}{Particle is a flow map.}{fig:pfm}{.275}{1.15}
Particles inherently embody flow maps. As depicted in Figure~\ref{fig:pfm}, consider a particle that flows according to a spatiotemporal velocity field $\bm u(\bm x,t)$. \junweirev{Its trajectory, originating from time \revv{$a$} and culminating at time \revv{$c$}, characterizes both the forward map \revv{$\bm \phi_{[a,c]}$} and the backward map \revv{$\bm \psi_{[a,c]}$}.} 
% \bo{check the symbol} 
Quantities like $\mathcal{F}$ and $\mathcal{T}$, which are determined through path integrals along a particle's trajectory (as exemplified by the path integrals of Equation~\ref{eq:T_mapping} for obtaining $\mathcal{F}$ and $\mathcal{T}$), can be calculated by conducting integrals while tracking a moving particle on the trajectory.

Next, we demonstrate that a particle flow map constitutes a perfect flow map. To establish this, we will show that a particle flow map adheres to Remarks~\ref{remark:roundtrip_psi_phi} and ~\ref{remark:roundtrip_FT}. 
% Subsequently, we illustrate that we can 
% how these properties can be conveniently acquired using a time-reversible numerical integrator. Finally, we present two illustrative numerical experiments to further elucidate this concept.

% \paragraph{A Mathematical Proof}
% We first show a particle flow map is a perfect flow map by checking the two remarks given by Equations~\ref{eq:roundtrip_psi} and ~\ref{eq:roundtrip_T}. 
% remark 1
First, consider a particle \revv{$p$} that moves according to a flow map \revv{$\bm \phi(\bm X,t)$}, starting from point $\bm X$ at time $0$ and ending at point $\bm x$ at time $t$. If we reverse this process, allowing the particle to move from $\bm x$ at time $t$ using \revv{the} backward flow map $\bm \psi(\bm x,t)$, the particle will follow precisely the same trajectory in reverse order. This alignment of the start and end points ensures that a particle flow map inherently satisfies Remark~\ref{remark:roundtrip_psi_phi}.
% remark 2

Next, we show a particle flow map also satisfies Remark~\ref{remark:roundtrip_FT}.  
We carried out a numerical experiment to illustrate this: As shown in Figure~\ref{fig:motivational_experiment}(a) and (b), we compute $\mathcal{F}\mathcal{T}$ on particles moving in a steady velocity field defined as $\omega(r) = -0.01(1 - e^{-r^2/0.02^2})/r$  and $u(x, y, r) = \omega(r)\begin{bmatrix}-y\\ x\end{bmatrix}$ where $x$ and $y$ are point locations and $r$ is the distance from point to rotation center at $(0.5, 0.5)$. As in Figure~\ref{fig:motivational_experiment}(c), we observed that $\mathcal{F}\mathcal{T}$ closely approximates the identity, with errors on the
order of $10^{-6}$ within 200 steps. 
These findings confirm that flow maps on particles are in alignment with Remark~\ref{remark:roundtrip_FT}.

\begin{figure}[t]
 \centering
 \includegraphics[width=.47\textwidth]{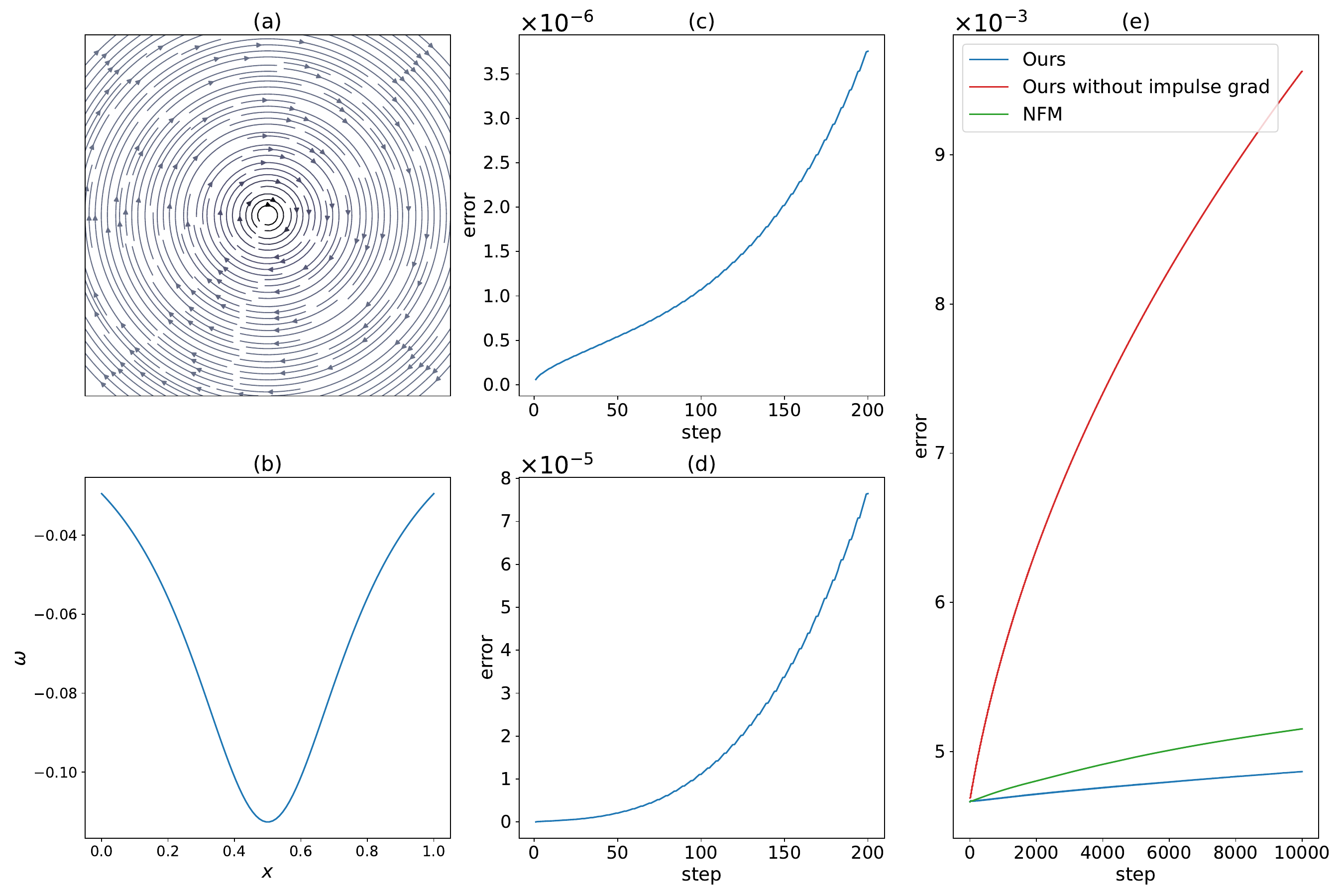}
 \caption{(a) Velocity field streamlines of a single vortex. (b) Angular velocity along $y$ = 0.2 of the velocity field in (a). (c) Deviation of $\mathcal{F}\mathcal{T}$ relative to the identity matrix. (d) Discrepancy between the forward-evolved $\mathcal{T}$ and the backward-evolved $\mathcal{T}$. In both (c) and (d), the errors correspond to the average of the individual errors across all particles, where the error on each particle is quantified by the Frobenius norm of the disparities between the respective matrices of each particle. (e) \revv{Disparity between the analytical velocity field and the velocity field reconstructed using our method. It demonstrates that incorporating impulse gradients in particle-to-grid process plays a crucial role in reducing errors.}}
 \label{fig:motivational_experiment}
\end{figure}

% \label{sec:pfm}
% \begin{figure}[t]
%   \centering
%   \includesvg[inkscapelatex=false, width = .47\textwidth]{PFM}
%   \caption{Short-range and long-range maps. \bo{Why do we need to specify these six cases?}}
%   \label{fig:PFM}
% \end{figure}

% \begin{comment}
% \subsection{Forward Evolution of backward map Jacobian $\mathcal{T}$}
% \subsection{Flow Maps on Particles}
\subsection{Forward Evolution of $\mathcal{T}$}
% In NFM \cite{deng2023fluid}, both $psi$ and $\mathcal{T}$ are stored on grid. To compute their values, a method involving "virtual particles" is employed. Specifically, the trajectory of a "virtual particle" terminates precisely at a grid node. To obtain the origin of the trajectory, which corresponds to $\psi$ at the grid node, and to propagate $\mathcal{T}$ along this trajectory, it is necessary to trace the path backward from the grid node to the initial position of the "virtual particle". This backward tracking necessitates the preservation of velocity field data from prior steps. Consequently, NFM incorporates a neural buffer to archive these velocity fields. However, the computational overhead associated with training this buffer during each step, coupled with the augmented memory demands for buffer storage, imposes constraints on the practicality of applying NFM.

In NFM \cite{deng2023fluid}, the trajectories of the forward and backward maps do not coincide. Consequently, at each step, it's essential to backtrack the path of a "virtual particle" and correspondingly backward evolve $\mathcal{T}$ along this path. Practically, this procedure employs the upper part of Equation~\ref{eq:T_mapping}, but with the time dimension inverted. The requirement to backward evolve $\mathcal{T}$ from the current step back to the initial step \revv{necessitates} the usage of a neural buffer for recording the velocity field at each timestep, imposing significant computational and storage costs.

% $\mathcal{T}$ is stored on grid. To compute its value, we need to backtrack the trajectory of a "virtual particle" whose end point aligns with the grid node. This process actually employs the upper equation of Equation~\ref{eq:T_mapping} with time reversed to backward evolve $\mathcal{T}$ from end point to initial point of "virtual particle". Owing to the fact that the trace of the backward map doesn't align with the forward map. this process necessitates a neural buffer to archive the velocity field at each step, incurring significant computational and storage demands.

Given that both the forward and backward maps in the particle flow map align with the same particle trajectory, there's no necessity to backward evolve $\mathcal{T}$ as is done in NFM.  Instead, we opt for a direct forward evolution of $\mathcal{T}$, as demonstrated in the bottom part of Equation~\ref{eq:T_mapping}.
% \end{comment}

We further conducted an experiment to demonstrate that the forward evolution of $\mathcal{T}$ on particles yields the same results as the backward evolution. We establish a steady velocity field, as depicted in Figure~\ref{fig:motivational_experiment}(a) and (b). Particles are placed and move within this velocity field, and both backward and forward evolution of $\mathcal{T}$ are conducted on each particle. 
% The first evolution mirrors the backward evolution of $\mathcal{T}$ in NFM, tracing from the particle's present position back to its initial position at each step. Conversely, the second evolution advances one step forward at each stage, moving from the particle's previous position to its current one. 
As illustrated in Figure~\ref{fig:motivational_experiment}(d), these two evolutions of $\mathcal{T}$ are nearly identical, with discrepancies on the order of $10^{-5}$ within 200 steps, which aligns with the theoretical proof that these two methods of evolution are equivalent, shown in Appendix~\ref{sec:proof_T_forward_backward}.

% In response to these limitations of NFM, we introduce a methodology that represents flow maps using particles, along with the forward evolution of the backward map Jacobian $\mathcal{T}$. This particle-based representation enables the direct storage of $\psi$ and $\mathcal{T}$ on particles which are not constrained to the position of grid points. This flexibility enables the continuous evolution of $\mathcal{T}$ concurrent with particle movement using Alg.~\ref{alg:RK4}, eliminating the necessity to backtrack the trajectory of a "virtual particle" at each step, thereby obviating the requirement for training and maintaining a neural buffer to archive velocity field data from previous steps.

% \setlength{\abovecaptionskip}{12pt}
% \begin{figure*}[t]
%  \centering
%  \includegraphics[width=.99\textwidth]{img/examples/propeller_vorticity.pdf}
%  \caption{The series of images effectively illustrates the vortex formation process initiated by a propeller as it rotates under wind influence. These visuals provide a clear view of the intricate, spiral-shaped patterns formed by the interconnected vortex tubes. Additionally, they highlight the dynamic alterations in the trailing wake that occur as a direct result of the propeller's movements. }
%  \label{fig:propeller}
% \end{figure*}

\begin{figure*}[h]
\centering
\begin{minipage}{.49\linewidth}
  \includegraphics[width=\linewidth]{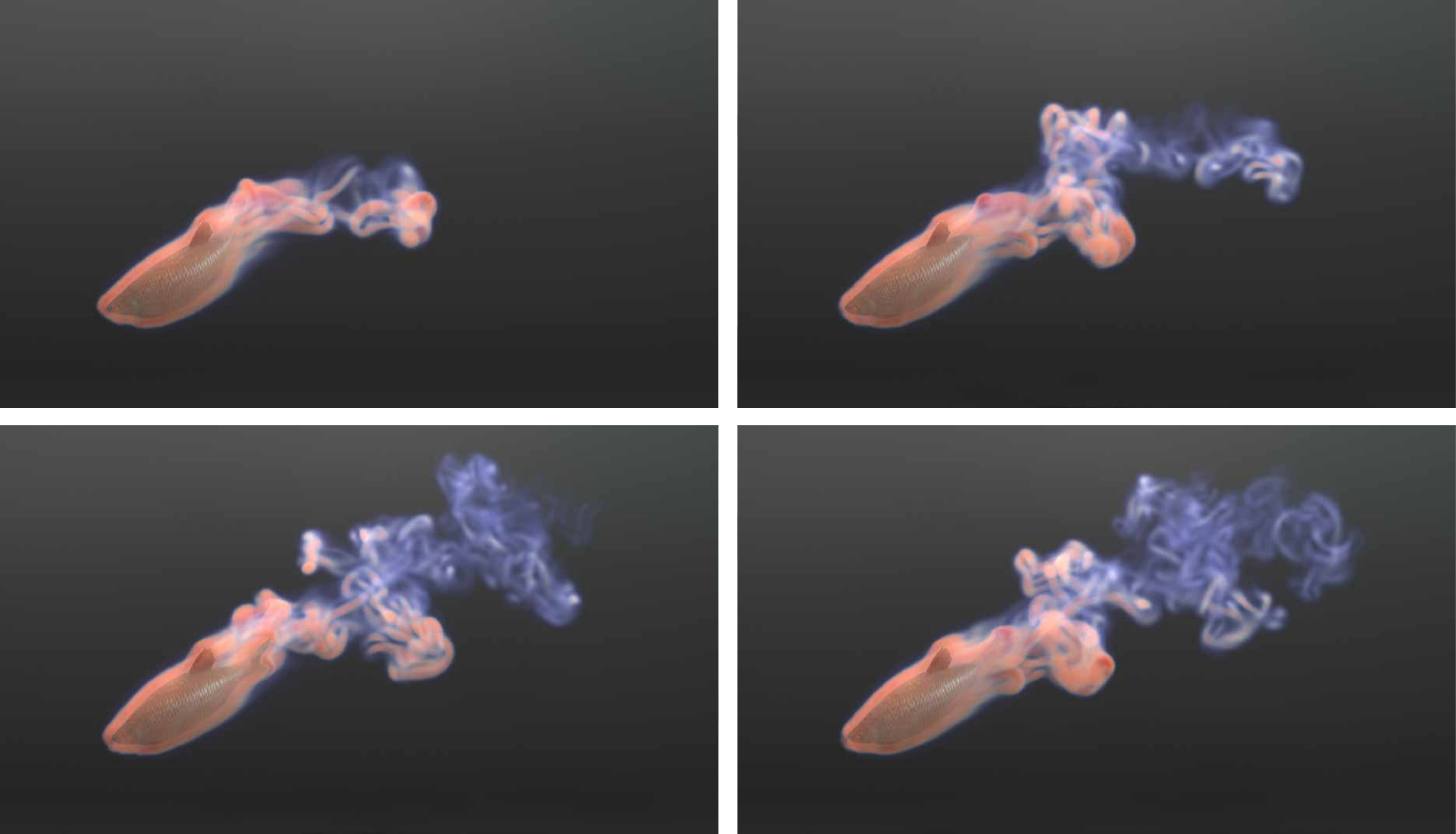}
 \caption{These images depict the vortices created by a fish's tail as it swims through water. The fish's periodic tail movements generate cyclical vortices, and the nesting of these vortex tubes creates a layered and complex structure.}
 \label{fig:fish}
\end{minipage}
\hspace{.01\linewidth}
\begin{minipage}{.49\linewidth}
  \includegraphics[width=\linewidth]{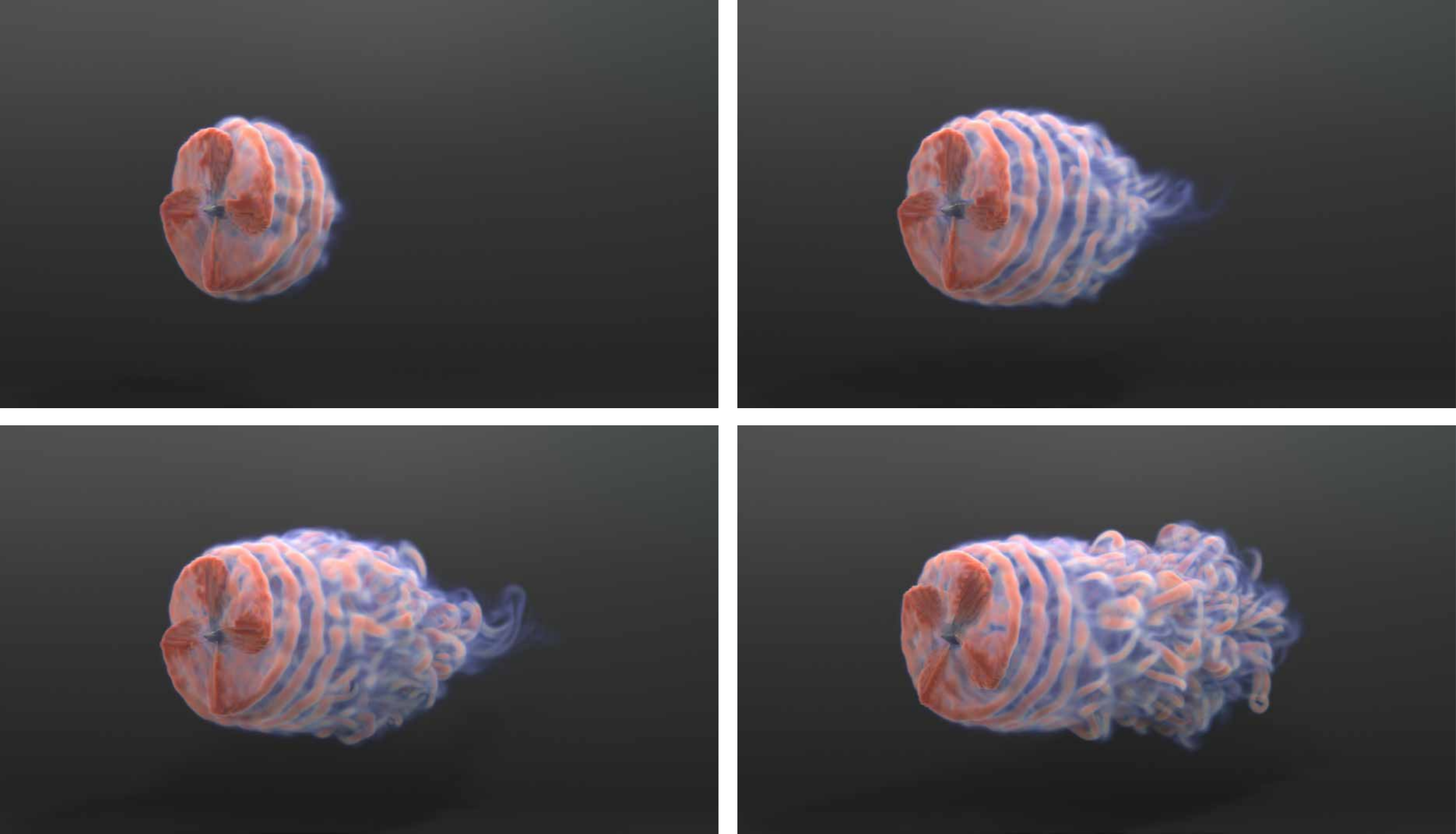}
 \caption{These images shows vortex formation process initiated by a propeller as it rotates under wind influence. These visuals provide a clear view of the intricate, spiral-shaped patterns formed by the interconnected vortex tubes.}
 \label{fig:propeller}
\end{minipage}
\end{figure*}

\subsection{Enabling Adaptivity on a Particle Flow Map}
\label{sec:short_long_map}
In real simulation scenarios, it is often necessary to use flow maps of varying lengths to transport different flow quantities \revv{as mentioned in \cite{sato2018spatially}}. For example, a long-range flow map can be used to transport the fluid impulse, while a short-range flow map may be necessary for its gradients. Generally, a quantity that is more sensitive to distortions in the background flow field, such as a high-order tensor, benefits from a shorter map. Experimenting with the length scales of these flow maps can optimize their transport effectiveness in the simulation. These needs necessitate the design of an adaptivity mechanism within our particle flow map representation.

% \begin{figure}[t]
%  \centering
%  \includegraphics[width=.47\textwidth]{img/illustration/adaptive_PFM_abc.pdf}
%  \caption{\revv{Adaptive flow map from time $a$ to time $c$ with temporal samples $[b_{n}, ..., b_2, b_1]$. Impulse is stored on each time samples and backward map Jacobian is stored between adjacent time samples.}
%  % \bo{I feel this figure is not very informative overall. We may consider delete it.}
%  }
%  \label{fig:adaptive_PFM}
% \end{figure}

\WrapFig{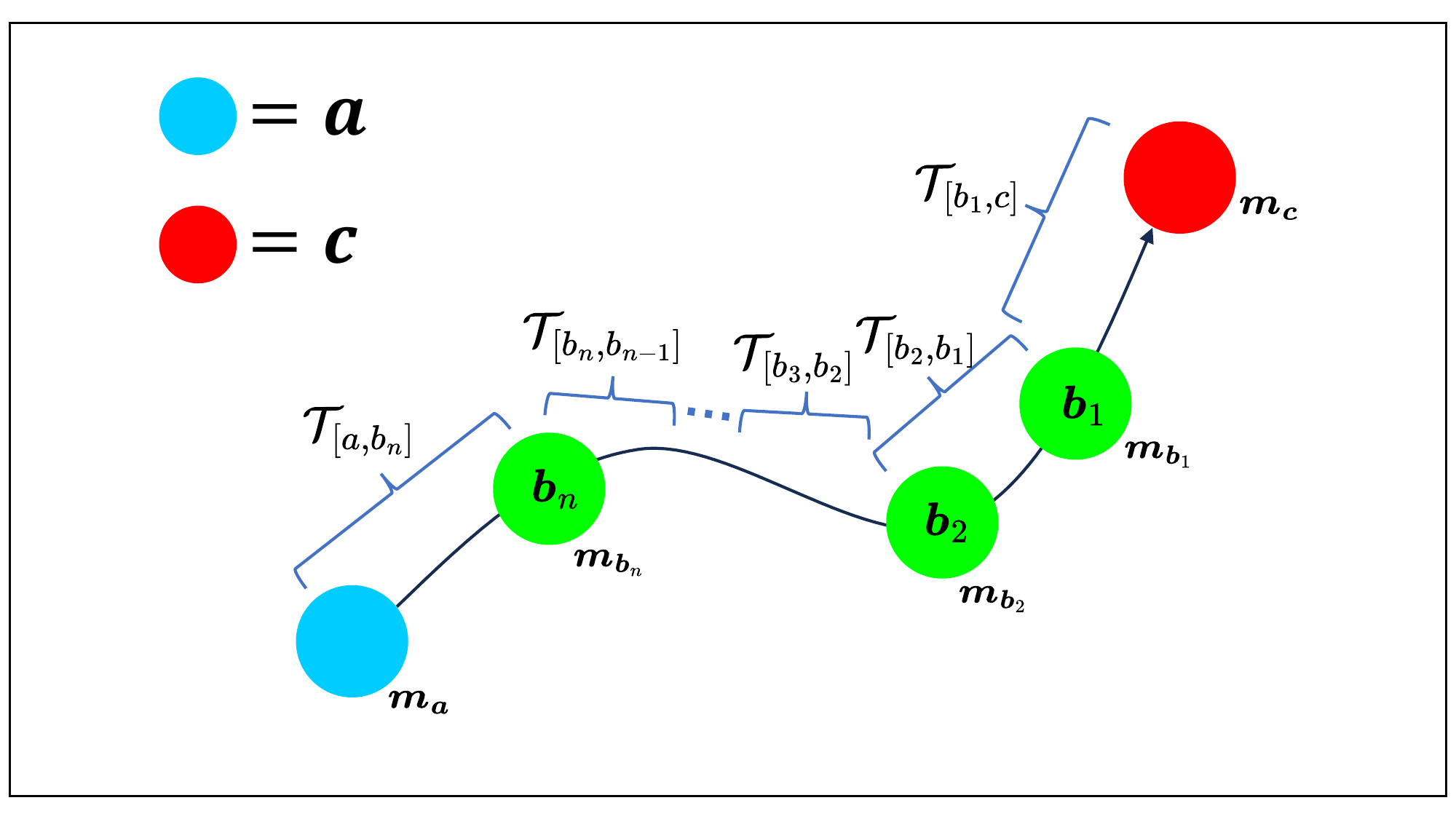}{Adaptive flow map from time $a$ to time $c$ with temporal samples $[b_{n}, ..., b_2, b_1]$. Impulse is stored on each time sample and backward map Jacobian is stored between adjacent time samples.}{fig:adaptive_PFM}{.275}{1.15}

\subsubsection{Adaptive Flow Map with Temporal Samples}
Fortunately, the adaptivity required can be effectively achieved through a simple idea by leveraging the Lagrangian nature of particle trajectories. As illustrated in Figure~\ref{fig:adaptive_PFM}, we can construct an adaptive particle flow map by storing samples at different time instants \revv{$[a, b_{n}, ..., b_2, b_1, c]$}, where \revv{$a < b_{n} < ... < b_2 < b_1 < c$}, along the particle's trajectory from start time \revv{$a$} to end time \revv{$c$}. We capture snapshots of quantities at each time sample (e.g., \revv{$\bm m_{b_n}, ..., \bm m_{b_2}, \bm m_{b_1}$}), and store the backward map Jacobian between each pair of adjacent time samples (e.g., \junweirev{\revv{$\mathcal{T}_{[b_i,b_{i-1}]}$} is stored between samples \revv{$b_i$} and \revv{$b_{i-1}$}).} With these time-axis samples created for each particle flow map, we can flexibly construct flow maps of different lengths by selecting different sample points along the trajectory as the start point, while all flow maps converge at the same endpoint, which is the same endpoint as the particle's trajectory. \junweirev{The backward Jacobians from the endpoint to a chosen sample point} can be computed by concatenating all Jacobians along the trajectory. For instance, for a flow map starting from the sample \revv{$b_i$}, the backward Jacobian \revv{$\mathcal{T}_{[b_i,c]}$} can be calculated as \revv{$\mathcal{T}_{[b_{i},b_{i-1}]}\mathcal{T}_{[b_{i-1},b_{i-2}]}...\mathcal{T}_{[b_1,c]}$}. The default flow map $[a,c]$ is a special case with zero sample points on the trajectory.
% \bo{We also need to mention Yitong's adaptive sampling design, and emphasize that we employ the idea on particles with NFM used neural networks.} 

\revvdel{
\begin{figure}[t]
 \centering
 \includegraphics[width=.47\textwidth]{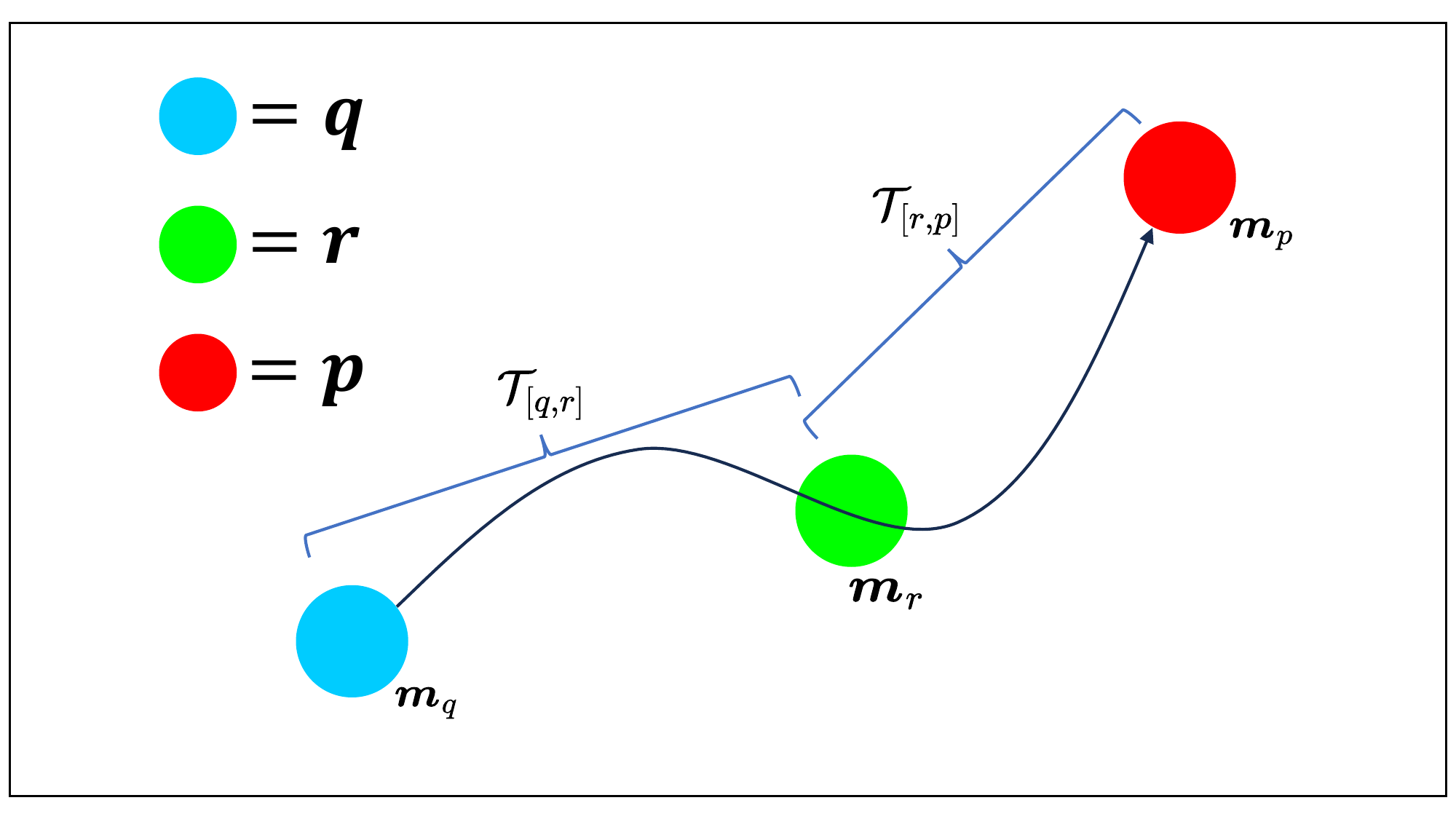}
 \caption{Long-Short Flow Map. The long map is from time $q$ to time $p$, and the short map is from time $r$ to time $p$. We only store the backward map Jacobian $\mathcal{T}_{[q, r]}$ and $\mathcal{T}_{[r, p]}$ and use them to compute $\mathcal{T}_{[q, p]}$. Impulses at time $q$, $r$ and $p$ are $\bm m_q$, $\bm m_r$, $\bm m_p$, respectively.
 % \bo{I feel this figure is not very informative overall. We may consider delete it.}
 }
 \label{fig:PFM_two_map}
\end{figure}
}

\subsubsection{A Single-Sample Case: Long-Short Flow Map}
\label{sec:long_short_flow_map}
We have implemented a single-sample strategy based on the aforementioned description to enable a long-short flow map mechanism. \revvdel{As depicted in Figure~\ref{fig:adaptive_PFM}, we position a single sample point $r$ between $q$ and $p$, with $r$ being closer to $p$ to construct a short map near the endpoint.} \revv{Between $a$ and $c$, we position one time stamp $b$ which is closer to $c$ to construct a short map near the endpoint.} This arrangement naturally produces two flow maps, \revv{$[a,c]$} and \revv{$[b,c]$}, where \revv{$[a,c]$} serves as the long flow map and \revv{$[b,c]$} as the short flow map. We store two backward Jacobians: \revv{$\mathcal{T}_{[a, b]}$} and \revv{$\mathcal{T}_{[b,c]}$}. The short backward Jacobian \revv{$\mathcal{T}_{[b,c]}$} is stored between \revv{$b$} and \revv{$c$}, and the long backward Jacobian \revv{$\mathcal{T}_{[a,c]}$} can be easily computed as shown in the following equation
% \junwei{We need to refer to this equation later so I change it to begin\{equation\}...end\{equation\}}
\revv{
\junweirev{
\begin{equation}
\label{eq:evolve_T_long}
    \mathcal{T}_{[a, c]} = \mathcal{T}_{[a, b]}\mathcal{T}_{[b, c]}.
\end{equation}
}
}

This long-short flow map mechanism was designed to transport both the impulse and its gradient. The long flow map is used for the long-range transport of the impulse, while the short flow map addresses the gradient, which is sensitive to flow distortion. We will specify details in \junweirev{Section~\ref{sec:impulse_transport}}.

% \bo{(done) A missing piece: How the adaptive flow maps are updated in each timeframe. Write a few sentences to explain.}

\junweirev{At each step, we update the short-range flow map by marching \revv{$\mathcal{T}_{[b, c]}$} by one step in parallel with the particle's advection, using our custom $4^{th}$ order of Runge-Kutta (RK4) integration scheme, as detailed in Alg.~\ref{alg:RK4}. Subsequently, the long-range flow map is updated according to Equation~\ref{eq:evolve_T_long}, with the updated \revv{$\mathcal{T}_{[b, c]}$} and \revv{$\mathcal{T}_{[a, b]}$} which has been reinitialized at time \revv{$b$}. We will demonstrate the reinitialization of the flow map in Section~\ref{sec:pfm_reinit}.}

% \bo{We need two illustrative figures here.}
\subsection{Particle Flow Map Reinitialization}
\label{sec:pfm_reinit}
% \bo{(done) TODO}

\junweirev{Owing to the distortion that occurs in the $\mathcal{T}$ when the flow map's range becomes excessively elongated, we reinitialize the flow map at regular intervals. Specifically, for the long-range map, a reinitialization is conducted every $n^L$ steps, during which time \revv{$a$} is reset to time \revv{$c$}, and \revv{$\mathcal{T}_{[a, b]}$} is set to identity:
\revv{
\begin{equation}
\label{eq:reinit_T_q_r_long}
    \mathcal{T}_{[a, b]} \gets \bm{I}.
\end{equation}
}

In addition, to prevent the clustering of particles over time, we will uniformly redistribute particles throughout the entire simulation domain at regular intervals of $n^L$ steps.  Insights into the impact of this particle redistribution, as well as a comparison among different redistribution strategies, is presented in Section~\ref{sec:redistribute_particle}.
}

\junweirev{For the short-range map, a reinitialization is performed every $n^S$ steps, resetting time \revv{$b$} to time \revv{$c$}. 
Given that \revv{$\mathcal{T}_{[a, c]}$} is indirectly updated via \revv{$\mathcal{T}_{[a, b]}$} with the evolved \revv{$\mathcal{T}_{[b, c]}$}, it becomes essential to recalibrate \revv{$\mathcal{T}_{[a, b]}$} to align with the latest \revv{$\mathcal{T}_{[a, c]}$} upon each reinitialization of the short-range map, that is
\revv{
\begin{equation}
\label{eq:reinit_T_q_r_short}
    \mathcal{T}_{[a, b]} \gets \mathcal{T}_{[a, b]}\mathcal{T}_{[b, c]}.
\end{equation}
}
Additionally, \revv{$\mathcal{T}_{[b, c]}$} is reset to the identity:
\revv{
\begin{equation}
\label{eq:reinit_T_r_p}
    \mathcal{T}_{[b, c]} \gets \bm{I}.
\end{equation}
}
Moreover, particle redistribution is not performed when short-range flow map is reinitialized.
}

\junweirev{In most instances, $n^L$ is selected as an integer multiple of $n^S$, ensuring that each reinitialization of the long-range flow map coincides with a reinitialization of the short-range flow map. However, there are situations where $n^L$ is not an integer multiple of $n^S$. In such cases, to prevent the misalignment in the progression of the two maps, the short-range map is still reinitialized at the reinitialization points of the long-range map, even if it hasn't reached its $n^S$ interval. An exploration of how these reinitialization intervals influence the process is detailed in Section~\ref{sec:reinit_steps}.
}

\begin{figure*}[h]
\centering
\begin{minipage}{.49\linewidth}
  \includegraphics[width=\linewidth]{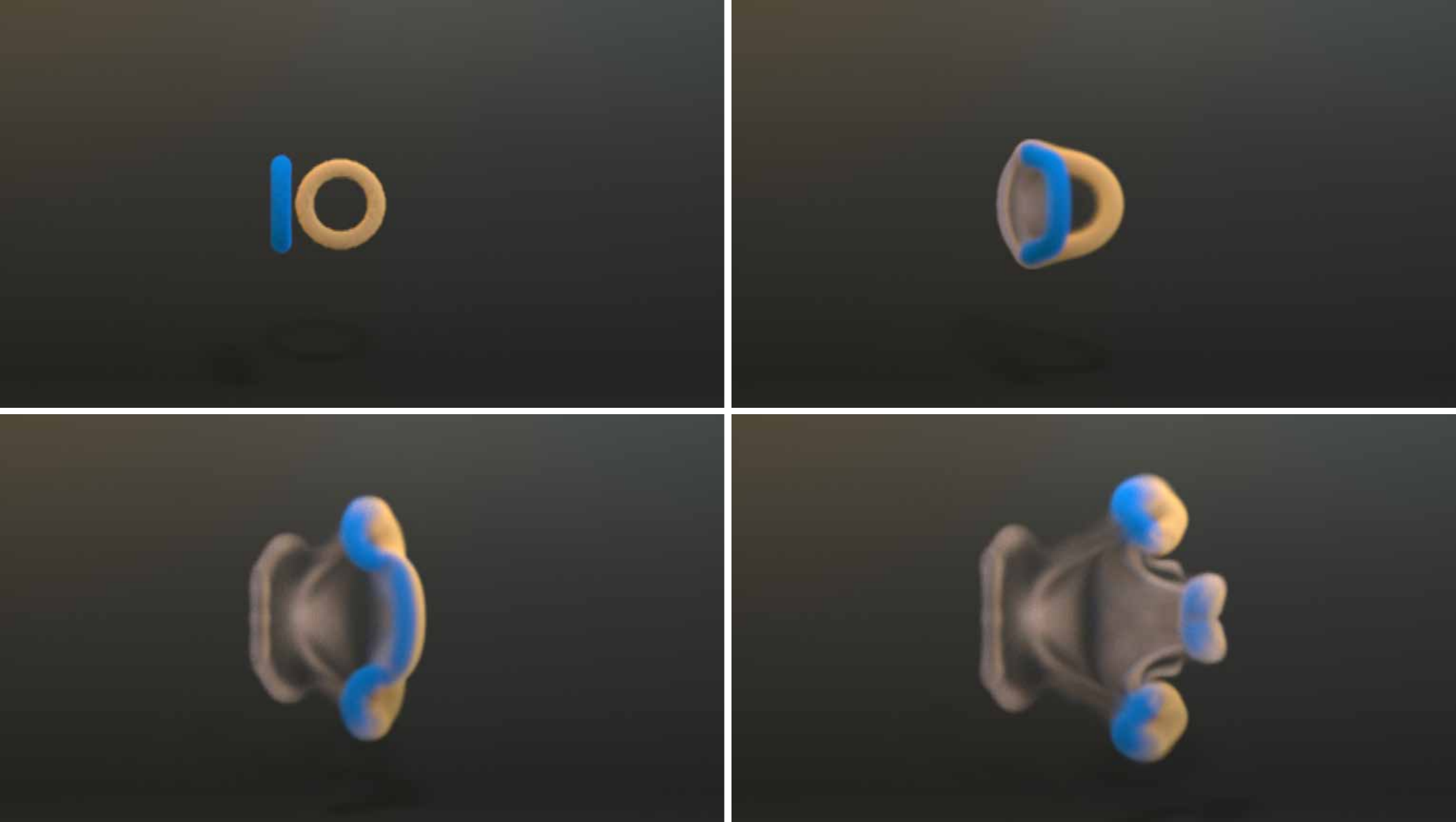}
  \captionof{figure}{{Simulation of oblique vortex rings. The transformation of a set of oblique vortex rings involves a process where the two vortices initially connect on the left side, experience several changes in their structure, and ultimately transform into three separate vortex rings.}}
  \label{fig:oblique}
\end{minipage}
\hspace{.01\linewidth}
\begin{minipage}{.49\linewidth}
  \includegraphics[width=\linewidth]{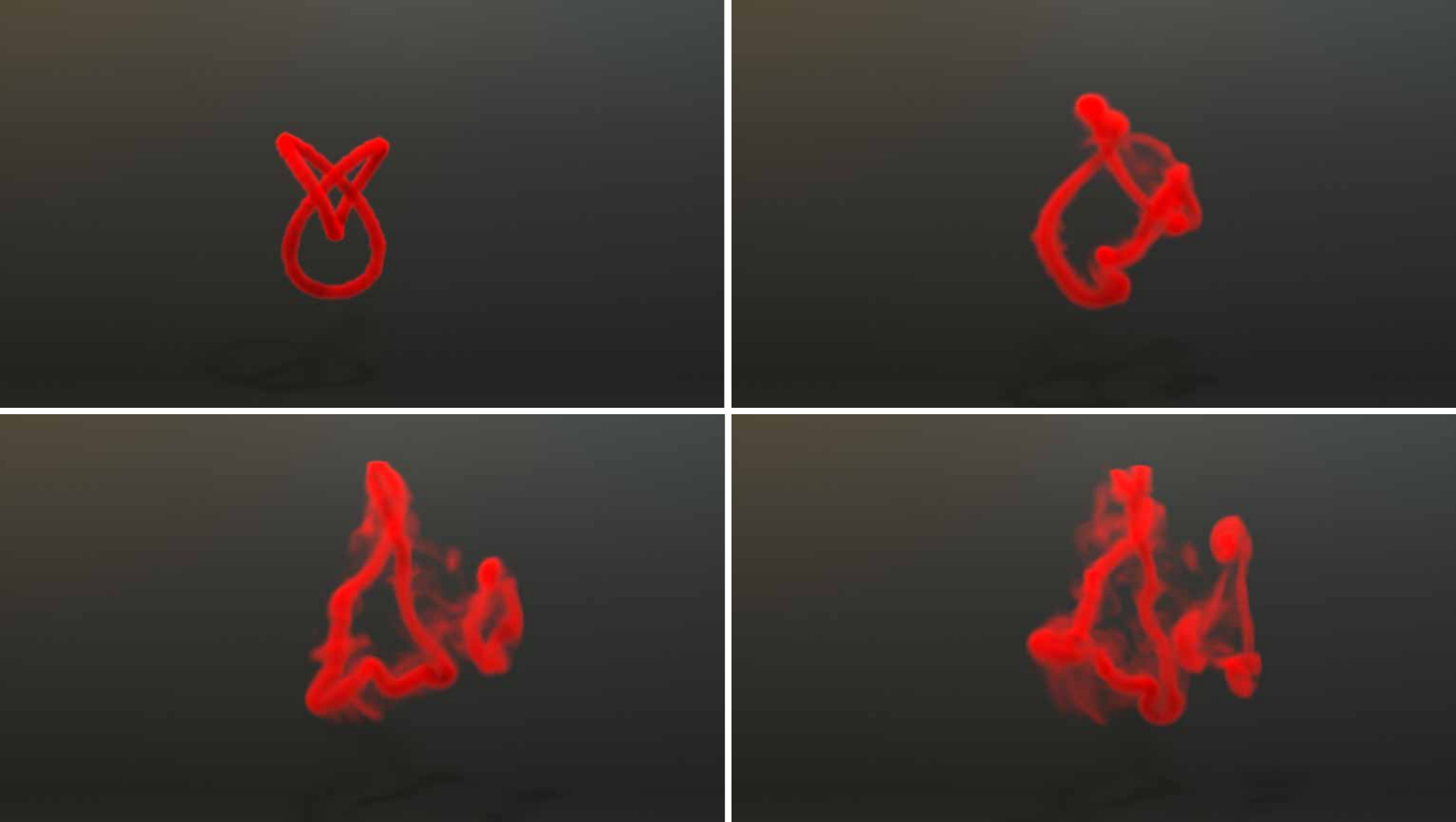}
  \captionof{figure}{{Evolution of a trefoil knot. Initially, the knot moves rightward while rotating. This motion causes collisions and reconnections nearby, breaking the knot into two distinct, unlinked rings. This pattern aligns well with the results seen in the referenced experiments \cite{kleckner2013creation}.}}
  \label{fig:trefoil}
\end{minipage}
\end{figure*}

\section{Eulerian-Lagrangian Framework}
\label{sec:Eulerian_Lagrangian_Framework}
\newcolumntype{z}{X}
\newcolumntype{s}{>{\hsize=.25\hsize}X}
\begin{table}
\caption{Summary of quantities stored on particles and grid.}
\centering\small
\begin{tabularx}{0.47\textwidth}{scz}
\hlineB{3}
Notation & Location & Definition\\
\hlineB{2.5}
\junweirev{\hspace{12pt}\revv{$\bm x_p$}} & particle & \junweirev{position of the \revv{$p$}-th particle} \\
\hlineB{1}
\hspace{12pt}\revv{$\mathcal{T}_{[a, b]}$} & particle & backward map Jacobian from time \revv{$b$} to \revv{$a$} \\
\hlineB{1}
\hspace{12pt}\revv{$\mathcal{T}_{[b, c]}$} & particle & backward map Jacobian from time \revv{$c$} to \revv{$b$} \\
\hlineB{1}
\hspace{12pt}\revv{$\bm{m}_a$} & particle & impulse on particles at time \revv{$a$} \\
\hlineB{1}
\hspace{12pt}\revv{$\bm{m}_b$} & particle & impulse on particles at time \revv{$b$} \\
\hlineB{1}
\hspace{8pt}\revv{$\nabla\bm{m}_b$} & particle & impulse gradients on particles at time \revv{$b$} \\
\hlineB{2.5}
\hspace{12pt}$\bm{u}_i$ & grid faces & \junweirev{velocity of the $i$-th cell}\\
\hlineB{1}
\junweirev{\hspace{12pt}$w_i$} & \junweirev{grid faces} & \junweirev{weights sum of the $i$-th cell}\\
% \hlineB{1}
% \junweirev{\hspace{12pt}$I_i$} & \junweirev{grid center} & \junweirev{indices of particles located within the $i$-th cell}\\
% \hlineB{1}
% \hspace{12pt} & grid centers & particle indices within cell\\
\hlineB{3}
% \hspace{12}
\end{tabularx}
\vspace{5pt}
% \caption{Summary of quantities stored on particles and grid. 
% % $q$, $r$, and $p$ refer to the start point of the long-range map, the start point of the short-range map, and the end point of both two maps, respectively.
% % \junwei{don't know what notation to use for particle index on grid and P2G weight summation}
% % \bo{(done) w is also a grid variable?}
% }
\label{tab:grid_particle_quantity}
% \vspace{-0.4in}
\end{table}

Equipped with the design of our particle flow map, we aim to develop a hybrid Eulerian-Lagrangian framework for simulating incompressible flow. Specifically, we plan to solve the impulse-based fluid model as governed by Equation~\ref{eq:impulse_euler_eq}, utilizing the impulse flow map outlined in Equation~\ref{eq:evolve_imp}. The overarching goal is to leverage particles for the accurate transport of fluid impulse and use the background grid to solve the Poisson equation and enforce incompressibility.

To construct this hybrid fluid solver, we must address three key issues: \textit{(1)} The discretization of different physical quantities on particles and grids, \textit{(2)} The construction of accurate flow maps for the transport of impulse, and \textit{(3)} The transfer of fluid impulse from particles to the grid for the incompressibility solution. We will elaborate our numerical solutions w.r.t. these aspects as follows. 

\subsection{Discretization}
We utilize particles to transport the fluid impulse and employ the grid for computing the divergence-free velocity field. 
% particles
\junweirev{For each particle \revv{$p$}, we store its current position \revv{$\bm x_p$}.} We utilize the long-short flow map mechanism introduced in Section~\ref{sec:long_short_flow_map} to enable the implementation of dual-layer flow maps on each particle. In particular, we maintain two backward Jacobians, \revv{$\mathcal{T}_{[a, b]}$} and \revv{$\mathcal{T}_{[b, c]}$}, we store snapshots of the impulse values, \revv{$\bm{m}_a$} and \revv{$\bm{m}_b$}, which are sampled at times \revv{$a$} (start point of the long-range map) and \revv{$b$} (start point of the short-range map) respectively. We also store the impulse gradient \revv{$\nabla\bm{m}_b$} for time \revv{$b$}.
% grid
On the grid side, we discretize the velocity field $\bm u$ on a MAC grid, with its velocity components stored on grid faces. We also store a scalar field \junweirev{$w_i$} on grid faces to specify particle-to-grid interpolation weights \junweirev{summation of cell $i$}\revvdel{, as well as an array $I_i$ on each cell to catalog the indices of particles located within the cell}. 
% \bo{(done) JW: confirm this. Anything else on grid?} 
Table~\ref{tab:grid_particle_quantity} outlines all the attributes stored on particles and the background grid. 
% notes
It is worth noting that a few other particle quantities, such as velocity \revv{$\bm u_p$}, endpoint impulse \revv{$\bm m_c$}, and endpoint backward Jacobian \revv{$T_{[a,c]}$}, can be calculated dynamically during the simulation. As a result, these quantities do not require dedicated storage as particle attributes.

%Specifically, particle attributes like $\bm x_p$, $\mathcal{T}_{[q, r]}$, $\mathcal{T}_{[r, p]}$, $\bm{m}_q$, $\bm{m}_r$, $\nabla\bm{m}_r$ are stored on particles. Conversely, $\bm{u}_i$ is maintained on the grid.

\subsection{Impulse Transport}
\label{sec:impulse_transport}

\subsubsection{Impulse Mapping}
Building upon the particle flow map updated at each timestep, we first advance the fluid impulse on each particle. This process is straightforwardly executed by computing the impulse values at time \revv{$c$}, using the backward Jacobian of the long-range flow map \revv{$\mathcal{T}_{[a,c]}$}. Specifically, this can be expressed as:
\revv{
\begin{equation}
\label{eq:evolve_imp_long}
    \bm m_c \gets \mathcal{T}_{[a, c]}^T \bm m_{a}.
\end{equation}
}
Here, \revv{$\mathcal{T}_{[a, c]}$} is calculated on-the-fly by multiplying \revv{$\mathcal{T}_{[a, b]}$} and \revv{$\mathcal{T}_{[b, c]}$}. 

\setlength{\abovecaptionskip}{12pt}
% \begin{figure*}[t]
%  \centering
%  \includegraphics[width=.99\textwidth]{img/examples/four_vorts_smoke.pdf}
%  \caption{{The interaction and merging of four vortices, each positioned at right angles to the adjacent ones, lead to their eventual collision and reconnection. This process results in the formation of two bigger vortices, each resembling a four-pointed star. These enlarged vortices then move towards the left and right boundaries, where they subsequently divide into four separate vortex tubes.}}
%  \label{fig:four_vorts}
% \end{figure*}
\begin{figure*}[t]
\centering
\begin{minipage}{.49\linewidth}
  \includegraphics[width=\linewidth]{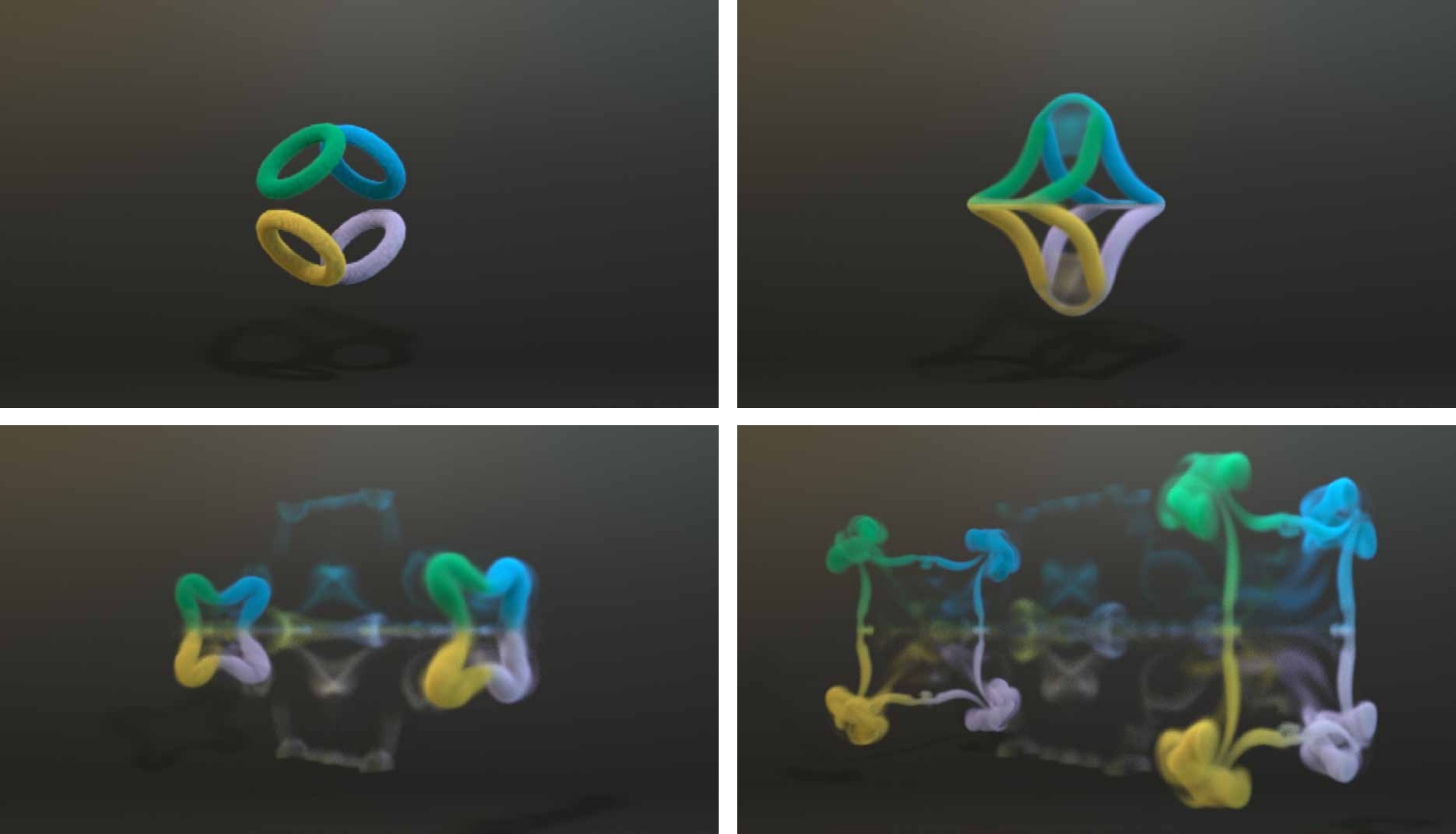}
  \caption{{The interaction and merging of four vortices, each positioned at right angles to the adjacent ones, lead to their eventual collision and reconnection. This process results in the formation of two bigger vortices, each resembling a four-pointed star. These enlarged vortices then move towards the left and right boundaries and subsequently divide into four separate vortex tubes.}}
 \label{fig:four_vorts}
\end{minipage}
\hspace{.01\linewidth}
\begin{minipage}{.49\linewidth}
  \includegraphics[width=\linewidth]{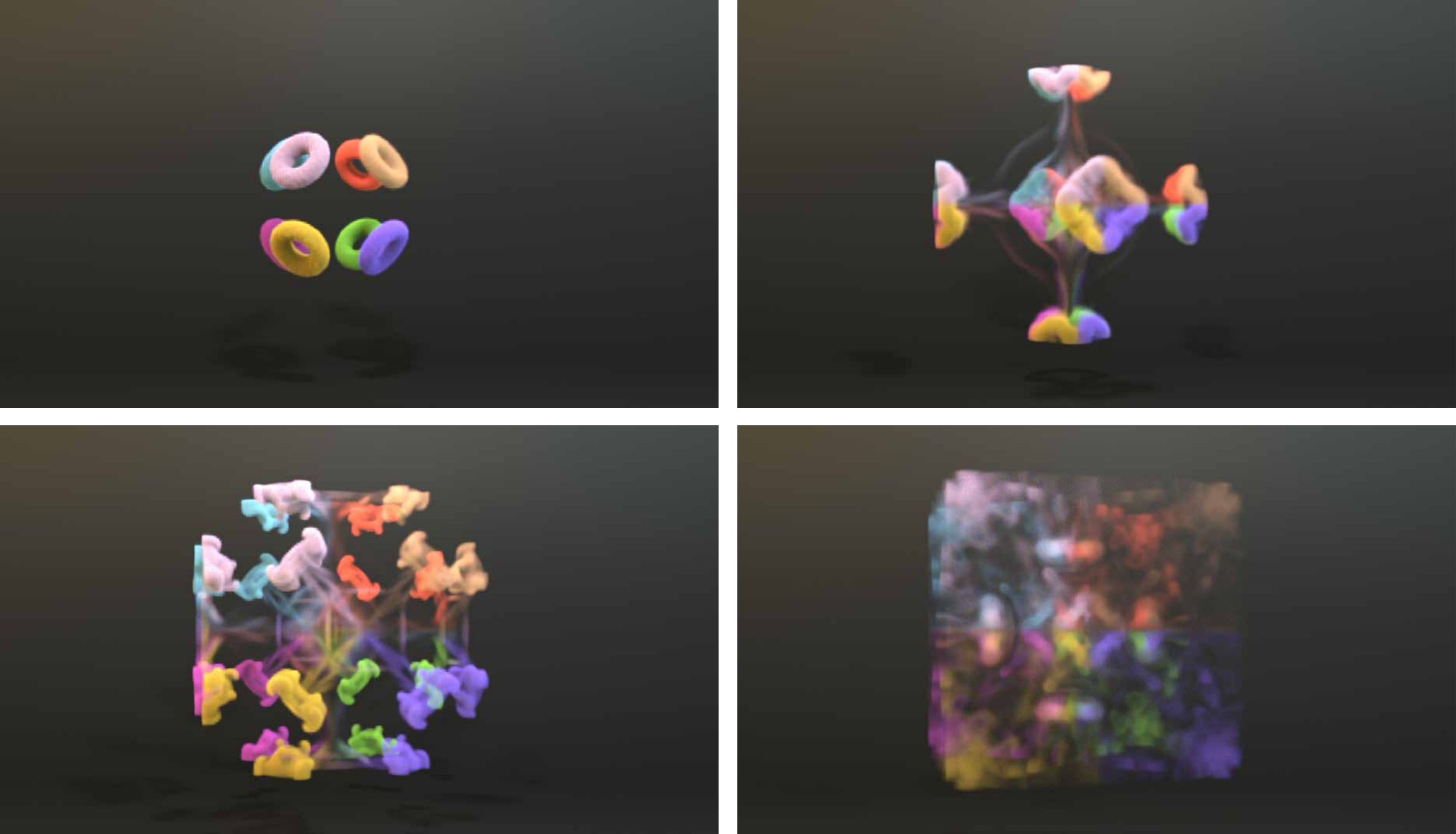}
 \caption{{Simulation of eight vortices collision. In a cubic simulation space, eight rings forming an octahedron at a 35.26° angle to the z-axis collide and reform into six rings. These rings move orthogonally to the cube's walls, splitting into vortex tubes that collide and recombine. Eventually, they form six rings again moving towards the center and finally split into eight separate parts.}}
 \label{fig:eight_vorts}
\end{minipage}
\end{figure*}
\subsubsection{Impulse Gradient Mapping}
Next, we evolve the impulse gradient using our particle flow map. 
\revvdel{
We opt for the short flow map to calculate the impulse gradient, recognizing that it is a high-order quantity highly sensitive to distortions in the flow map. Therefore, the focus of our calculations for the impulse gradient is exclusively on the short\junweirev{-range} flow map $[r,p]$. As specified in Equation~\ref{eq:evolve_grad_imp}, the computation of \junweirev{$\nabla m$} entails not only the Jacobian of the backward flow map, $\mathcal{T}$, but also its Hessian, $\nabla \mathcal{T}$. To this end, we have developed a particle-based scheme for calculating this term. \junweirev{Specifically, for particle $k$, $\nabla\mathcal{T}_{[r, p]}^k$ is computed by accumulating the $\mathcal{T}_{[r, p]}^l$ values of other particles, each weighted by the quadratic interpolation kernel $w_{kl}$ defined in Appendix~\ref{sec:interp_kernel}, as described by:
\begin{equation}
\label{eq:compute_grad_T}
    \nabla \mathcal{T}_{[r, p]}^k \gets \sum_{l} \nabla w_{kl}\mathcal{T}_{[r, p]}^l,
\end{equation}
where the superscript of $\mathcal{T}$ denotes the particle index, thereby $\nabla\mathcal{T}_{[r, p]}^k$ and $\mathcal{T}_{[r, p]}^l$ represents the backward map Hessian of particle $k$ and the backward map Jacobian of particle $l$, respectively.} 
% \bo{(done) Specify the subscript of T for this equation.}

% \bo{(done) How was the gradient of w calculated? We should first define w, and specify how to calculate w.}
As $w_{kl}$ diminishes to zero beyond 1.5 times the cell width from particle $k$, this method effectively aggregates the $\mathcal{T}_{[r, p]}^l$ values of all particles $l$ within a neighborhood 1.5 times the cell width around particle $k$. For practical computation, this step is executed in two substeps: \junweirev{firstly, indices of particles located within each cell $i$ is collected in an array $I_i$}; secondly, for each particle, the summation is carried out over particles located within two hops of neighboring cells relative to its residing cell. This scheme was inspired by the Particle Strength Exchange (PSE) method, commonly used for viscosity treatment in vortex-in-cell method literature \cite{rivoalen1999numerical, zhang2016resolving}. 
% It's important to note that alternative methods for computing $\nabla \mathcal{T}$ exist, and an analysis of these different methodologies is provided in Section~\ref{sec:analysis_grad_T}.

After computing $\nabla \mathcal{T}_{[r, p]}$ for each particle,}
\revv{We map $\nabla \bm m_b$ from time $b$ to time $c$ utilizing $\mathcal{T}_{[b, c]}$, as depicted by
\begin{equation}
\label{eq:evolve_grad_imp_short}
    \nabla \bm m_c \gets \mathcal{T}_{[b, c]}^T \nabla \bm m_{b} \mathcal{T}_{[b, c]} 
    %+ \nabla \mathcal{T}_{[r, p]}^T \bm m_{r}.
\end{equation}
This equation employs the first term from the right-hand side of Equation~\ref{eq:evolve_grad_imp}, while the second term is omitted due to the overhead and difficulty of accurately calculating \revvv{it. Additionally, excluding the term does not reduce the desired simulation qualities.} This issue is further discussed in Section~\ref{sec:discussion_hessian}.
}

\subsubsection{Impulse Particle-to-Grid Transfer}
%Leveraging our proposed PFM, we attain precise propagation of $\mathcal{T}$ on particles, facilitating the accurate evolution of impulse on particles. However, the eventual computation of velocity from impulse requires resolving Poisson equations on grid, which necessitates transferring impulse values from particles onto the grid. This process introduces significant numerical errors during interpolation, undermining the precision initially achieved by our precise flow map. To address this challenge, we introduce an innovative approach that not only transmits impulse on particles but also the impulse gradients. We then employ the long-range map and a short-range map proposed in Section~\ref{sec:short_long_map} to facilitate their evolution on particles. Ultimately, the resultant values are incorporated into the particle-to-grid process to significantly elevate the precision of the impulse field on the grid. The method's pseudo code is outlined in Alg.~\ref{alg:pfm_simulation}. The specifics of these processes will be elucidated subsequently.
After calculating \revv{$\bm m_c$} and \revv{$\nabla \bm m_c$}, we execute a particle-to-grid transfer using values on particles to compute the impulse field $\bm m_i$ on the grid as: 
\begin{equation}
\label{eq:imp_P2G}
    \bm m_i \gets \sum_p w_{ip}(\bm m_c^p + \nabla \bm m_c^p (\bm x_i - \bm x_p)) \,/\, \sum_p w_{ip}, 
\end{equation}
where the superscript of \revv{$\bm m_c$} and \revv{$\nabla \bm m_c$} means they are stored on particle \revv{$p$}. The impulse field $\bm m_i$ will be projected by solving Poisson equation, ultimately yielding the velocity field $\bm u_i$ on grid.
% \bo{It is better to talk about some connection to APIC or other particle-grid interpolation schemes.}

\paragraph{Validation Experiment} 
We execute an experiment to validate the efficacy of our proposed approach. Initially, we set up a velocity field as illustrated in Figure~\ref{fig:motivational_experiment}(a) and (b). 
% \bo{(done) which subplot?} 
Subsequently, particles are placed within this velocity field, and both initial velocity and velocity gradients are interpolated from the grid to the particles to establish the initial impulse and impulse gradients. The particles are then advected within the field, and both impulse and impulse gradients are evolved through the utilization of the two maps as previously described. The evolved values are then applied in a particle-to-grid process to compute the velocity field at each step. Figure~\ref{fig:motivational_experiment}(e) depicts the discrepancy between the analytical velocity field and the velocity field generated by our method. We also assess the outcomes when only impulse is transferred to the grid. It is evident that incorporating impulse gradients in the particle-to-grid process significantly reduces the error. \revv{In addition, we conduct this experiment using NFM, and the results of our method show a smaller discrepancy from the analytical velocity field compared to NFM, further demonstrating the efficacy of our approach.}

\subsection{Impulse Reinitialization}
% \bo{(done) Move part of this to PFM reinit, keep the impulse part here. Integrate equations from Sec 6 Step 1 and 2 to these two sections accordingly.}
Every $n^L$ steps and $n^S$ steps, reinitialization is undertaken for long-range and short-range flow maps, respectively, as described in Section~\ref{sec:pfm_reinit}. During the reinitialization of the long-range map, since time \revv{$a$} is reset to time \revv{$c$}, \revv{particles $p$'s impulse $\bm m_a^p$} is recalibrated by interpolating the velocity field at time \revv{$c$}:
\revv{
\begin{equation}
\label{eq:reinit_imp_long}
    \bm m_{a}^p \gets \sum_i w_{ip} \bm u_i.
\end{equation}
}
Similarly, when the short-range flow map is reinitialized with time \revv{$b$} reset to time \revv{$c$}, \revv{particles $p$'s impulse gradient $\nabla \bm m_b^p$} is recalibrated via interpolation from the velocity field at time \revv{$c$}:
\revv{
\begin{equation}
\label{eq:reinit_imp_short}
    \nabla \bm m_{b}^p \gets \sum_i \nabla w_{ip} \bm u_i.
\end{equation}
}

\section{Time Integration}
\revvdel{
\begin{figure*}[t]
 \centering
 \includegraphics[width=.98\textwidth]{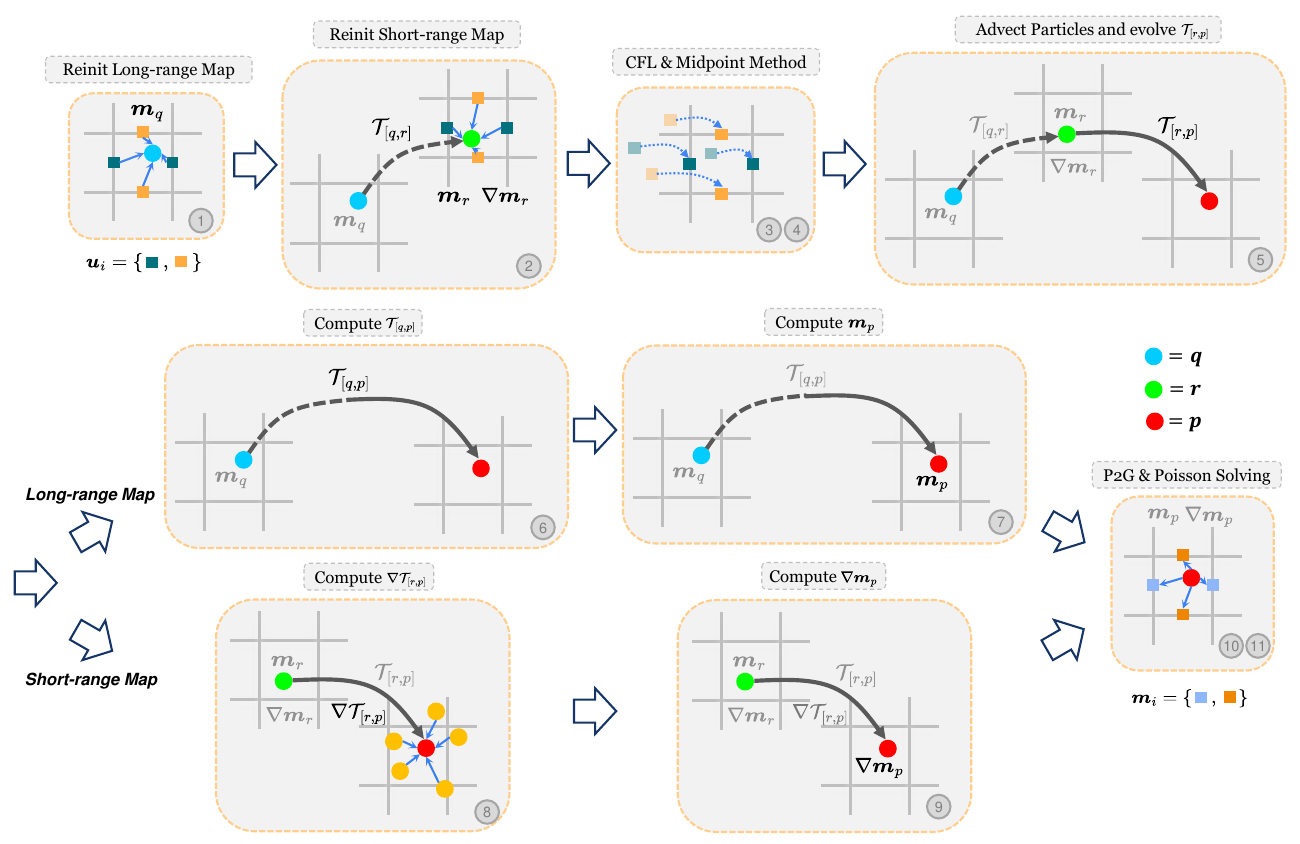}
 \caption{Illustration of time integration: The process begins at step 1, with step 1 and step 2 recurring every $n^L$ and $n^S$ steps, respectively. From step 3 onwards, we proceed in sequence through to step 5. Subsequently, steps 6 and 7 are associated with the long-range map, whereas steps 8 and 9 pertain to the short-range map. Execution of these two maps can occur concurrently. Following this parallel processing, the results derived from both maps converge in steps 10 and 11. The short-range map from time $r$ to time $p$, is depicted using a solid line, while the map from time $q$ to time $r$, is represented by a dashed line. By connecting these solid and dashed lines, the long-range map from time $q$ to time $p$ is formed. At each stage, the quantities computed in the current step are presented distinctly, devoid of transparency, contrasting with the previous steps’ quantities, which are displayed with transparency for differentiation. Moreover, circles symbolize particles, while squares denote grid faces.
 % \bo{I feel this figure is not very informative overall. We may consider delete it.}
 }
 \label{fig:time_integration}
\end{figure*}
}

\begin{algorithm}[t]
\caption{Particle Flow Map Simulation}
\label{alg:pfm_simulation}
\begin{flushleft}
        \textbf{Initialize:} $\bm u_i$ to initial velocity; \revv{$\mathcal{T}_{[a, b]}$}, \revv{$\mathcal{T}_{[b, c]}$} to $\bm{I}$
\end{flushleft}
\begin{algorithmic}[1]
\For{$k$ in total steps}
\State $j \gets k \Mod {n^L}$;
\State $l \gets k \Mod {n^S}$;
\If{$j$ = 0}
\State Uniformly distribute particles;
\State Reinitialize \revv{$\bm m_a$}; \hfill $\triangleright$ eq.~(\ref{eq:reinit_imp_long})
\State Reinitialize \revv{$\mathcal{T}_{[a,b]}$}\; \hfill $\triangleright$ eq.~(\ref{eq:reinit_T_q_r_long})
\EndIf
\If{$l$ = 0}
\State Reinitialize \revvdel{$\bm m_r$ and }\revv{$\nabla \bm m_b$}; \hfill $\triangleright$ eq.~(\ref{eq:reinit_imp_short})
\State Reinitialize \revv{$\mathcal{T}_{[a, b]}$} and \revv{$\mathcal{T}_{[b, c]}$}; \hfill $\triangleright$ eq.~(\ref{eq:reinit_T_q_r_short}, \ref{eq:reinit_T_r_p})
\EndIf
\State Compute $\Delta t$ with $\bm{u}_i$ and the CFL number;
\State Estimate midpoint velocity $\bm{u}_i^\text{mid}$; \hfill $\triangleright$ Alg. \ref{alg:midpoint}
\State March $\bm x_p$, \revv{$\mathcal{T}_{[b, c]}$} with $\bm{u}_i^\text{mid}$ and $\Delta t$; \hfill $\triangleright$ Alg.~\ref{alg:RK4}
\State Compute \revv{$\mathcal{T}_{[a, c]}$} with \revv{$\mathcal{T}_{[a, b]}$} and \revv{$\mathcal{T}_{[b, c]}$}; \hfill $\triangleright$ eq.~(\ref{eq:evolve_T_long})
\State Compute \revv{$\bm m_c$} with \revv{$\bm m_a$} and \revv{$\mathcal{T}_{[a, c]}$}; \hfill $\triangleright$ eq.~(\ref{eq:evolve_imp_long})
\revvdel{\State Compute $\nabla \mathcal{T}_{[r, p]}$; \hfill $\triangleright$ eq.~(\ref{eq:compute_grad_T})}
\State \revv{Compute $\nabla \bm m_c$ with \revvdel{$\bm m_b$, }$\nabla \bm m_b$ and $\mathcal{T}_{[b, c]}$ \revvdel{and $\nabla \mathcal{T}_{[b, c]}$}; \hfill $\triangleright$ eq.~(\ref{eq:evolve_grad_imp_short})}
\State Compute $\bm m_i$ by transfering \revv{$\bm m_c$}, \revv{$\nabla \bm m_c$} to grid; \hfill $\triangleright$ eq.~(\ref{eq:imp_P2G})
\State $\bm u_i \gets \textbf{Poisson}(\bm m_i)$;

% \State March $\phi, \mathcal{F}$ according to Alg.~\ref{alg:RK4}, using $\bm{u}_\text{mid}$ and $\Delta t_j$;
% \State Reconstruct $\bm{u}$ with $(\bm{u}_0, \psi, \mathcal{T}, \phi, \mathcal{F})$ as in Alg. \ref{alg:bfecc};
% \If{use external force}
% \State $\hat{\bm{u}} \gets \bm{u} + \int_0^t\bm{f}_\text{ext}(\psi(\bm x, \tau), \tau)d\tau$;
% \State $\bm{u} \gets \textbf{Poisson}
% (\hat{\bm{u}})$;
% \EndIf
\EndFor{}
\end{algorithmic}
\end{algorithm}

% \setlength{\abovecaptionskip}{12pt}
% \begin{figure*}[t]
%  \centering
%  \includegraphics[width=.99\textwidth]{img/examples/four_vorts_smoke.pdf}
%  \caption{{The interaction and merging of four vortices, each positioned at right angles to the adjacent ones, lead to their eventual collision and reconnection. This process results in the formation of two bigger vortices, each resembling a four-pointed star. These enlarged vortices then move towards the left and right boundaries, where they subsequently divide into four separate vortex tubes.}}
%  \label{fig:four_vorts}
% \end{figure*}
% \setlength{\abovecaptionskip}{12pt}
% \begin{figure*}[t]
%  \centering
%  \includegraphics[width=.99\textwidth]{img/examples/propeller_vorticity.pdf}
%  \caption{The series of images effectively illustrates the vortex formation process initiated by a propeller as it rotates under wind influence. These visuals provide a clear view of the intricate, spiral-shaped patterns formed by the interconnected vortex tubes. Additionally, they highlight the dynamic alterations in the trailing wake that occur as a direct result of the propeller's movements. }
%  \label{fig:propeller}
% \end{figure*}

\revv{In this section, we delineate the time integration scheme of our method, and the pseudocode of our method is outlined in Alg~\ref{alg:pfm_simulation}.}
\revvdel{which integrates both a long-range and short-range map, reinitialized at intervals of $n^L$ and $n^S$ steps, respectively. Next, the calculation of the time step $\Delta t$ and midpoint velocity is undertaken. Utilizing the long-range and short-range maps, we evolve the impulse and its gradients from their points of reinitialization to the current step. Finally, we transfer these evolved values from particles to grid faces, creating the impulse field on the grid faces, which is subsequently projected by solving the Poisson equation, resulting in the derivation of the velocity field. The pseudocode of our method is outlined in Alg~\ref{alg:pfm_simulation}. 
% And each step of our method is illustrated in Figure~\ref{fig:time_integration}. 
We will elucidate these steps in further detail as follows. }
% \junwei{still use $n^L$, $n^S$?}
\begin{enumerate}[leftmargin=*]
    \item \textit{Reinitialize Long-range Map}. \ 
    Every $n^L$ step, we will uniformly redistribute particles throughout the entire simulation domain. Then, we reinitialize each particle's initial long-range map impulse \revv{$\bm m_a$} by interpolating the velocity field $\bm{u}_i$ on grid faces, according to Equation~\ref{eq:reinit_imp_long}. And backward map Jacobian \revv{$\mathcal{T}_{[a, b]}$} is reset to identity, outlined in Equation~\ref{eq:reinit_T_q_r_long}.
    % \begin{equation}
    % \label{eq:reinit_imp_long}
    %     \begin{dcases}
    %         \bm m_{q} \gets \sum_i w_{iq} \bm u_i, \\
    %         \mathcal{T}_{[q, r]} \gets \bm{I}.
    %     \end{dcases}
    % \end{equation}
    
    \item \textit{Reinitialize Short-range Map}. \ \revv{Every $n^S$ step, each particle's \revvdel{initial short-range map impulse $\bm m_r$ and }impulse gradients $\nabla \bm m_b$ are reinitialized through interpolation from the grid's velocity field $\bm{u}_i$, as depicted in Equation~\ref{eq:reinit_imp_short}.}
    % \begin{equation}
    % \label{eq:reinit_imp_short}
    %     \begin{dcases}
    %         \bm m_{r} \gets \sum_i w_{ir} \bm u_i, \\
    %         \nabla \bm m_{r} \gets \sum_i \nabla w_{ir} \bm u_i.
    %     \end{dcases}
    % \end{equation}
    In addition, each particle's backward map Jacobian \revv{$\mathcal{T}_{[a, b]}$} and \revv{$\mathcal{T}_{[b, c]}$} are also reset by Equations~\ref{eq:reinit_T_q_r_short} and~\ref{eq:reinit_T_r_p}.
    % \begin{equation}
    % \label{eq:reinit_T_two_maps}
    %     \begin{dcases}
    %         \mathcal{T}_{[q, r]} \gets \mathcal{T}_{[q, r]}\mathcal{T}_{[r, p]}, \\
    %         \mathcal{T}_{[r, p]} \gets \bm{I}.
    %     \end{dcases}
    % \end{equation}
    % \bo{(done) Move these equations to Sec 5.3, and make these two bullets shorter.}
    
    \item \textit{CFL Condition}. \ We compute $\Delta t$ with velocity field $\bm{u}_i$ and the CFL number.

    \item \textit{Midpoint Method}. \ 
    Similar to NFM \cite{deng2023fluid}, we employ the second-order, midpoint method as outlined in Alg.~\ref{alg:midpoint} to substantially diminish the truncation error associated with temporal integration. 
    % For this purpose, we have $\psi$ and $\mathcal{T}$ on the grid. We conduct an RK4 backtrack using the velocity field $\bm u_i$ and a time step of $0.5 \cdot \Delta t$ as described in Alg.~\ref{alg:RK4}. Subsequently, the backtracked impulse undergoes evolution via $\mathcal{T}$, resulting in $\bm{u}_i^\text{mid}$. It's important to note that since $\psi$ and $\mathcal{T}$ on the grid are exclusively utilized in the midpoint method, they are considered temporary variables. Consequently, they can be discarded immediately after the midpoint process, ensuring they do not contribute to an increase in our method's memory footprint.
    % \bo{(done) Make this paragraph shorter. Like other paragraphs.}
    
    \item \textit{Advect Particles and Evolve \revv{$\mathcal{T}_{[b, c]}$}}. \ We march $\bm x_p$, \revv{$\mathcal{T}_{[b, c]}$} with Alg.~\ref{alg:RK4}, using $\bm{u}_i^\text{mid}$ and $\Delta t$.

    \item \textit{Compute \revv{$\mathcal{T}_{[a, c]}$}}. \ We employ \revv{$\mathcal{T}_{[a, b]}$} and \revv{$\mathcal{T}_{[b, c]}$} to calculate \revv{$\mathcal{T}_{[a, c]}$}, according to Equation~\ref{eq:evolve_T_long}.

    \item \textit{Compute \revv{$\bm m_c$}}. \ We evolve \revv{$\bm m_a$} from time \revv{$a$} to time \revv{$c$} using long-range map to calculate \revv{$\bm m_c$}, according to Equation~\ref{eq:evolve_imp_long}.

    \revvdel{\item \textit{Compute $\nabla \mathcal{T}_{[r, p]}$}. \ For each particle $k$, we aggregate the backward map gradient $\mathcal{T}_{[r, p]}^k$ of other particles $l$, weighted by $w_{kl}$ to calculate $\nabla \mathcal{T}_{[r, p]}$, as outlined by Equation~\ref{eq:compute_grad_T}.}
    % \begin{equation}
    % \label{eq:compute_grad_T}
    %     \nabla \mathcal{T}_{[r, p]}^k \gets \sum_{l} \nabla w_{pl}\mathcal{T}_{[r, p]}^l.
    % \end{equation}
    % where $\mathcal{T}_{[r, p]}^l$ is the backward map Jacobian of particle $l$.
    % \bo{(done) This equation was mentioned earlier}
    
    \item \textit{Compute \revv{$\nabla \bm m_c$}}. \ \revv{We evolve \revvdel{$\bm m_{r}$ and }$\nabla \bm m_{b}$ from time $b$ to time $c$ using short-range map to calculate impulse gradients $\nabla \bm m_c$, according to Equation~\ref{eq:evolve_grad_imp_short}.}

    \item \textit{Particle-to-grid}. \ We transfer impulse \revv{$\bm m_c$} and impulse gradients \revv{$\nabla \bm m_c$} from particles to grid faces, resulting in impulse field $\bm m_i$ on grid, according to Equation~\ref{eq:imp_P2G}.

    \item \textit{Poisson Solving}. \ We obtain velocity field $\bm u_i$ on the grid from impulse field $\bm m_i$ by solving Poisson equation.
\end{enumerate}

\begin{figure*}[t]
 \centering
 \includegraphics[width=.985\textwidth]{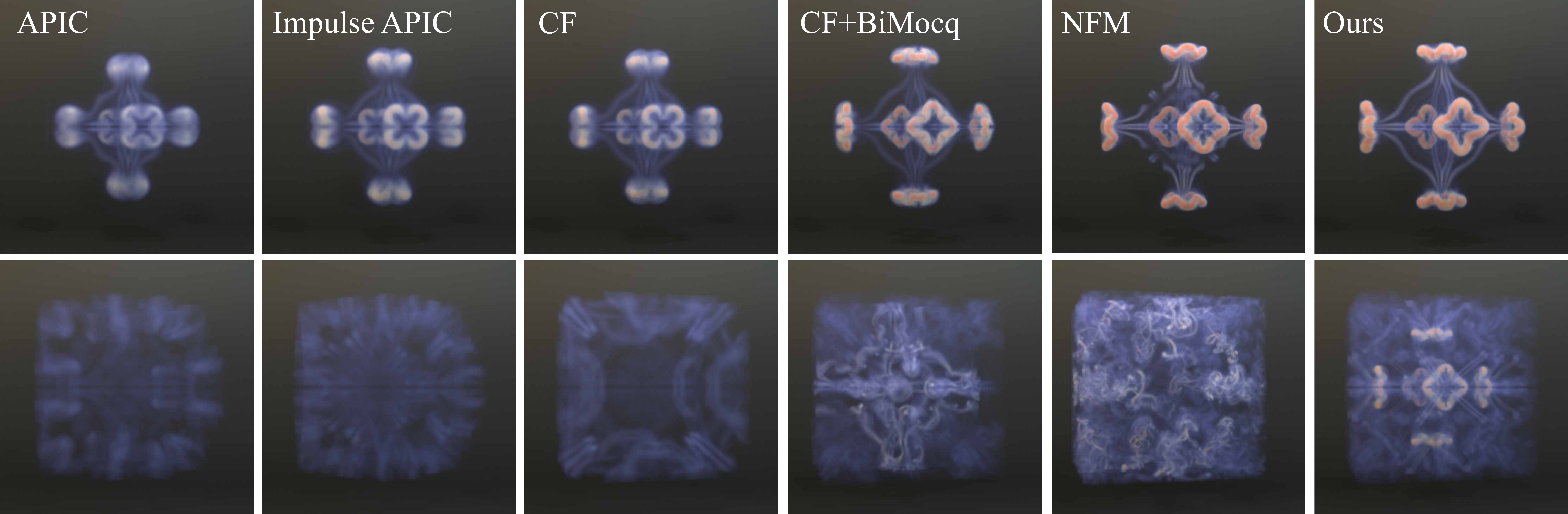}
 \caption{\revv{We show the comparison of collision between eight vortices. \revvv{The figures on the first row display the results at $t=10s$, and the figures on the second row display the results at $t=55s$}. From this experiment, we show our method is the only one capable of recovering the six vortex tubes after the vortices collide with the walls and reflect. This \revv{manifests} the ability of our method to preserve vorticity and the correctness of our treatment with solid \revv{boundaries}.}}
 \label{fig:eight_vort_compare}
\end{figure*}
\section{Validation}

\setlength{\abovecaptionskip}{12pt}
% \begin{figure*}[t]
%  \centering
%  \includegraphics[width=.985\textwidth]{img/examples/3D_eight_vorts_vort_compare.pdf}
%  \caption{\revv{We show the comparison of collision between eight vortices. From this experiment, we show our method is the only one capable of recovering the 6 vortex tubes after the vortices collide with the walls and reflect. This \revv{manifests} the ability of our method to preserve vorticity and the correctness of our treatment with solid \revv{boundaries}.}}
%  \label{fig:eight_vort_compare}
% \end{figure*}

The goal of this section is to validate the effectiveness of our proposed method. We will first introduce a variation of PFM, which will be used for comparison. Then, we will compare PFM with \revv{five} methods to validate its effectiveness. Next, we will conduct an ablation study to validate some steps of our method.
\label{sec:validation}

\subsection{Comparison to Other Approaches}
\label{sec:compare_to_other_methods}
% \junwei{How to introduce Impulse-modified APIC?}

In this section, we assess the effectiveness of our method against established benchmarks, including methods of \revv{CF \cite{nabizadeh2022covector}, CF+BiMocq \cite{nabizadeh2022covector,qu2019efficient}}, NFM \cite{deng2023fluid}, APIC \cite{jiang2015affine}, and an impulse-modified APIC \revv{(IM APIC)} as we specified below. Our primary focus lies in comparing our method with these techniques regarding vortex preservation capabilities, energy conservation, visual intricacy, and the computational time and memory costs involved in the simulations.

\paragraph{Impulse-modified APIC}
\label{sec:impulse_apic}
To showcase that the particle flow map mechanism plays an essential role, we implemented an impulse-modified APIC solver to conduct ablation tests. By configuring both $n^L$ and $n^S$ to 1, effectively causing the long-range map and short-range map to be reinitialized at every step, PFM transitions into a methodology closely resembling APIC \cite{jiang2015affine}, with the primary distinction being the substitution of impulse for velocity. We refer to this variation as \textit{impulse-modified APIC}. In this approach, particles at each step interpolate the grid's velocity field to acquire impulse and impulse gradients. Then, the \revv{backward map Jacobian $\mathcal{T}$}\revvdel{ and $\nabla \mathcal{T}$} is calculated similarly to the approach in PFM. These elements are utilized to evolve the impulse and impulse gradients by a single step during the particle advection process. Subsequently, the evolved values are conveyed to the grid via a particle-to-grid process mirroring PFM's. An evaluative comparison between impulse-modified APIC and PFM is detailed in the following experiments.

\revvdel{
\setlength{\abovecaptionskip}{12pt}
\begin{figure}[t]
 \centering
 \includegraphics[width=.48\textwidth]{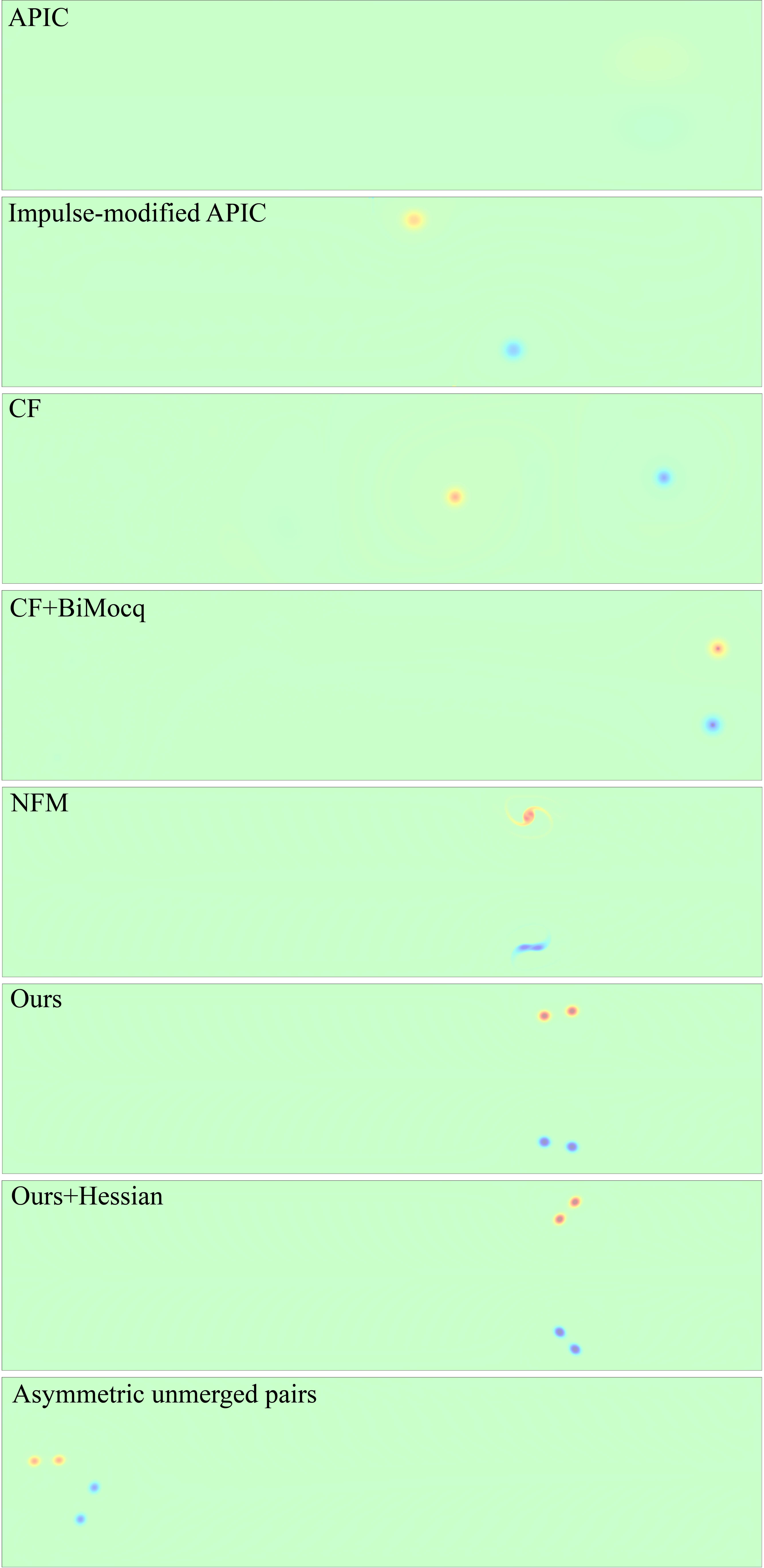}
 \caption{\revv{The top seven subfigures are snapshots capturing the state of 2D leapfrogging vortices at 234 seconds. The subfigure at the bottom shows two pairs of vortices in a 2D leapfrogging scenario that are asymmetric about the $y$-axis.}}
 \label{fig:2D_leapfrog_vort_compare}
\end{figure}
}

% The top seven subfigures are snapshots capturing the state of 2D leapfrogging vortices at 234 seconds. At this moment, the vortices in the APIC have almost vanished. Meanwhile, the vortices in the impulse-modified APIC, CF and CF+BiMocq have not disappeared but have been merged for some duration. In contrast, the vortices in the NFM are on the verge of merging, whereas the PFM successfully maintains two distinct pairs of vortices. The subfigure at the bottom shows two pairs of vortices in a 2D leapfrogging scenario that are asymmetric about the $y$-axis.

\revvdel{
\begin{figure}[t]
 \centering
 \includegraphics[width=.48\textwidth]{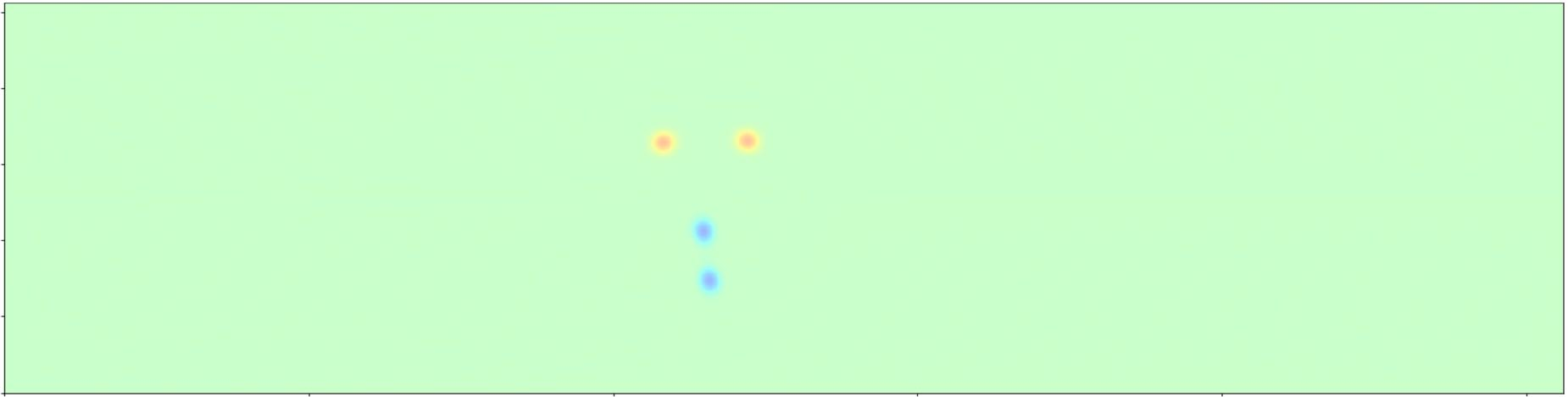}
 \caption{\junweirev{Two pairs of vortices in 2D leapfrogging scenario are asymmetric about $y$-axis.}}
 \label{fig:2D_leapfrog_vort_asymmetric}
\end{figure}
}

\begin{table}[t]
\caption{\junweirev{\revv{Simulation world time (frames $\times$ timestep)} of 2D leapfrogging vortices from the initial state to the end state. The observation reveals that \revv{APIC, impulse-modified APIC, CF, and CF+BiMocq} are unable to maintain the vortices over an extended period. On the other hand, the NFM exhibits improved performance, yet it remains less effective compared to the PFM.}}
\centering\small
\begin{tabularx}{0.5\textwidth}{ Y  Y }
\hlineB{3}
Method & Time (sec) \\
\hlineB{2.5}
APIC & 7.9 \\
\hlineB{1}
\revv{IM APIC} & 92.0 \\
\hlineB{1}
\revv{CF} & \revv{50.1} \\
\hlineB{1}
\revv{CF+BiMocq} & \revv{45.0} \\
\hlineB{1}
NFM & 234.0 \\
\hlineB{1}
Ours & 408.5 \\
\hlineB{1}
\revv{Ours+Hessian} & \revv{317.5} \\
\hlineB{3}

\end{tabularx}
\vspace{5pt}
% \caption{\junweirev{\revv{Simulation world time (frames $\times$ timestep)} of 2D leapfrogging vortices from the initial state to the end state. The observation reveals that \revv{APIC, impulse-modified APIC, CF, and CF+BiMocq} are unable to maintain the vortices over an extended period. On the other hand, the NFM exhibits improved performance, yet it remains less effective compared to the PFM.}}
\label{tab:2d_leapfrog_real_world_time}
\end{table}

\begin{figure}[t]
 \centering
 \includegraphics[width=.47\textwidth]{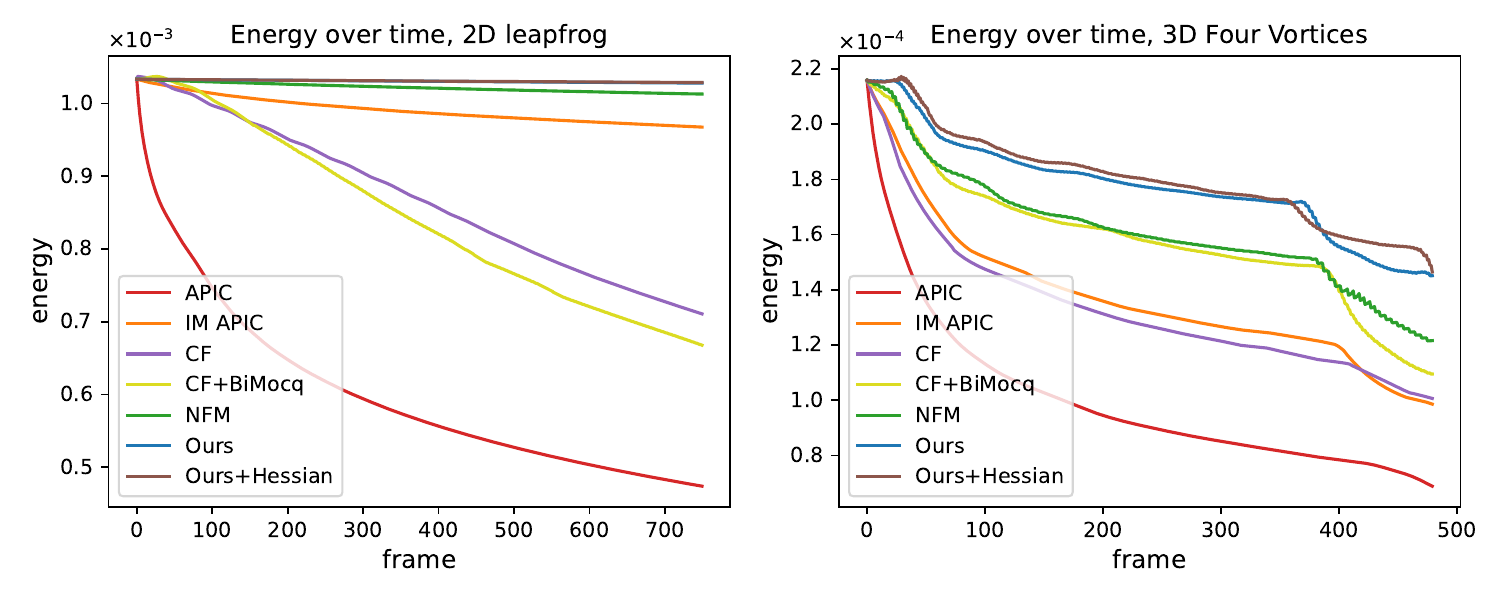}
 \caption{\junweirev{Time-varying energy of 2D leapfrogging vortices (left) and 3D four vortices collision (right). \revvdel{In both of these scenarios, the ranking of methods in terms of their ability to conserve energy, from the least effective to the most effective, is as follows: APIC, impulse-modified APIC, NFM, and PFM.}}}
 \label{fig:energy_compare}
 % \vspace{-0.2in}

\end{figure}

\paragraph{2D Analysis: Leapfrogging Vortices}
% \junwei{to-do experiment: energy measurement}

% As illustrated in Figure~\ref{fig:2D_leapfrog_vort_compare},
We establish the classic 2D leapfrogging vortex rings experiment. Initially, two negative and two positive vortices are placed on the left side of the domain. These vortices then move rightward and eventually return to their starting positions after colliding with the boundary. In an energy-conserving scenario, this cycle could continue indefinitely. However, in practical simulations, the vortices either merge into a single negative and a single positive vortex, or the two pairs become asymmetric along the $y$-axis. Both outcomes indicate numerical diffusion in the simulation. We timed the duration from the initial configuration to one of these end states, and the performance of various methods is summarized in Table~\ref{tab:2d_leapfrog_real_world_time}. The results reveal that the APIC method maintains the vortices for only a brief period before they dissipate, \revv{whereas impulse-modified APIC, CF, and CF+BiMocq perform slightly better than APIC.} The NFM method outperforms them but is still surpassed by PFM. This is consistent with the energy variation results presented in Figure~\ref{fig:energy_compare}, where PFM excels in energy preservation over the other methods.

Interestingly, although impulse-modified APIC also utilizes evolved impulse and impulse gradients on particles to reconstruct the velocity field, its effectiveness in maintaining vorticity is significantly inferior to that of PFM. This indicates that merely using impulse and impulse gradients is insufficient. For optimal vorticity retention, it is crucial to evolve these elements on particles using a flow map with adequate range, thereby enhancing numerical accuracy.

\setlength{\abovecaptionskip}{12pt}
\begin{figure}[t]
 \centering
 \includegraphics[width=.99\columnwidth]{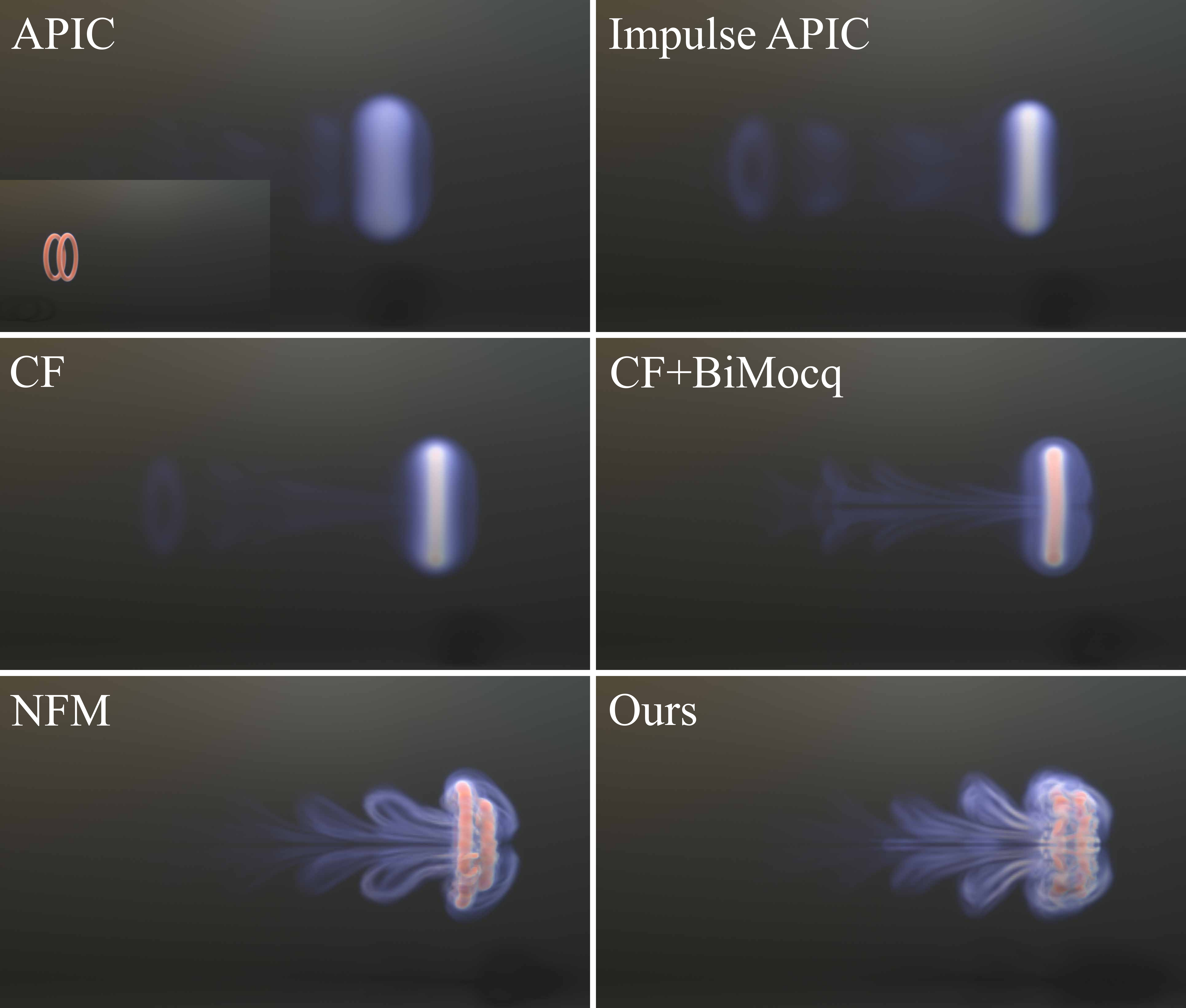}
 \caption{\revv{Comparison of 3D leapfrog vorticies. 
 Our method is able to maintain the separation of vortex rings for a range of time comparable with NFM, but other methods merge quickly. 
 The initial frame for all methods is located at the bottom-left corner of the APIC subfigure.}}
 \label{fig:3D_leapfrog_vort_compare}
 % \vspace{-0.2in}
\end{figure}

\paragraph{3D Analysis: Leapfrogging Vortices}
% \junwei{to-do experiment: energy measurement}

As illustrated in Figure~\ref{fig:3D_leapfrog_vort_compare}, we experiment with 3D leapfrogging vortex rings. Analogous to the 2D scenario, two rings in this experiment will perpetually leapfrog around each other in a conservative setting. \revv{In our findings, APIC, impulse-modified APIC, CF and CF+BiMocq manage to keep the two rings separate only up to the $3^\text{rd}$ leap. And both NFM and PFM successfully sustain to the $5^\text{th}$ leap.}

\setlength{\abovecaptionskip}{12pt}
\begin{figure*}[t]
 \centering
 \includegraphics[width=.985\textwidth]{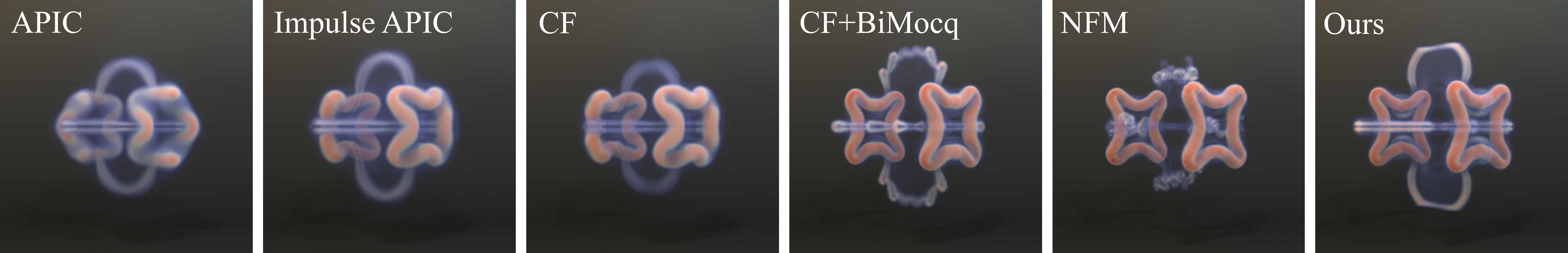}
 \caption{\revv{We show a time instance for the vorticity of four vortices colliding at a right angle. Our method maintains a comparable performance with NFM\revvdel{where the connecting bridge fades out as simulation goes on}.}}
 \label{fig:four_vort_compare}
\end{figure*}

\begin{figure*}
 \centering
 \includegraphics[width=.99\textwidth]{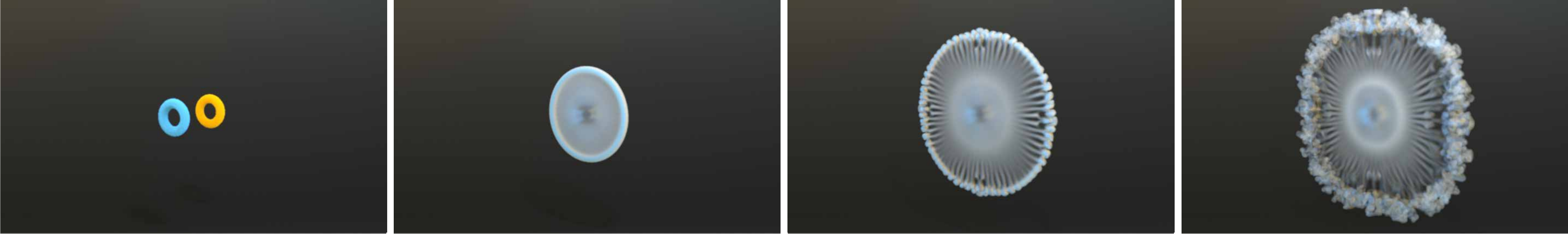}
 \caption{Simulation of two opposing vortex rings, apart on the $x$-axis. When they collide, they form a single ring that rapidly elongates in the $yz$-plane and contracts along the $x$-axis, leading to fragmentation into a circle of smaller secondary vortices.}
 \label{fig:headon}
\end{figure*}

\begin{figure*}
 \centering
 \includegraphics[width=.99\textwidth]{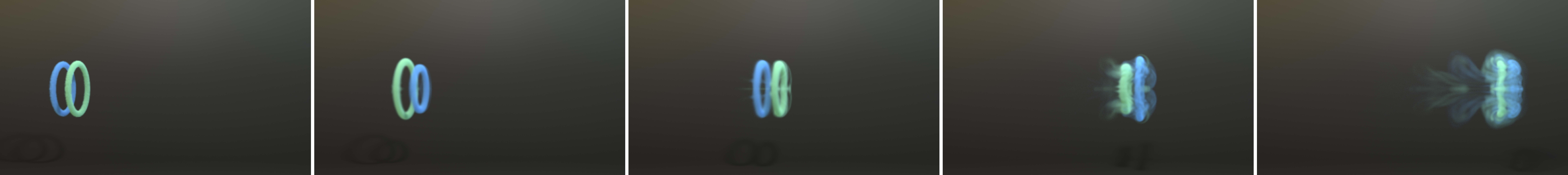}
 \caption{Smoke for 3D leapfrogging vortices. As shown in Section~\ref{sec:ablation_grad_rho}, our usage for $\nabla \rho_s$ with the correct advection equation provides the clearest result.}
 \label{fig:3D_leapfrog_smoke}
\end{figure*}

\paragraph{3D Analysis: Four Vortices Collision}
As depicted in Figure~\ref{fig:four_vorts}, we execute an experiment involving the collision of four vortices, during which the vortices merge and subsequently divide into two star-shaped vortices that drift in opposite directions until they encounter the boundary. We compare the simulation results of PFM to the other methods, as shown in Figure~\ref{fig:four_vort_compare}. The results reveal that all \revv{six} methods can generate the twin star-shaped vortices. However, post-collision, \revv{APIC, impulse-modified APIC, and CF} struggle to maintain vorticity, leading to dissipative, slower-moving vortices. In contrast, \revv{CF+BiMocq, NFM, and PFM} effectively preserve vorticity. This observation is consistent with the time-varying energy analysis presented in Figure~\ref{fig:energy_compare}, where \revv{APIC, impulse-modified APIC, and CF} exhibit rapid energy loss, whereas \revv{CF+BiMocq, NFM and PFM} demonstrate superior energy conservation, with PFM outperforming \revv{CF+BiMocq and NFM} in this regard.

% \setlength{\abovecaptionskip}{12pt}
% \begin{figure*}[t]
%  \centering
%  \includegraphics[width=.985\textwidth]{img/examples/3D_eight_vorts_vort_compare.pdf}
%  \caption{\revv{We show the comparison of collision between eight vortices. From this experiment, we show our method is the only one capable of recovering the 6 vortex tubes after the vortices collide with the walls and reflect. This \revv{manifests} the ability of our method to preserve vorticity and the correctness of our treatment with solid \revv{boundaries}.}}
%  \label{fig:eight_vort_compare}
% \end{figure*}

\paragraph{3D Analysis: Eight Vortices Collision}
As illustrated in Figure~\ref{fig:eight_vorts}, we undertake an experiment involving the collision of eight vortices. The amalgamation of these vortices forms six star-shaped vortices, which proceed along the positive and negative axes of three orthogonal dimensions. Upon impact with the cube-shaped boundary, the vortices disperse into several vortex tubes, tracing the boundary surface before eventually merging to reconstruct six distinct vortices. 
%The successful regeneration of these six vortices hinges on maintaining a highly symmetrical structure. According to the findings presented 
As shown in Figure~\ref{fig:eight_vort_compare}, \revv{APIC, impulse-modified APIC, and CF} exhibit excessive dissipation, failing not only to preserve the vortical structures but also to regenerate the six vortices. While \revv{CF+BiMocq and NFM} manages to maintain the tubular formations, it falls short in preserving the overall symmetry of the system. In contrast, PFM preserves the entire system's symmetry and effectively reconstructs the six-vortex formation.

\begin{table}[t]
\caption{Average simulation time and maximum GPU memory cost of our 2D and 3D simulation examples.}
\centering\small
\begin{tabularx}{0.5\textwidth}{Y  Y  Y  Y }
\hlineB{2}
Name & Method & Time (sec / step) & GPU Mem. (GB) \\
\hlineB{1.5}
\textbf{2D Leapfrog} & APIC & 0.24 & 1.21 \\
\textbf{2D Leapfrog} & \revv{IM APIC} & 0.25 & 1.41 \\
\textbf{2D Leapfrog} & NFM & 23.11 & 2.00 \\
\textbf{2D Leapfrog} & Ours & 0.47 & 1.41 \\
\hlineB{1}
\textbf{3D Leapfrog} & APIC & 4.92 & 3.51 \\
\textbf{3D Leapfrog} & \revv{IM APIC} & 4.69 & 7.91 \\
\textbf{3D Leapfrog} & NFM & 52.05 & 13.33 \\
\textbf{3D Leapfrog} & Ours & 4.75 & 8.21 \\
\hlineB{1}
\textbf{3D Four Vort} & APIC & 4.60 & 3.51 \\
\textbf{3D Four Vort} & \revv{IM APIC} & 4.26 & 7.91 \\
\textbf{3D Four Vort} & NFM & 50.71 & 13.34 \\
\textbf{3D Four Vort} & Ours & 4.20 & 8.21 \\
\hlineB{1}
\textbf{3D Eight Vort} & APIC & 2.04 & 2.31 \\
\textbf{3D Eight Vort} & \revv{IM APIC} & 1.82 & 4.41 \\
\textbf{3D Eight Vort} & NFM & 47.98 & 7.91 \\
\textbf{3D Eight Vort} & Ours & 1.95 & 4.61 \\
% \hlineB{1}
% \textbf{3D Oblique} & Ours & & \\
% \hlineB{1}
% \textbf{3D Headon} & Ours & & \\
% \hlineB{1}
% \textbf{3D Trefoil} & Ours & & \\
\hlineB{2}
\end{tabularx}
\vspace{5pt}
% \caption{Average simulation time and maximum GPU memory cost of our 2D and 3D simulation examples.}
\label{tab:time_memory}
% \vspace{-0.3in}

\end{table}

\paragraph{Time and Memory Cost Analysis}
As depicted in Table~\ref{tab:time_memory}, PFM significantly outperforms NFM in terms of execution speed, delivering results comparable to or surpassing those achieved by NFM. Specifically, PFM is 49.1, 10.9, 12.0, and 24.6 times faster than NFM in scenarios involving 2D leapfrogging vortices, 3D leapfrogging vortices, 3D four vortices collision, and 3D eight vortices collision, respectively. This considerable increase in speed primarily stems from eliminating both the backtrack substeps and the neural buffer training process present in NFM. Moreover, by obviating the need for a neural buffer, PFM also realizes substantial memory savings, registering reductions of 29.5\%, 38.4\%, 38.4\%, and 41.7\% in memory usage compared to NFM for these respective scenarios.
Compared to APIC and impulse-modified APIC, PFM maintains a similar execution time and memory footprint, yet it significantly enhances the simulation outcomes for these examples, as previously mentioned.

\subsection{Ablation Study}

\paragraph{Redistribute Particles}
\label{sec:redistribute_particle}
% \WrapFig{img/reinit_stragegy.pdf}{}{fig:redis_stragegy}{.25}{1.5}
% \begin{figure}
%  \centering
%  \includegraphics[width=.32\textwidth]{img/reinit_stragegy.pdf}
%  \caption{TODO}
%  \label{fig:redis_stragegy}
% \end{figure}

We conduct a 2D leapfrogging vortices analysis using three different particle redistribution strategies: uniform particle redistribution, random particle redistribution, and no particle redistribution. These approaches are examined under conditions where ($n^L$, $n^S$) assumes values of (20, 5), (20, 8), (20, 9), (20, 10), and (20, 15). The \revv{simulation world time} from the initial state to the end state of the leapfrogging vortices, utilizing the three mentioned strategies, is illustrated in Figure~\ref{fig:redis_stragegy}. The results indicate that both the no redistribution and random redistribution methods fall short of achieving satisfactory results, with random redistribution slightly outperforming the former. However, it is the uniform particle redistribution that delivers the most efficient performance, an approach that is notably employed by PFM.

\begin{figure}[t]
 \centering
 \includegraphics[width=.47\textwidth]{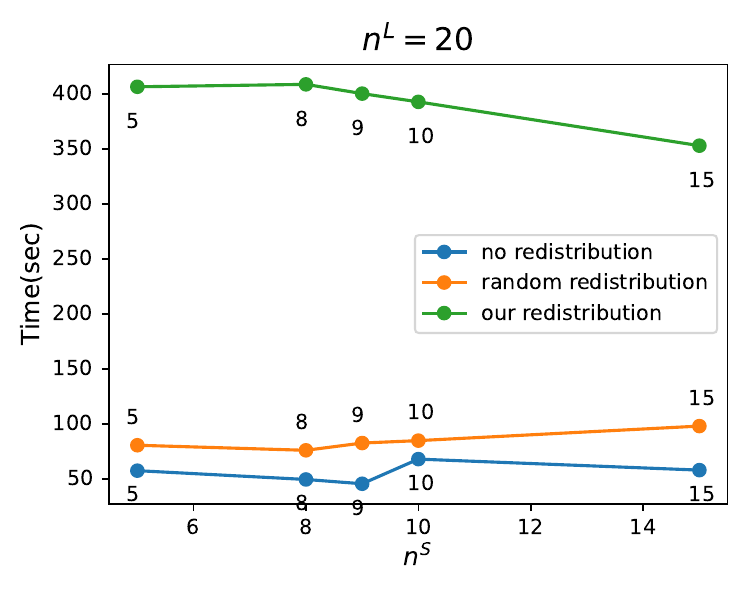}
 % \vspace{-0.3in}
 \caption{2D leapfrogging vortices' \revv{simulation world time} results of different particle redistribution strategies. The green line shows the results of the uniform redistribution strategy adopted in PFM, while the orange line and blue line show \revv{the} results of random redistribution and no \revv{redistribution}, respectively. The results show that uniform redistribution performs much better than the other two strategies.}
 \label{fig:redis_stragegy}
\end{figure}

% \begin{figure}[11]{r}{0.25\textwidth}
%     \centering
%     \vspace{-0.3cm}
%     \includegraphics[width=0.25\textwidth]{img/reinit_stragegy.pdf}
%     \vspace{-1.0cm}
%     \caption{2D leapfrogging vortices' real world time results of different particle distribution strategies. The green line shows the result of the uniform distribution strategy adopted in PFM, while the orange line and blue line shows yje results of random redistribution and no distribution, respectively. The results shows that uniform distribution performs much better than the other two strategies.}
%     \label{fig:redis_stragegy}
% \end{figure}

% \begin{figure}
% \centering
% \begin{minipage}{.33\linewidth}
%   \includegraphics[width=\linewidth]{img/reinit_stragegy.pdf}
%   \captionof{figure}{TODO}
%   \label{fig:reinit_step}
% \end{minipage}
% % \hspace{.01\linewidth}
% \begin{minipage}{.66\linewidth}
%   \includegraphics[width=\linewidth]{img/reinit_steps.pdf}
%   \captionof{figure}{TODO}
%   \label{fig:redis_stragegy}
% \end{minipage}
% \end{figure}

\paragraph{Reinitialization Steps}
\label{sec:reinit_steps}
% \begin{figure}
%  \centering
%  \includegraphics[width=.47\textwidth]{img/reinit_steps.pdf}
%  \caption{2D leapfrogging vortices' \revv{simulation world time} results of various reinitialization steps. The left side corresponds to a fixed $n^L=20$, and the right side corresponds to a fixed $n^S=8$.}
%  \label{fig:reinit_step}
% \end{figure}

% \begin{figure}
%  \centering
%  \includegraphics[width=.47\textwidth]{img/validations/3D_leapfrog_energy_ablation.pdf}
%  \caption{\revv{3D leapfrogging vortices' time-varying energy results of various reinitialization steps. The left side corresponds to a fixed $n^L=20$ and varying $n^S$, and the right corresponds to a fixed $n^S=1$ and varying $n^L$.}}
%  \label{fig:reinit_step_3D}
% \end{figure}
We examine the impact of reinitialization intervals $n^L$ and $n^S$ by conducting \revv{both 2D and 3D leapfrogging vortices experiments. Specifically, for the 2D scenario,} we set $n^S$ to values of 1, 4, 5, 6, 7, 8, 9, 10, 11, 12, 15, and 20 when $n^L$ is fixed at 20, and we alter $n^L$ to 8, 10, 15, 18, 20, 25, 30, 35, and 40 when $n^S$ is fixed at 8. \revv{For 3D scenario, we set $n^S$ to values of 1, 2, 3, 4, 5, 6 and 7 when $n^L$ is fixed at 20, and we alter $n^L$ to 5, 8, 12, 15, 18, 20, 22 and 25 when $n^S$ is fixed at 1.} The outcomes are depicted in Figure~\ref{fig:reinit_step} and~\ref{fig:reinit_step_3D} for 2D and 3D experiments, respectively. The results reveal that for a constant $n^L$, a smaller $n^S$ typically yields better results. However, as $n^S$ approaches the value of $n^L$, the efficiency of PFM diminishes rapidly. \revv{This confirms the importance of using both short-range and long-range maps in PFM instead of two flow maps of identical lengths.} Conversely, when $n^S$ is held constant, both excessively large and small values of $n^L$ lead to suboptimal performance. 
% An $n^L$ that is approximately 2.5 to 4 times as large as $n^S$ is observed to deliver optimal results.

% \paragraph{Computation of Backward Map Hessian}
% \label{sec:analysis_grad_T}
% \bo{TODO!}

\begin{figure}
 \centering
 \includegraphics[width=.47\textwidth]{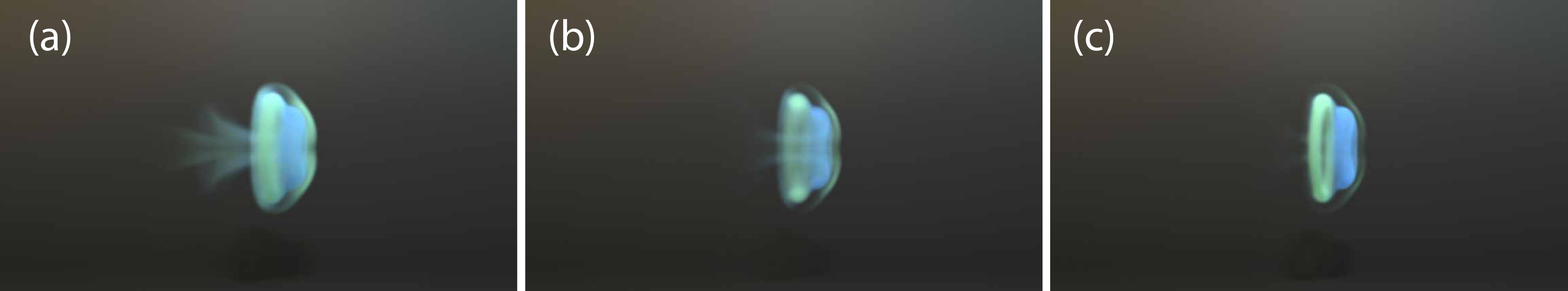}
 \caption{In this comparison, we show three leapfrog simulations of smoke but with different treatments for $\nabla \rho_s$. (a) show the simulation without using $\nabla \rho_s$; (b) shows the simulation if we just \revv{store} $\nabla \rho_s$ at every reinitialization step without evolving it during simulation; (c) shows our current implementation, which is using the evolution of $\nabla \rho_s$.}
 \label{fig:smoke_ablation}
\end{figure}

\paragraph{Evolution of Smoke Density Gradients}
\label{sec:ablation_grad_rho}
In our method, we employ a particle-to-grid strategy that not only facilitates the transport of vector quantities such as impulse, transferring both the quantity itself and its evolved gradients to the grid, but also extends this methodology to the transport of scalar quantities, like smoke density \revv{$\rho_s$}. 
This section highlights the \revvdel{pivotal}importance of accurately evolving $\nabla \rho_s$ with the flow map and effectively transferring it from the particles to the grid, especially in smoke advection.
In our approach, we store both the smoke density $\rho_s$ and its gradients $\nabla \rho_s$ on particles, with the gradients evolving using the flow map. Subsequently, we employ a strategy akin to that used for impulse to transfer these values to the grid. 
% \revvdel{In particular, we have the following theorem.
% \begin{theorem}
% Assuming $\rho_s(\bm x,t)$ is a scalar function field and satisfies $\rho_s(\bm x,t) = \rho_s[\bm \psi(\bm x, t),0]$, the following relation holds:
% \begin{equation}
% \bm \nabla \rho_s(\bm x, t) = \mathcal{T}(\bm x, t)^T \bm \nabla \rho_s(\bm \psi, 0)
% \label{eq:rhosxt}
% \end{equation}
% \begin{proof}
% We can directly derive \eqref{eq:rhosxt} using the chain rule.
% \begin{equation}
% \begin{aligned}
% &\bm \nabla \rho_s(\bm x, t) = \frac{\partial \bm \nabla \rho_s(\bm \psi,0)}{\partial \bm x} \\
% =& \frac{\partial \bm \nabla \rho_s(\bm \psi,0)}{\partial \psi_i}\frac{\partial \psi_i(\bm x,t)}{\partial \bm x} = \bm \nabla \rho_s (\bm \psi,0)\frac{\partial \psi_i(\bm x,t)}{\partial \bm x}.
% \end{aligned}
% \end{equation}
% \end{proof}
% \end{theorem}
% }
We show in Figure~\ref{fig:smoke_ablation} that not transferring $\nabla \rho_s$ to the grid or not evolving $\nabla \rho_s$ will cause smoothing behavior in smoke simulation. The first subfigure shows the smoke advected without $\nabla \rho_s$, and the smoke gets blurred quickly. The middle subfigure shows that if we keep $\nabla \rho_s$ for reinitialization without evolving it, \revv{smoke} also gets blurred (but slightly better). Only by evolving $\nabla \rho_s$ and transferring both $\rho_s$ and $\nabla \rho_s$ can we \revv{obtain} the optimal result.

\begin{figure}
 \centering
 \includegraphics[width=.47\textwidth]{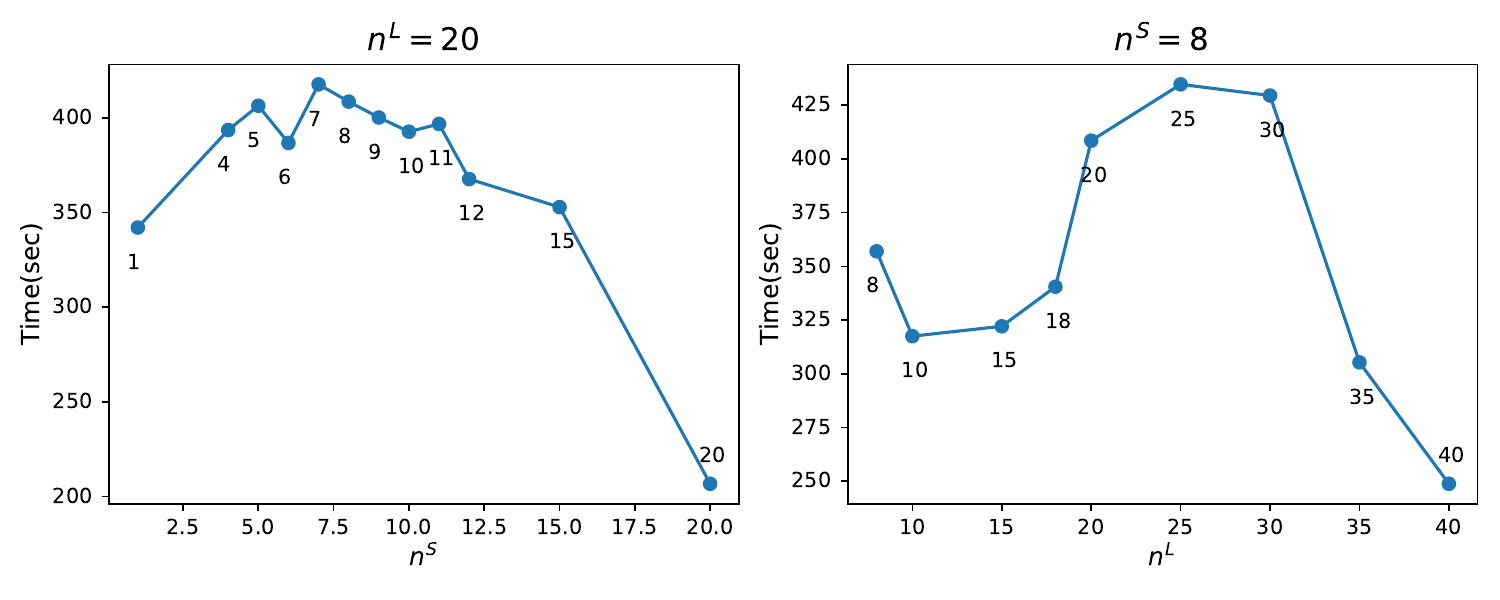}
 \caption{2D leapfrogging vortices' \revv{simulation world time} results of various reinitialization steps. The left side corresponds to a fixed $n^L=20$ and varying $n^S$, and the right side corresponds to a fixed $n^S=8$ and varying $n^L$.}
 \label{fig:reinit_step}
\end{figure}

\begin{figure}
 \centering
 \includegraphics[width=.47\textwidth]{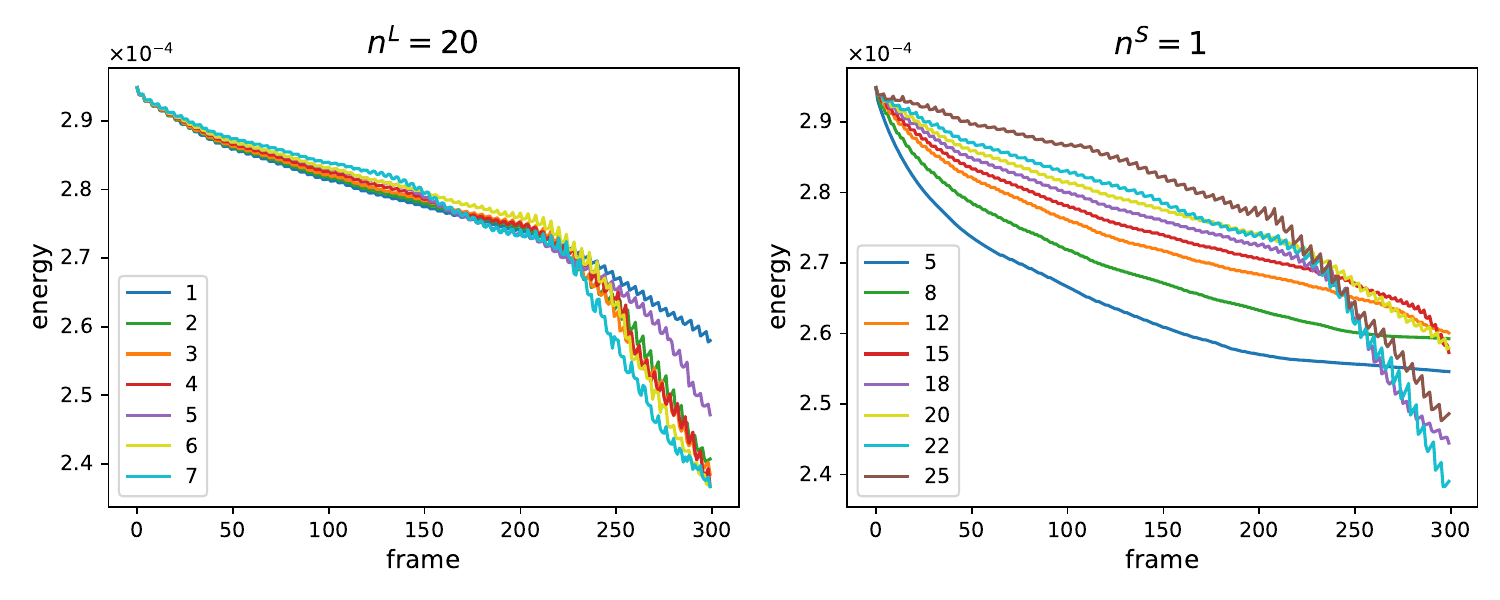}
 \caption{\revv{3D leapfrogging vortices' time-varying energy results of various reinitialization steps. The left side corresponds to a fixed $n^L=20$ and varying $n^S$, and the right side corresponds to a fixed $n^S=1$ and varying $n^L$.}}
 \label{fig:reinit_step_3D}
\end{figure}

\section{Examples}
In this section, we present various intricate simulation examples to demonstrate the efficacy of our PFM simulator. A comprehensive list of these examples and their respective configurations is available in Table~\ref{tab: examples_table}. It is assumed that the shortest edge of the simulation domain is of unit length. \revv{We refer readers to our supplemented video for detailed simulation results.} All simulations were performed on workstations equipped with an AMD Ryzen Threadripper 5990X and \revv{an} NVIDIA RTX 3080/A6000.

\subsection{Vortical Flow}
\revv{\paragraph{Taylor Vortex}
We adapt the setting from \cite{mckenzie2007hola} where the simulation domain is $[-\pi, \pi]$ and two Taylor vortices separated by $0.8$. We show the separation process in Figure~\ref{fig:taylor_karman}}.
\paragraph{Leapfrogging Vortices (2D)}
Table~\ref{tab:2d_leapfrog_real_world_time} shows our 2D leapfrog experiment where we position four point vortices with equal magnitudes of 0.005 at $x$ = 0.0625 and $y$ values of 0.26, 0.38, 0.62, and 0.74, with the upper two negative and the lower two positive. These vortices are modeled using a mollified Biot-Savart kernel with a support of 0.02. \revv{In the experiment, our method accurately maintains vortex structure for 408.5 seconds.}
% \setlength{\abovecaptionskip}{12pt}
% \begin{figure}[t]
%  \centering
%  \includegraphics[width=.99\columnwidth]{img/examples/four_vort_compare.pdf}
%  \caption{We show a time instance for vorticity of four vortices colliding at a right angle. Our method maintains a comparable performance with NFM where the connecting bridge fades out as simulation goes on.}
%  \label{fig:four_vort_compare}
% \end{figure}
\paragraph{Leapfrogging Vortices (3D)}
In Figure~\ref{fig:3D_leapfrog_smoke}, the initial setup involves two vortex rings aligned parallelly, positioned at $x$ = 0.16 and 0.29125. These rings have a major radius of 0.21 and a minor radius (the mollification support of the vortices) of 0.0168. Our approach maintains separation of the vortex rings beyond the $5^\text{th}$ leap.

\paragraph{Oblique Vortex Collision}
\revv{In Figure~\ref{fig:oblique}, we detail an experiment with two perpendicular vortex rings, offset by 0.3 units along the $x$-axis, each having a major radius of 0.13 and a minor radius of 0.02. }
% In Figure~\ref{fig:oblique}, we present an experiment featuring two vortex rings oriented perpendicularly to each other. These rings are positioned with their centers offset by 0.3 units along the $x$-axis. The dimensions of each ring include a major radius of 0.13 and a minor radius of 0.02. During the experiment, as the rings intersect, they amalgamate into a single, larger vortex ring. This combined ring then moves towards the right side, where it eventually divides into three separate, smaller vortices.  

\paragraph{Headon Vortex Collision}
\revv{Figure~\ref{fig:headon} shows an experiment with two opposing vortex rings, separated by 0.3 units on the $x$-axis, each with a major radius of 0.065 and a minor radius of 0.016. When these rings collide, they form a single ring that elongates in the $yz$-plane and contracts along the $x$-axis, leading to destabilization and fragmentation into secondary vortices \cite{lim1992instability}.}
% Figure~\ref{fig:headon} illustrates our experiment with two vortex rings placed in opposition, separated by 0.3 units along the $x$-axis. These rings have a major radius of 0.065 and a minor radius of 0.016. Upon collision, these rings merge into a new single ring. This formation undergoes rapid elongation in the $yz$-plane and simultaneous contraction along the $x$-axis. This dynamic leads to the destabilization of the structure, resulting in its fragmentation into a circular array of smaller, radially-oriented secondary vortices. This phenomenon closely mirrors the experimental observations reported by \citet{lim1992instability}.

\paragraph{Trefoil Knot}
\revv{Figure~\ref{fig:trefoil} replicates the trefoil knot, using an initialization file from \citet{nabizadeh2022covector}, based on \citet{kleckner2013creation}. Our simulation result matches the experiments.}
% Figure~\ref{fig:trefoil} demonstrates our replication of the classic trefoil knot configuration originally presented by \citet{kleckner2013creation}. For our simulation, we utilize the initialization file open-sourced by \citet{nabizadeh2022covector}. Initially, the knot progresses towards the right while simultaneously rotating. During this movement, collisions and reconnections occur in adjacent regions, leading to the disintegration of the knot into two separate, unlinked rings. One of these rings is larger and moves more slowly, while the other is smaller but travels at a higher speed. This behavior closely matches the outcomes observed in the referenced experiments.

\setlength{\abovecaptionskip}{12pt}
\begin{figure}[t]
 \centering
 \includegraphics[width=.99\columnwidth]{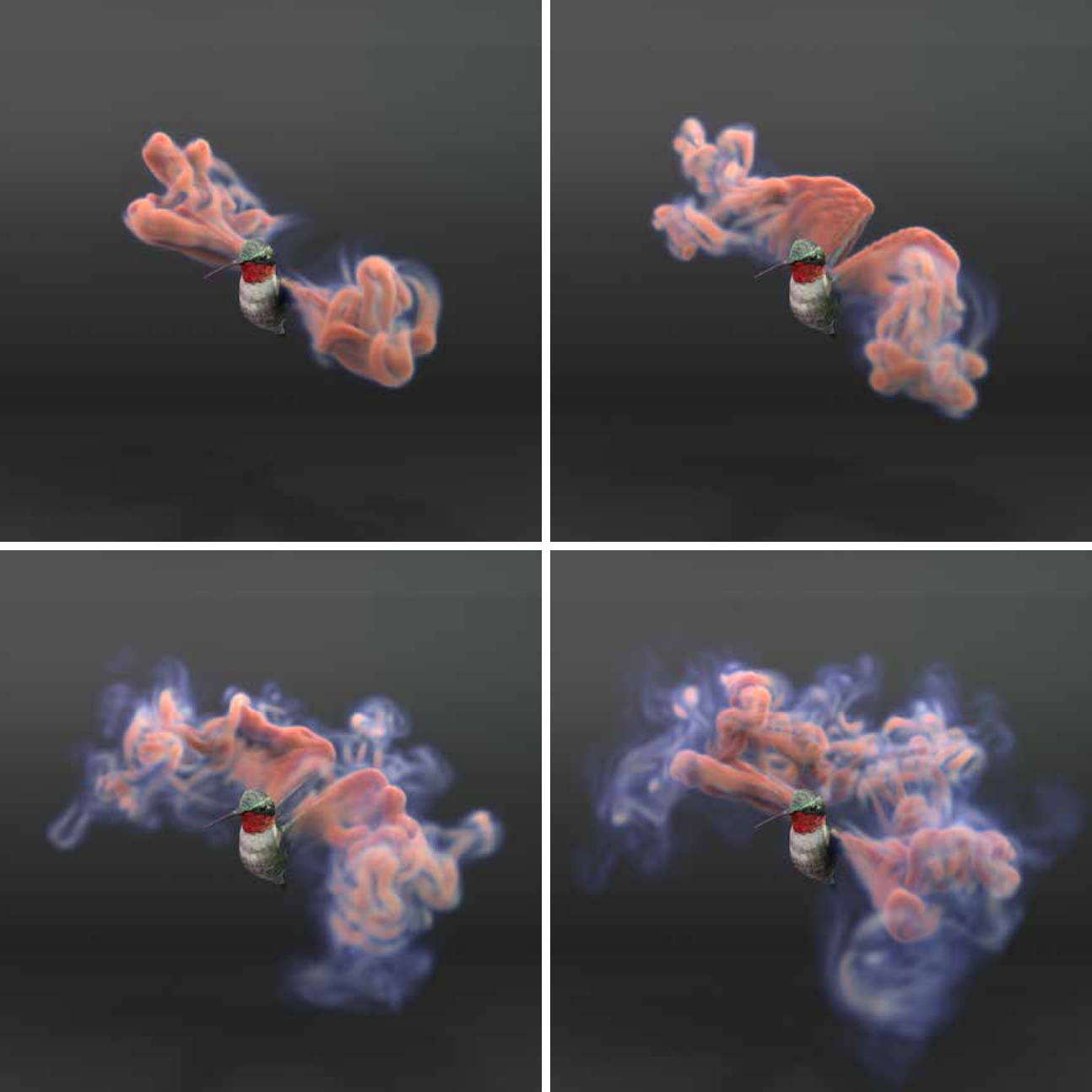}
 \caption{These images depict a bird in the act of flapping its wings, a behavior that generates visually stunning and complex vortices. These vortices interact with each other, forming intricate shapes.}
 \label{fig:bird}
\end{figure}

\paragraph{Four Vortices Collision}
\revv{Figure~\ref{fig:four_vorts} depicts an experiment based on \citet{Matsuzawa2022VideoTT} where four vortex rings, arranged in a square pattern, collide in the $yz$-plane. These rings, each with a major radius of 0.15 and a minor radius of 0.024 merge into two four-pointed star-shaped vortices until colliding with the boundaries. We compare our result with previous methods in Figure~\ref{fig:four_vort_compare}.}
% In Figure~\ref{fig:four_vorts}, we follow the setup described by \citet{Matsuzawa2022VideoTT} to initiate a scenario involving four colliding vortices. These rings are arranged to form right angles with their neighbors, effectively creating a square pattern in the $yz$-plane. The rings have a major radius of 0.15 and a minor radius of 0.024.
% Upon colliding, these four vortex rings undergo a transformation, merging into two vortices with a four-pointed star shape. This interaction results in complex, turbulent trails. The two newly formed vortex structures then move towards the left and right boundaries, constantly altering their form. They shift from a star with points oriented towards the y and z axes to a star with points angled at 45 degrees to these axes, and then back again, until they collide with the boundary. This dynamic repeats until they collide with the boundaries. Upon this collision, the two vortices disintegrate into four separate vortex tubes. We show the comparison between our result with previous methods in Fig.~\ref{fig:four_vort_compare}.

\begin{figure}[t]
 \centering
 \includegraphics[width=.99\columnwidth]{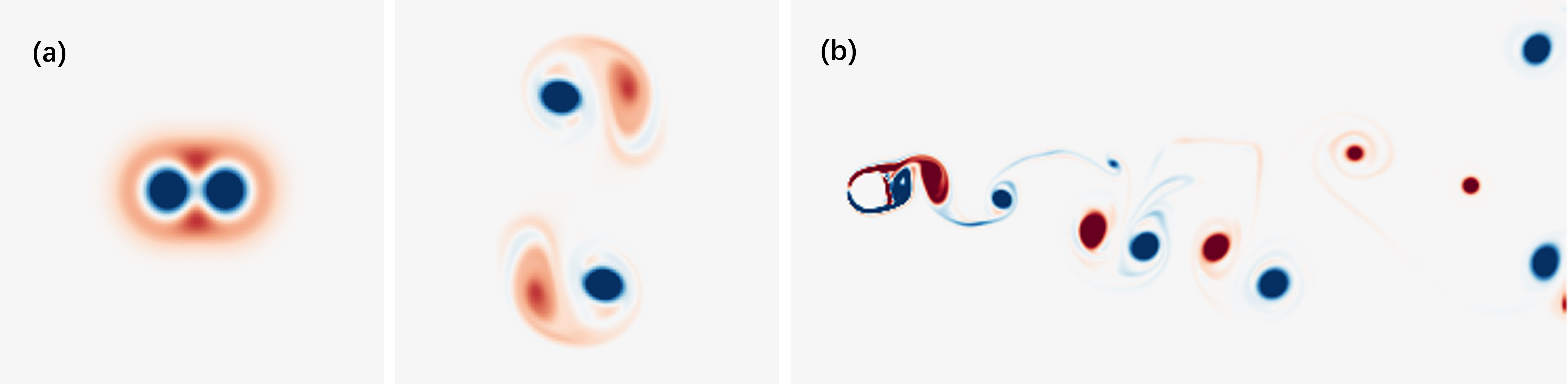}
 \caption{(a) shows taylor vortex splitting and (b) shows Karman Vortex Street at $t=60s$ with no viscosity.}
 \label{fig:taylor_karman}
\end{figure}

\paragraph{Eight Vortices Collision}
\revv{Figure~\ref{fig:eight_vorts}, based on \citet{Matsuzawa2022VideoTT}, shows an experiment with eight vortices forming a regular octahedron, each angled at 35.26° relative to the $z$-axis. These rings, with major radius of 0.08 and minor radius of 0.024, collide and reconnect into six rings within a cubic boundary. Figure~\ref{fig:eight_vort_compare} shows our method uniquely recreating this dynamic.}

% In Figure~\ref{fig:eight_vorts}, our experiment, following \citet{Matsuzawa2022VideoTT}, involves eight vortices arranged in a way that their respective planes form a regular octahedron. This means each ring is angled at 35.26° ($\arctan(\frac{\sqrt{2}}{2})$) relative to the $z$-axis. The rings have a major radius of 0.08 and a minor radius of 0.024. The simulation space is bounded by a cube. After colliding, the eight rings reconnect into six rings, which then move towards the positive and negative directions of three mutually perpendicular axes until they hit the boundaries. Subsequently, each ring splits into four separate vortex tubes. These tubes move towards the corners of the boundary cube's faces, colliding at the cube's corners with tubes from two adjacent faces. After these three-way collisions, the tubes split into three and move along the edges of the cube, eventually merging with a counter-moving tube at the midpoint of cube edge and then splitting into two tubes. Then, they head towards the center of their respective cube face, where they collide with three other tubes. This collision results in the formation of six rings again, which converge towards the center of the cube. Finally, the collision of these six rings results in their division into eight distinct parts. In Fig.~\ref{fig:eight_vort_compare}, we illustrate that our method is the only one capable of recreate this procedure and recreates six rings in the end of the simulation.

\begin{figure*}
\centering
\begin{minipage}{.49\linewidth}
  \includegraphics[width=\linewidth]{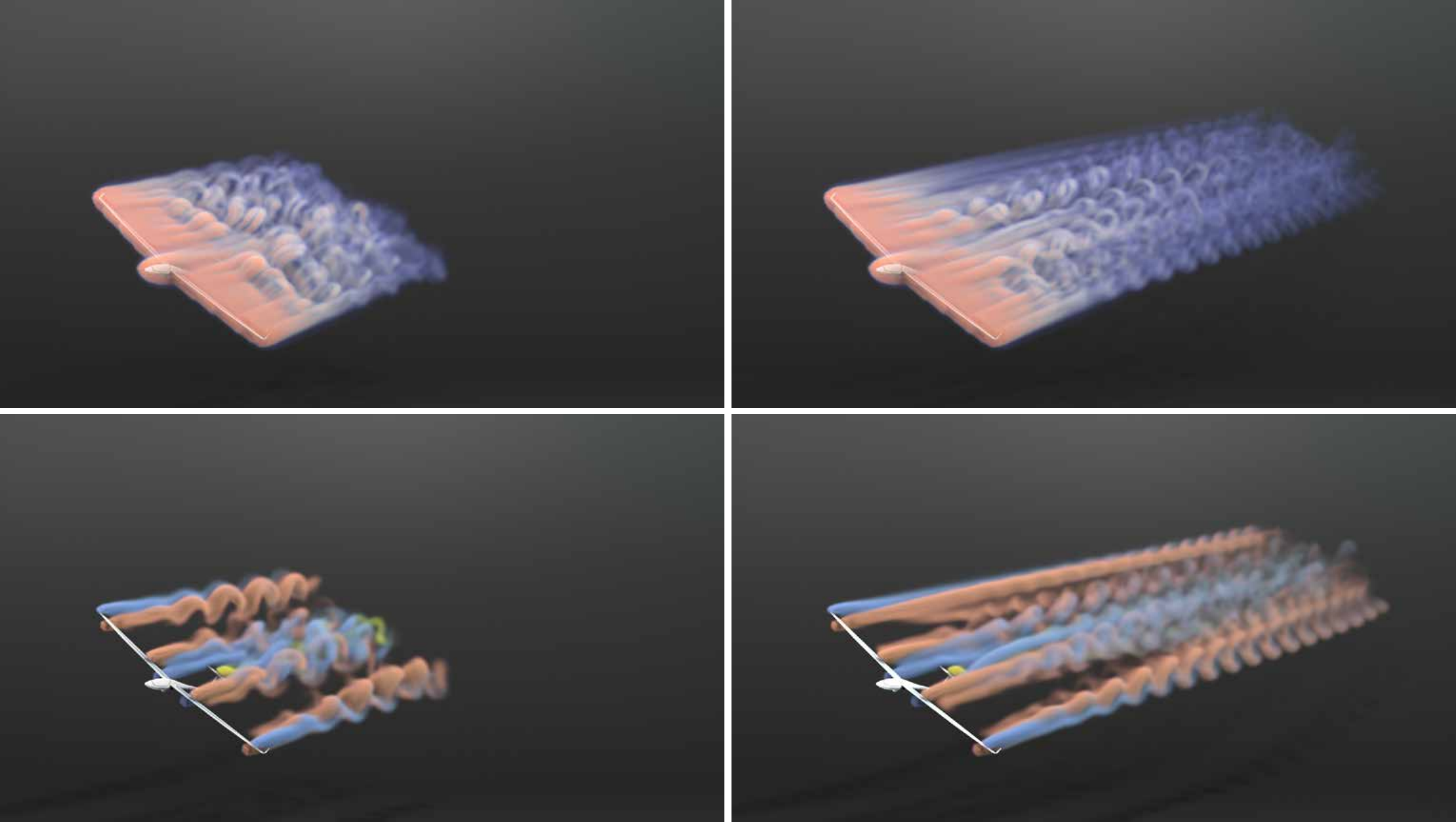}
  \captionof{figure}{These images show the formation of spiral-shaped vortices on the outer sides of a flat-wing aircraft under the influence of airflow coming at an angle with the plane. This phenomenon is identical to the contrails we observe in the sky.}
  \label{fig:plane_flat}
\end{minipage}
\hspace{.01\linewidth}
\begin{minipage}{.49\linewidth}
  \includegraphics[width=\linewidth]{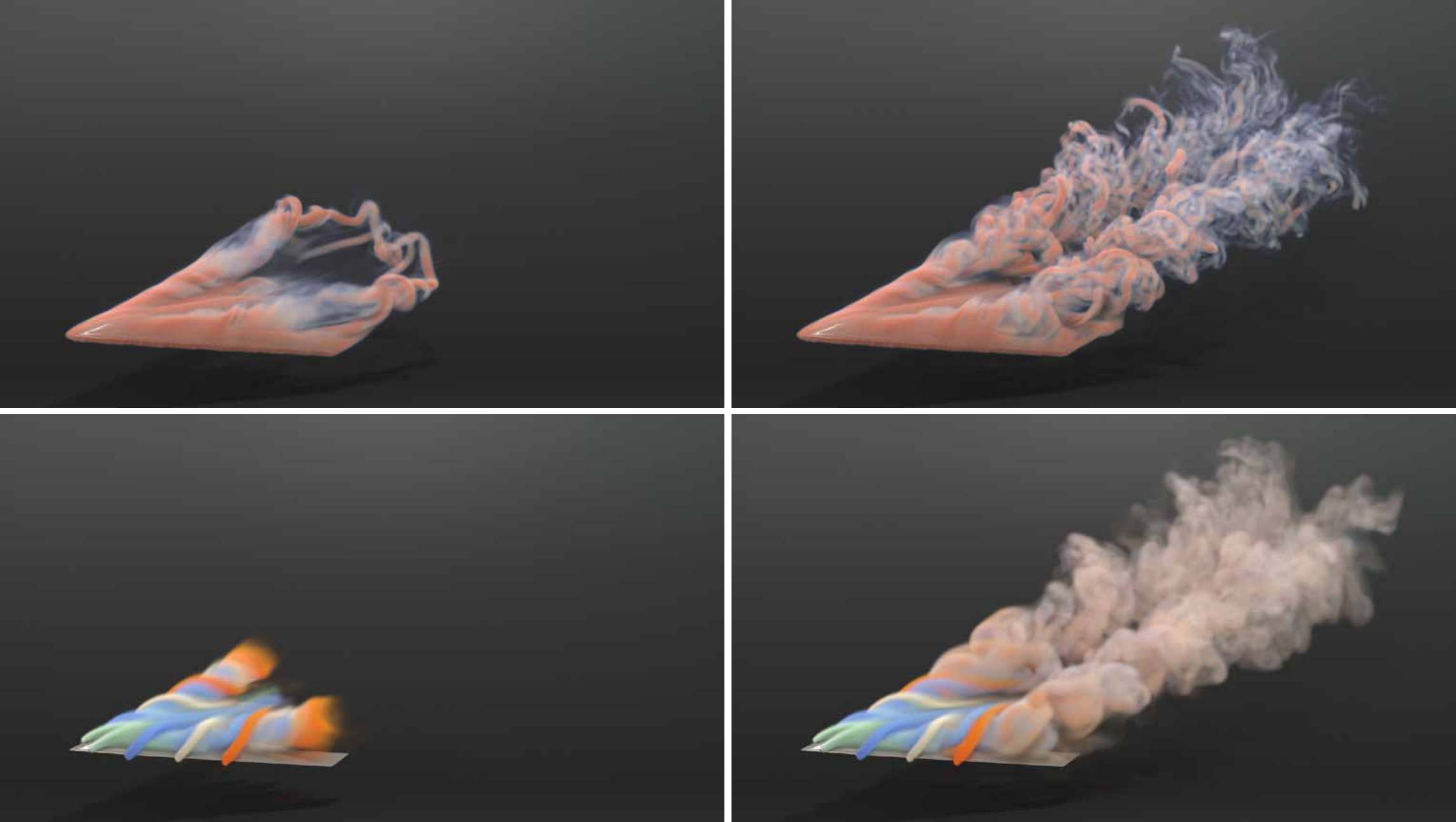}
  \captionof{figure}{The pictures depicts a delta wing with a sweep angle exposed to airflow at a angle of attack, and the "vortex lift" phenomenon becomes observable. This observation is consistent with the physical experiments reported in \cite{delery2001robert}.}
  \label{fig:plane_delta}
\end{minipage}
\end{figure*}

\subsection{Solid Boundaries}
% \duowen{
We employ a triangle mesh for solid-driven flow applications, either static or animated by a preset skeleton. The forward difference between two consecutive timesteps determines the velocity at each mesh vertex. We adopt an interpolation method, as described in \cite{robinson2008two}, to transfer velocity from mesh vertices to a grid. The weight of this transfer is based on the intersection area between the grid cell and the triangle patch. This velocity is then used as a solid boundary condition to drive the fluid simulation.
% }
\revv{\paragraph{Karman Vortex Street}
In Figure~\ref{fig:taylor_karman}, we place a cylinder at [0.2, 0.5] with the radius of 0.05. Due to lack of viscosity, incoming flow creates more turbulent vortices, as shown in \cite{nabizadeh2022covector}.
}
\paragraph{Lifting Airplane}
\revv{In Figures \ref{fig:plane_flat} and \ref{fig:plane_delta}, we showcase planes with different lifting angles facing incoming flow from the front boundary. Both delta wing and flat wing configurations are depicted. The trajectories observed in our smoke visualization model closely mimic real-world airplane condensation trails.}
% \duowen{
\paragraph{Flapping Bird, Swimming Fish, and Rotating Propeller} Figure~\ref{fig:bird}, \ref{fig:fish} and \ref{fig:propeller} illustrate our method coupling with the skeleton driven meshes. Skeletons are built in Blender in each example and set with periodic movements such as continuous rotations. Bezier interpolations between keyframes are performed to ensure smoothness.
% }
% \begin{figure*}
% \centering
% \begin{minipage}{.49\linewidth}
%   \includegraphics[width=\linewidth]{img/examples/plane_flat_smoke_vort.pdf}
%   \captionof{figure}{These images show the formation of spiral-shaped vortices on the outer sides of a flat-wing aircraft under the influence of airflow coming at an angle with the plane. This phenomenon is identical to the contrails we observe in the sky.}
%   \label{fig:plane_flat}
% \end{minipage}
% \hspace{.01\linewidth}
% \begin{minipage}{.49\linewidth}
%   \includegraphics[width=\linewidth]{img/examples/plane_delta_smoke_vort.pdf}
%   \captionof{figure}{The pictures depicts a delta wing with a sweep angle is exposed to airflow at a angle of attack, the "vortex lift" phenomenon, as mentioned in \cite{anderson2010aircraft}, becomes observable. This observation is consistent with the physical experiments reported in \cite{delery2001robert}.}
%   \label{fig:plane_delta}
% \end{minipage}
% \end{figure*}

\begin{table*}[t]
\caption{The catalog of all our 2D and 3D simulation examples. $n^L$ and $n^S$ are the reinitialization steps of long-range and short-range maps, respectively, and \#Particles is the particle count per cell at reinitialization step.}
\centering\small
\begin{tabularx}{\textwidth}{Y | Y | Y | Y | Y | Y | Y}
\hlineB{3}
Name & Figure\revv{/Table} & Resolution & CFL & $n^L$ & $n^S$ & \#Particles \\
\hlineB{2.5}
2D Leapfrog & Table~\ref{tab:2d_leapfrog_real_world_time} & 1024 $\times$ 256 & 1.0 & 20 & 8 & \revv{16} \\
\hlineB{2}
\revv{2D Taylor Vortex} & Figure~\ref{fig:taylor_karman} (a) & 128 $\times$ 128 & 1.0 & 20 & 8 & 16 \\
\hlineB{2}
\revv{2D Karman Vortex} & Figure~\ref{fig:taylor_karman} (b) & 512 $\times$ 256 & 0.5 & 20 & 8 & 16 \\
\hlineB{2}
3D Leapfrog & Figure~\ref{fig:3D_leapfrog_smoke} & 256 $\times$ 128 $\times$ 128 & 0.5 & \revv{20} & 1 & \revv{8} \\
\hlineB{2}
3D Oblique & Figure~\ref{fig:oblique} & 128 $\times$ 128 $\times$ 128 & 0.5 & 12 & 4 & \revv{8} \\
\hlineB{2}
3D Headon & Figure~\ref{fig:headon} & 128 $\times$ 256 $\times$ 256 & 0.5 & 12 & 4 & \revv{8} \\
\hlineB{2}
3D Trefoil & Figure~\ref{fig:trefoil} & 128 $\times$ 128 $\times$ 128 & 0.5 & 12 & 1 & \revv{8} \\
\hlineB{2}
3D Four Vortices & Figure~\ref{fig:four_vorts} & 128 $\times$ 128 $\times$ 256 & 0.5 & 12 & 4 & \revv{8} \\
\hlineB{2}
3D Eight Vortices & Figure~\ref{fig:eight_vorts} & 128 $\times$ 128 $\times$ 128 & 0.5 & 12 & 4 & \revv{8} \\
\hlineB{2}
Flat Wing & Figure~\ref{fig:plane_flat} & 384 $\times$ 128 $\times$ 128 & 0.5 & 12 & 4 & \revv{8} \\
\hlineB{2}
Delta Wing & Figure~\ref{fig:plane_delta} & 384 $\times$ 128 $\times$ 128 & 0.5 & 12 & 4 & \revv{8} \\
\hlineB{2}
Flapping Bird & Figure~\ref{fig:bird} & 128 $\times$ 128 $\times$ 128 & 0.5 & 12 & 4 & \revv{8} \\
\hlineB{2}
Swimming Fish & Figure~\ref{fig:fish} & 256 $\times$ 128 $\times$ 128 & 0.5 & 12 & 4 & \revv{8} \\
\hlineB{2}
Rotating Propeller & Figure~\ref{fig:propeller} & 256 $\times$ 128 $\times$ 128 & 0.5 & 12 & 4 & \revv{8} \\
\hlineB{3}
\end{tabularx}
\vspace{5pt}
% \caption{The catalog of all our 2D and 3D simulation examples. $n^L$ and $n^S$ are the reinitialization steps of long-range and short-range maps, respectively, and \#Particles is the particle count per cell at reinitialization step.}
\label{tab: examples_table}
% \vspace{-0.3in}

\end{table*}

\section{Discussion and Conclusion}

\setlength{\abovecaptionskip}{12pt}
In this paper, we present a particle flow map method (PFM) that achieves state-of-the-art advection fidelity and simulation accuracy in terms of energy conservation, vorticity preservation, and visual complexity, through its novel bridging of impulse fluid mechanics, long-range flow map, and hybrid Eulerian-Lagrangian simulation. Compared to the accurate but expensive NFM method, our PFM method drastically reduces the computational cost without sacrificing the physical fidelity, as it leverages the mathematical and numerical insight that a set of forward-evolving equations can be derived and naturally solved on particles to replace the costly backtracing procedure. To realize this idea to its full potential, we carefully devise a two-scale flow map representation along with a novel particle-to-grid transfer scheme to devise a hybrid solver whose accuracy and versatility are thoroughly verified through a set of challenging numerical experiments, including leapfrogging vortices, vortex tube reconnections, and turbulent flows. 

\paragraph{Connection to NFM} The goal for both PFM and NFM is to accurately compute impulse $\bm m$ on the grid, but they present two different numerical perspectives: NFM traces ``virtual particles'' whose final positions coincide with the grid points, but their initial positions are inside grid cells. Hence, \textit{a grid-to-particle interpolation scheme is required at the initial time}. PFM maintains forward simulating particles whose final positions are inside cells, but their initial values are known from (re)initialization. Hence, \textit{a particle-to-grid transfer scheme is required at the final time}. As a result, both perspectives require exactly one potentially lossy communication between particles and grids for each simulation step. To this end, NFM uses a naive linear interpolation scheme and compensates for the interpolation error with BFECC \cite{kim2006advections}. PFM, on the other hand, devises an advanced particle-to-grid transfer scheme that largely reduces the error. Thus, for accuracy, both methods offer comparable levels of performance, which are validated in our experiments. However, from the efficiency standpoint, PFM is by far more desirable due to the elimination of the costly backtracing process, as it can be up to $49\times$ faster and $41\%$ more compact than NFM. 
% By tracking the long-range bidirectional flow map directly on particles, we maintain a low memory cost with high computational speed among all existing impulse fluid model that can achieve state-of-art performance. Our two-scale flow map scheme offers advantages in getting different levels of accuracy for flow quantities and accommodates different reinitilaization requirements natually.
% Together with proper interpolation scheme to transfer that to grid when needed, we device the Eulearian-Lagrangian solver for impulse fluid model using flow map for advection. 

\begin{figure}[t]
 \centering
 \includegraphics[width=.99\columnwidth]{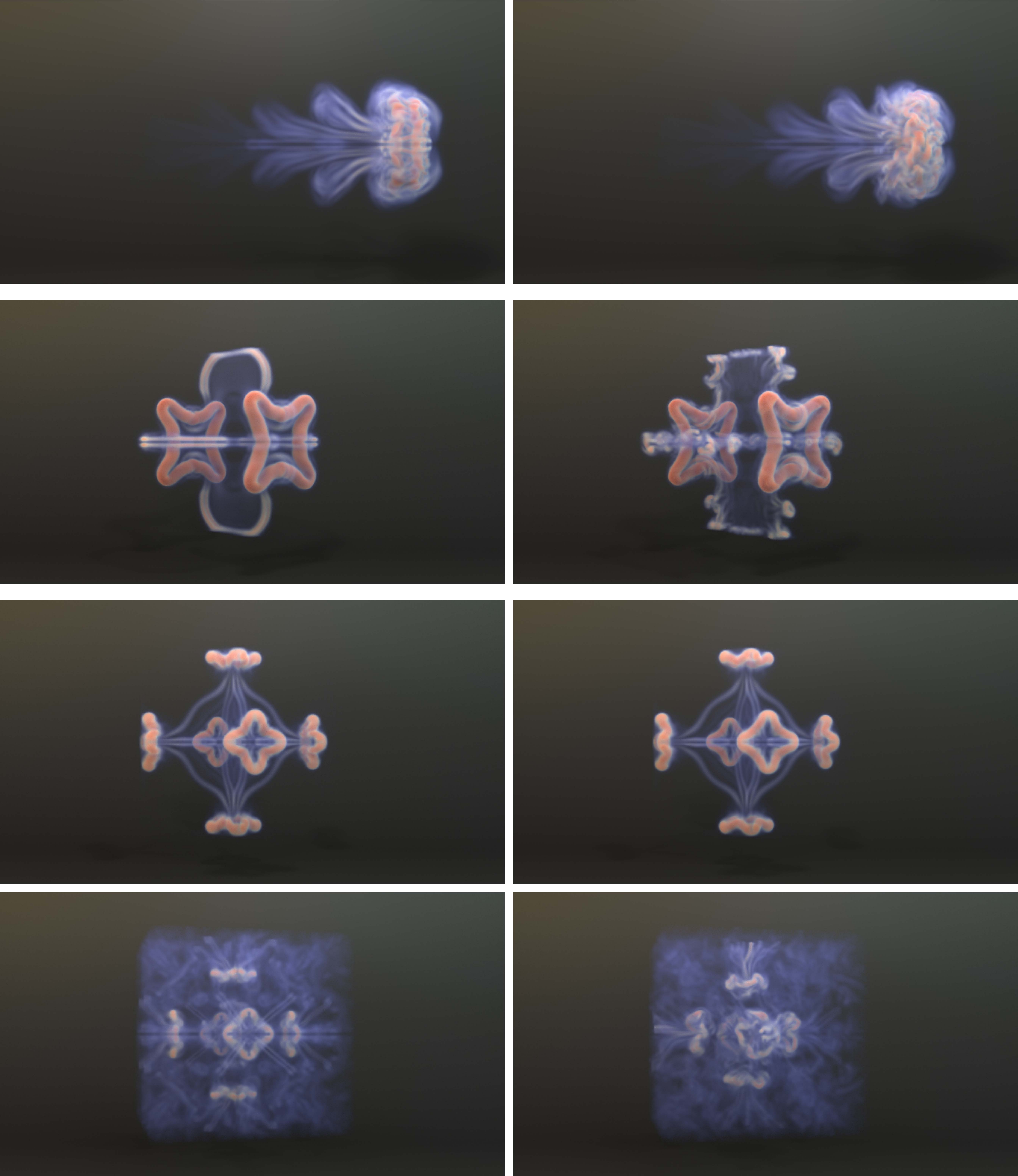}
 \caption{\revvv{Simulation results of our method (left column) and our method incorporating the Hessian term (right column). The figures on the first and second rows show the results of 3D leapfrogging vortices and four vortices collision, respectively. The bottom two rows depict the results of eight vortice collisions at $t=10s$ and $t=55s$, respectively.}}
 \label{fig:ours_use-hessian}
 % \vspace{-0.2in}
\end{figure}

\paragraph{Connection to Hybrid Eulerian-Lagrangian Methods} Our PFM method presents itself as a new type of hybrid Eulerian-Lagrangian method, which is in line with PIC \cite{harlow1962particle}, APIC \cite{jiang2015affine}, and similar methods, but one with significantly improved ability to preserve the vortical structure, which is enabled by its innovative exploitation of the particles' roles not only as samples of physical fields but as samples of high-quality and adaptive-range flow maps. With this new perspective, we devise novel mathematical formulations to endow particles with the crucial ability to evolve backward \revv{map} Jacobian $\mathcal{T}$ in a forward-simulating manner, which allows our hybrid framework to be readily bridged with the highly effective impulse-based fluid paradigm. Not only is combining Lagrangian particles with a flow map physically inspired and computationally efficient, but it also naturally empowers the development of our length-adaptive mechanism, which is shown to be key to successful simulation, as verified in our ablation tests. Inspired by the conservative transfer schemes proposed by APIC \cite{jiang2015affine} and Taylor-PIC \cite{nakamura2023taylor}, our method conservatively transfers $\bm m$ by leveraging its evolved gradients $\nabla \bm m$, and our scheme can be seen as a generalization of APIC as it allows $\nabla \bm m$ to be accurately evolved using adaptive-length flow maps rather than single-step ones. In our comparison tests with APIC, we \revv{showcase} that such an improved accuracy for $\nabla \bm m$ contributes to the simulation accuracy.

\paragraph{Further Ablation Tests on Hessian} 
\label{sec:discussion_hessian}
\revv{
We excluded the Hessian term in Equation~\ref{eq:evolve_grad_imp} in our proposed numerical implementation in Section~\ref{sec:impulse_transport}. In this section, we provide further ablation tests to evaluate the efficacy of Hessian in enhancing numerical accuracy.
We calculated backward map Hessian $\nabla \mathcal{T}_{[b, c]}^p$ for each particle $p$ by aggregating backward map Jacobian $\mathcal{T}_{[b, c]}^k$ of its neighboring particles, weighted by $\nabla w_{pk}$, where $k$ is the index of the neighboring particle.
We then included $(\nabla \mathcal{T}_{[b, c]}^p)^T \bm m_{b}$ to the right-hand side of Equation~\ref{eq:evolve_grad_imp_short} to align it with Equation~\ref{eq:evolve_grad_imp}.
Then, we conduct validation experiments with the same setup as detailed in Section~\ref{sec:compare_to_other_methods} and show the results in Table~\ref{tab:2d_leapfrog_real_world_time} and Figure~\ref{fig:ours_use-hessian} for 2D and 3D leapfrogging vortices, \revvv{four vortices collision}, and \revvv{eight vortices collision}.
Additionally, the time-varying energy for the 2D leapfrogging vortices and \revvv{four vortices collision} is depicted in Figure~\ref{fig:energy_compare}. 
From 2D/3D leapfrog experiments, we observed no significant improvement in preserving vortical structures compared to our proposed scheme. Similarly, the formula incorporating the Hessian term does not preserve symmetry structures as effectively in \revvv{four/eight vortices collision} experiments as our scheme does. We conjecture that this discrepancy in ability is due to potential inaccuracies in the backward map Hessian when directly calculated from particles, which leads to errors in mapping impulse gradients through flow maps.
% Consequently, we opt not to incorporate the Hessian term into our algorithm. 
}
%
%
%\revv{
% As shown in experiments in Sec.~\ref{sec:hessian}, Hessian term comes with overhead in calculation and lacking accuracy. 
%According to these ablation tests, we conclude that incorporating Hessian does not improve the numerical accuracy but comes with overhead in the calculation, for which we conjecture that this relates to the issue of lacking accuracy in the calculation of high-order terms from pure Lagrangian representations. Such a problem has been identified in methods such as SPH and still remains an open research question \cite{zhang2017smoothed}. }
%
\revv{Furthermore, from a theoretical perspective, we observe that our method corresponds to a specific case of transferring \(\mathcal{T}\) from particles to the grid. We demonstrate that using only the first term on the right-hand side of Equation~\ref{eq:evolve_grad_imp} is equivalent to this transfer in a PIC \cite{harlow1962particle} manner, as detailed in Appendix~\ref{sec:hessian_pic_apic}. Conversely, including the Hessian term modifies the transfer method from PIC to an APIC \cite{jiang2015affine} manner. Consequently, inspired by APIC's enhancements over PIC, we anticipate uncovering additional numerical benefits by incorporating the Hessian term towards a higher-order PFM in our future exploration.}

%\revv{
%In summary, excluding the Hessian term results in another mode of using PIC but including it results in accuracy issue mentioned above. However, addressing the issue of lacking computation accuracy for Hessian term comes beyond the scope of this paper and future research for improving accuracy in calculating high order terms from particles can easily be adapted to our method. Therefore, a potential future enhancement for PFM could involve developing an accurate method for computing the Hessian term.
%}

\paragraph{Uniform Particle Redistribution}
\revv{Uniform particle redistribution (or otherwise terms as remeshing or regridding) has been a widely used practice to address the particle distortion problem in the Vortex-In-Cell (VIC) literature \cite{koumoutsakos2008flow}. As consolidated in Figure~\ref{fig:redis_stragegy}, a uniform sampling strategy is necessary to reduce the interpolation error for the impulse grid-to-particle transfer. We conjure the reason to be the fact that we are transferring $\mathcal{T}$ from particles to the grid, as shown in the previous discussion, and when reinitialization happens, uniform redistributing particles gives a more structured reconstruction of its current space. %As a flow map serves as a geometric representation of space, aiming for a more structured reconstruction of the current space comes with the same reasoning as scenarios requiring re-meshing and regrinding.
}

\paragraph{Limitations \& Future Works}
Our method is currently subject to a few limitations. First, it currently considers Euler's equation for inviscid fluid flows, and the incorporation of viscosity and diverse external forces remains an open problem for our flow map-based framework. Secondly, 
handling interfacial phenomena still proves challenging for our impulse-based formulation as it introduces additional terms in the Poisson equation which brings about considerable numerical instability. Solving these technical challenges will allow free-surface liquids to be simulated which unlocks a new range of vortical phenomena like toroidal bubbles. Finally, we currently employ kinematic coupling for handling solid-fluid interactions, and the development of dynamic two-way coupling schemes in the PFM framework would be an exciting future challenge.

\section*{Acknowledgements}
\revv{
We express our gratitude to the anonymous reviewers for their insightful feedback. We especially appreciate Zhiqi Li for his insights on the connection between the Hessian term and the transfer of the backward map Jacobian. We also thank Jinyuan Liu and Taiyuan Zhang for their insightful discussion. Georgia Tech authors acknowledge NSF IIS \#2313075, ECCS \#2318814, CAREER \#2420319, IIS \#2106733, OISE \#2153560, and CNS \#1919647 for funding support. We credit the Houdini education license for video animations.
}
% \input{limitation}

% \clearpage

% \newpage
{
% \small
% \nocite{*} 
\bibliography{ref}
\bibliographystyle{ACM-Reference-Format}
% \printbibliography
}

\newpage
\appendix
% \section{Additional Experiment Figures}
% \revv{
% The simulation results of taylor vortex and Karman vortex street are shown in Figure~\ref{fig:taylor_vortex} and~\ref{fig:karman_vortex}, respectively.
% }
% \setlength{\abovecaptionskip}{12pt}
% \begin{figure}[t]
%  \centering
%  \includegraphics[width=.99\columnwidth]{img/examples/taylor_vortex.png}
%  \caption{\revv{Initial setting of taylor vortex and the frame of seperation}}
%  \label{fig:taylor_vortex}
% \end{figure}

% \setlength{\abovecaptionskip}{12pt}
% \begin{figure}[t]
%  \centering
%  \includegraphics[width=.99\columnwidth]{img/examples/karman_vortex.png}
%  \caption{\revv{Karman vortex street at t=20 and t=60}}
%  \label{fig:karman_vortex}
% \end{figure}
\section{Code}
% \revvv{Code and data for this paper are at www.dummy.url.}
\revvv{Our code is available at \href{https://github.com/zjw49246/particle-flow-maps}{https://github.com/zjw49246/particle-flow-maps}. We implement our methods with Taichi Programming Language \cite{hu2019taichi}.}

\section{Evolution of the Jacobian matrix of the flow map}
\label{sec:Jacobian_prov}

\begin{theorem}

Assuming
\begin{equation}
\begin{dcases}
\mathcal{F}(\bm \phi,t) =\frac{\partial \bm \phi (\bm x,t)}{\partial \bm x}\\
\frac{\partial \bm \phi(\bm x,t)}{\partial t} = \bm u [\bm \phi(\bm x, t), t]
\end{dcases}
\label{eq:Fxt}
\end{equation}
then
\begin{equation}
\frac{\partial \mathcal{F}_{ij}(\bm x,t)}{\partial t} + u_k(\bm x,t)\frac{\partial \mathcal{F}_{ij}(\bm x,t)}{\partial x_k}=\frac{\partial  u_i (\bm x,t )}{\partial \bm x_k} \mathcal{F}_{kj}(\bm x, t).
\label{eq:DFphi}
\end{equation}
\end{theorem}

\begin{proof}

It can be deduced from the chain rule of differentiation and the second equation in Equation~\ref{eq:Fxt} that:
\begin{equation}
\begin{aligned}
\frac{D \mathcal{F}_{ij}(\bm \phi,t)}{D t} &= \frac{\partial \mathcal{F}_{ij}(\bm \phi,t)}{\partial t} + \frac{\partial \mathcal{F}_{ij}(\bm \phi,t)}{\partial \phi_k} \frac{\partial \phi_k(\bm x,t)}{\partial t} \\
&= \frac{\partial \mathcal{F}_{ij}(\bm \phi,t)}{\partial t} + u_k(\bm \phi,t)\frac{\partial \mathcal{F}_{ij}(\bm \phi,t)}{\partial \phi_k} .
\end{aligned}
\label{eq:dFphi2}
\end{equation}

From the expression in Equation~\ref{eq:Fxt}, we also obtain:
\begin{equation}
\begin{aligned}
\frac{D \mathcal{F}_{ij}(\bm \phi,t)}{D t} 
&=\frac{\partial \phi_i(\bm x,t)}{\partial t \partial x_j}
=\frac{\partial u_i (\bm \phi,t )}{\partial x_j}\\
&=\frac{\partial  u_i (\bm \phi,t )}{\partial \bm \phi_k}\frac{\partial \phi_k(\bm x,t)}{\partial x_j}
=\frac{\partial  u_i (\bm \phi,t )}{\partial \bm \phi_k} \mathcal{F}_{kj}(\bm \phi, t).
\end{aligned}
\label{eq:dFphi3}
\end{equation}

Combining Equations~\ref{eq:dFphi2} and \ref{eq:dFphi3} gives \ref{eq:DFphi}.
\end{proof}

\begin{theorem}
Assuming
\begin{equation}
\begin{dcases}
\mathcal{T}(\bm x,t) =\frac{\partial \bm \psi(\bm x,t) }{\partial \bm x}\\
\frac{\partial \bm \psi(\bm x,t)}{\partial t} = -\frac{\partial \bm \psi(\bm x,t)}{\partial x_i} u_i (\bm x, t)
\end{dcases}
\label{eq:Txt}
\end{equation}
then
\begin{equation}
\frac{\partial  \mathcal{T}_{ij}(\bm x,t)}{\partial t}+u_k(\bm x,t)\frac{\partial \mathcal{T}_{ij}(\bm x,t)}{\partial x_k} =- \mathcal{T}_{ik}(\bm x,t)\frac{\partial u_{k}(\bm x,t)}{\partial x_j}.
\end{equation}
\end{theorem}

\begin{proof}

It can be deduced from Equation~\ref{eq:Txt} that:
\begin{equation}
\begin{aligned}
\frac{\partial  \mathcal{T}_{ij}(\bm x,t)}{\partial t}&=
\frac{\partial^2 \psi_i(\bm x,t)}{\partial  t \partial x_j}
=-\frac{\partial }{\partial x_j}\left[\frac{\partial \psi_i(\bm x,t)}{\partial x_k} u_k (\bm x, t)\right]\\
&=-u_k(\bm x,t)\frac{\partial \mathcal{T}_{ij}(\bm x,t)}{\partial x_k} - \mathcal{T}_{ik}(\bm x,t)\frac{\partial u_{k}(\bm x,t)}{\partial x_j} .
\end{aligned}
\end{equation}
\end{proof}

\section{Evolution of Impulse and Its Gradients}
\label{sec:mdm}

\begin{theorem}
\begin{equation}
\bm m(\bm x ,t)= m_i(\bm \psi ,0)\frac{\partial \psi_i(\bm x,t)}{\partial \bm x}.
\label{eq:mphit}
\end{equation}
\end{theorem}
\begin{proof}
Let us posit
\begin{equation}
\bm g(\bm x,t) \equiv \bm m(\bm \phi ,t)- m_i(\bm x ,0)\frac{\partial \psi_i(\bm \phi,t)}{\partial \bm \phi}.
\end{equation}
Clearly, $\bm g(\bm x,0) = \bm 0$. Furthermore,
\begin{equation}
\begin{aligned}
    &\frac{\partial \bm g(\bm x,t)}{\partial t }\\
    =&\frac{\partial \bm m(\bm \phi ,t)}{\partial t}+\frac{\partial \bm m(\bm \phi ,t)}{\partial\phi_i}\frac{\partial\phi_i}{\partial t}\\
   & - m_i(\bm x ,0)\left[\frac{\partial^2 \psi_i(\bm \phi,t)}{\partial \bm \phi \partial t}+\frac{\partial^2 \psi_i(\bm \phi,t)}{\partial \bm \phi \partial \phi_j}\frac{\partial \phi_j(\bm x,t)}{\partial t}\right]\\
    =& -\frac{\partial u_i(\bm \phi,t)}{\partial \bm \phi} m_i(\bm \phi,t)\\
   & -m_i(\bm x ,0)\left\{\frac{\partial}{\partial \bm \phi} \left[-\frac{\partial \psi_i(\bm \phi,t)}{\partial \bm \phi_j}u_j(\bm \phi,t)\right]+\frac{\partial^2 \psi_i(\bm \phi,t)}{\partial \bm \phi \partial \phi_j}u_j(\bm \phi,t)\right\}\\
    =& -\frac{\partial u_j(\bm \phi,t)}{\partial \bm \phi} m_j(\bm \phi,t)+ m_i(\bm x ,0)\frac{\partial u_j(\bm \phi,t)}{\partial \bm \phi} \frac{\partial \psi_i(\bm \phi,t)}{\partial \bm \phi_j}\\
    =&-g_i(\bm x,t)\frac{\partial  u_i(\bm \phi,t)}{\partial \bm \phi }.
\end{aligned}
\end{equation}
Thus, $\bm g$ satisfies
\begin{equation}
    \begin{dcases}
    \frac{\partial \bm g(\bm x,t)}{\partial t }=-g_i(\bm x,t)\frac{\partial  u_i(\bm \phi,t)}{\partial \bm \phi },\\
    \bm g(\bm x,0) = \bm 0
    \end{dcases}
\end{equation}
The unique solution to this equation is $\bm g(\bm x,t )\equiv \bm 0$.
\end{proof}

\begin{theorem}
    \begin{equation}
    \begin{aligned}
    \label{eq:evolve_grad_imp_i}
    &\frac{\partial m_j(\bm x,t)}{\partial x_k} \\
    =& 
    \frac{\partial m_i(\bm \psi,0)}{\partial \psi_l}\frac{\partial \psi_l(\bm x, t)}{\partial x_k} \frac{\partial \psi_i(\bm x, t)}{\partial x_j} + m_i(\bm \psi,0) \frac{\partial^2 \psi_i(\bm x, t)}{\partial x_j \partial x_k}.
    \end{aligned}
\end{equation}
\end{theorem}
Taking the derivative of Equation~\ref{eq:mphit} with respect to \revv{$x_k$} directly yields Equation~\ref{eq:evolve_grad_imp_i}.
We remark that, due to Equation~\ref{eq:FT}, Equations~\ref{eq:evolve_grad_imp_i} and \ref{eq:evolve_grad_imp} are equivalent.

\section{Equivalence of forward evolution and backward evolution of $\mathcal{T}$}
\label{sec:proof_T_forward_backward}
Consider a virtual particle that moves from $\bm{x}_0$ to $\bm{x}_n$ in $n$ timesteps. In timestep $i$, the particle occupies a position $\bm{x}_i$, moves with velocity $\bm{u}_i$ and is associated with the backward map Jacobian, denoted as $\mathcal{T}_i$. Initially, $\mathcal{T}_0$ is initialized as the identity matrix.

At timestep $n$, to get the impulse at $\bm{x}_n$, $\mathcal{T}_n$ has to be calculated. In the process described by \citet{deng2023fluid}, which utilizes an Eulerian grid and a neural velocity buffer, the calculation of $\mathcal{T}$ involves backtracing the virtual particle's trajectory from $\bm{x}_n$ to $\bm{x}_0$. This process progresses in reverse chronological order, therefore we refer to it as "backward evolution of $\mathcal{T}$". Conversely, in our method, we adopt a "forward evolution of $\mathcal{T}$," where Lagrangian particles progress from $\bm{x}_0$ to $\bm{x}_n$, and $\mathcal{T}$ is updated sequentially at each timestep. During backward evolution, the update of $\mathcal{T}$ in timestep $i$ involves left-multiplying the existing result with $\left(\textbf{I} - \nabla \bm{u}_i\Delta t\right)$, proceeding from step $n$ to $1$. In forward evolution, however, the update of $\mathcal{T}$ in timestep $i$ entails right-multiplying the existing result with $\left(\textbf{I} - \nabla \bm{u}_i\Delta t\right)$, progressing from step $1$ to $n$. Despite these differing approaches, both backward and forward evolution start from the identity matrix and ultimately converge to the same result:
\begin{equation}
    \mathcal{T}_n = \Pi_{i=1}^n\left(\textbf{I}-\nabla\bm{u}_i\Delta t\right).
\end{equation}

% \junweirev{
% \section{Evolution of $\nabla \rho$}
% \label{sec:proof_advect_grad_rho}
% In this section, we prove the advection equation for $\nabla \rho$ follows the same advection equation for $\bm{m}$ and should be denoted as 
% \begin{equation}
% \nabla \rho(\bm{x}, t) = \mathcal{T}^T(\bm{x})\nabla \rho(\bm\psi(\bm x), 0)
% \end{equation}
% For a material quantity $\rho$, its material derivative is given by $\frac{D \rho}{Dt} = 0$. Expanding this by definition of total derivative and substituding $\nabla \rho$ to this equation will give us:
% \begin{align}
%     \frac{\partial \rho}{\partial t} + \bm{u} \cdot \nabla \rho &= 0\\
%     \frac{\partial \rho}{\partial t} + u_i \frac{\partial \rho}{\partial x_i} &= 0\\
%     \frac{\partial}{\partial t} \frac{\partial \rho}{\partial x_j} + \frac{\partial u_i}{\partial x_j}\frac{\partial \rho}{\partial x_i} + u_i \frac{\partial^2 \rho}{\partial x_i \partial x_j} &= 0\\
%     \frac{\partial}{\partial t}(\frac{\partial \rho}{\partial x_j}) + u_i \frac{\partial}{\partial x_i} \frac{\partial \rho}{\partial x_j} &= - \frac{\partial u_i}{\partial x_j}\frac{\partial \rho}{\partial x_i}\\
%     \frac{D}{Dt}\left(\nabla \rho\right)&= -\nabla \rho \cdot \nabla\bm{u}
% \end{align} 
% Observe the material derivative is the same for $\nabla \rho$ as for $\bm{m}$. We conclude it follows the advection equation defined in Equation~\ref{eq:evolve_imp}.
% }

\section{Additional Pseudocode}
\label{sec:additional_pseudocode}
In this section, we provide additional pseudocodes to supplement Section~\ref{sec:Eulerian_Lagrangian_Framework}.

Algorithm~\ref{alg:midpoint} outlines the second-order, midpoint method, and Algorithm~\ref{alg:RK4} details the procedure for our custom RK4 integration scheme.

\begin{algorithm}
\caption{Midpoint Method}
\label{alg:midpoint}
\begin{flushleft}
        \textbf{Input:} $\bm{u}$~\\
        \textbf{Output:} $\bm{u}^\text{mid}$
\end{flushleft}
\begin{algorithmic}[1]
\State Reset $\bm\psi, \mathcal{T}$ to identity;
\State March $\bm\psi, \mathcal{T}$ with $\bm{u}$ and $-0.5\Delta t$ using Alg. 2 in \citet{deng2023fluid};
\State $\bm{m}^\text{mid} \gets \mathcal{T}^T\bm{u}(\bm\psi)$;
\State $\bm{u}^\text{mid} \gets \textbf{Poisson}({\bm{m}}^\text{mid} )$;
\end{algorithmic}
\end{algorithm}

\begin{algorithm}
\caption{Interleaved RK4 for $\bm x$, $\mathcal{T}$}
\label{alg:RK4}
\begin{flushleft}
        \textbf{Input:} $\bm{u}$, $\bm x$, $\mathcal{T}$, $\Delta t$~\\
        \textbf{Output:} $\bm x_\text{next}$, $\mathcal{T}_\text{next}$
\end{flushleft}
\begin{algorithmic}[1]
\State $(\bm{u}_1, \nabla \bm{u}\vert_1) \gets \textbf{Interpolate}(\bm{u}, \bm x)$;
\State $\frac{\partial \mathcal{T}}{\partial t}\vert_1 \gets (\nabla \bm{u}\vert_1)^T \mathcal{T}$;
% \State $\frac{\partial \mathcal{T}'}{\partial t}\vert_1 \gets \nabla \bm{u}\vert_1 \mathcal{T}'$;
\State $\bm x_1 \gets \bm x + 0.5 \cdot \Delta t \cdot \bm{u}_1$;
\State $\mathcal{T}_1 \gets \mathcal{T} - 0.5 \cdot \Delta t \cdot \frac{\partial \mathcal{T}}{\partial t}\vert_1$;
% \State $\mathcal{T}'_1 \gets \mathcal{T}' - 0.5 \cdot \Delta t \cdot \frac{\partial \mathcal{T}'}{\partial t}\vert_1$;
\State $(\bm{u}_2, \nabla \bm{u}\vert_2)\gets \textbf{Interpolate}(\bm{u}, \bm x_1)$;
\State $\frac{\partial \mathcal{T}}{\partial t}\vert_2 \gets (\nabla \bm{u}\vert_2)^T \mathcal{T}_1$;
% \State $\frac{\partial \mathcal{T}'}{\partial t}\vert_2 \gets \nabla \bm{u}\vert_2 \mathcal{T}'_1$;
\State $\bm x_2 \gets \bm x + 0.5 \cdot \Delta t \cdot \bm{u}_2$;
\State $\mathcal{T}_2 \gets \mathcal{T} - 0.5 \cdot \Delta t \cdot \frac{\partial \mathcal{T}}{\partial t}\vert_2$;
% \State $\mathcal{T}'_2 \gets \mathcal{T}' - 0.5 \cdot \Delta t \cdot \frac{\partial \mathcal{T}'}{\partial t}\vert_2$;
\State $(\bm{u}_3, \nabla \bm{u}\vert_3)\gets \textbf{Interpolate}(\bm{u}, \bm x_2)$;
\State $\frac{\partial \mathcal{T}}{\partial t}\vert_3 \gets (\nabla \bm{u}\vert_3)^T \mathcal{T}_2$;
% \State $\frac{\partial \mathcal{T}'}{\partial t}\vert_3 \gets \nabla \bm{u}\vert_3 \mathcal{T}'_2$;
\State $\bm x_3 \gets \bm x + \Delta t \cdot \bm{u}_3$;
\State $\mathcal{T}_3 \gets \mathcal{T} - \Delta t \cdot \frac{\partial \mathcal{T}}{\partial t}\vert_3$;
% \State $\mathcal{T}'_3 \gets \mathcal{T}' - \Delta t \cdot \frac{\partial \mathcal{T}'}{\partial t}\vert_3$;
\State $(\bm{u}_4, \nabla \bm{u}\vert_4)\gets \textbf{Interpolate}(\bm{u}, \bm x_3)$;
\State $\frac{\partial \mathcal{T}}{\partial t}\vert_4 \gets (\nabla \bm{u}\vert_4)^T \mathcal{T}_3$;
% \State $\frac{\partial \mathcal{T}'}{\partial t}\vert_4 \gets \nabla \bm{u}\vert_4 \mathcal{T}'_3$;
\State $\bm x_\text{next} \gets \bm x + \Delta t \cdot \frac{1}{6}\cdot (\bm{u}_1 + 2 \cdot \bm{u}_2 + 2 \cdot \bm{u}_3 + \bm{u}_4)$;
\State $\mathcal{T}_\text{next} \gets \mathcal{T} - \Delta t \cdot  \frac{1}{6} \cdot (\frac{\partial \mathcal{T}}{\partial t}\vert_1 
 + 2 \cdot \frac{\partial \mathcal{T}}{\partial t}\vert_2 + 2 \cdot\frac{\partial \mathcal{T}}{\partial t}\vert_3 + \frac{\partial \mathcal{T}}{\partial t}\vert_4)$;
 % \State $\mathcal{T}'_\text{next} \gets \mathcal{T}' - \Delta t \cdot  \frac{1}{6} \cdot (\frac{\partial \mathcal{T}'}{\partial t}\vert_1 
 % + 2 \cdot \frac{\partial \mathcal{T}'}{\partial t}\vert_2 + 2 \cdot\frac{\partial \mathcal{T}'}{\partial t}\vert_3 + \frac{\partial \mathcal{T}'}{\partial t}\vert_4)$;
\end{algorithmic}
\end{algorithm}

\section{Backward Map Hessian}
\label{sec:hessian_pic_apic}
\revv{
Excluding the Hessian term from Equation~\ref{eq:evolve_grad_imp} is equivalent to substituting Equation~\ref{eq:evolve_grad_imp_short} into Equation~\ref{eq:imp_P2G}. This leads to the following equations:
\begin{equation}
    \begin{aligned}
        \bm m_i &= \sum_p w_{ip}(\bm m_c^p + \nabla \bm m_c^p (\bm x_i - \bm x_p)) \,/\, \sum_p w_{ip} \\
        &= \sum_p w_{ip}(\bm m_c^p + (\mathcal{T}_{[b, c]}^p)^T \nabla\bm m_b^p \mathcal{T}_{[b, c]}^p (\bm x_i - \bm x_p)) \,/\, \sum_p w_{ip} \\
        &= \sum_p w_{ip}(\bm m_c^p + (\mathcal{T}_{[b, c]}^p)^T \nabla\bm m_b^p (\bm\psi(\bm x_i) - \bm\psi(\bm x_p))) \,/\, \sum_p w_{ip} \\
        &= \sum_p w_{ip}(\bm m_c^p + (\mathcal{T}_{[b, c]}^p)^T (\bm m_b^p[\bm\psi(\bm x_i)] - \bm m_b^p[\bm\psi(\bm x_p)])) \,/\, \sum_p w_{ip} \\
        &= \sum_p w_{ip}(\bm m_c^p + ((\mathcal{T}_{[b, c]}^p)^T\bm m_b^p[\bm\psi(\bm x_i)] \\
        &\quad\quad\quad\quad\quad\quad - (\mathcal{T}_{[b, c]}^p)^T\bm m_b^p[\bm\psi(\bm x_p)])) \,/\, \sum_p w_{ip} \\
        &= \sum_p w_{ip}(\bm m_c^p + ((\mathcal{T}_{[b, c]}^p)^T\bm m_b^p[\bm\psi(\bm x_i)] - \bm m_c^p)) \,/\, \sum_p w_{ip} \\
        &= (\sum_p w_{ip}(\mathcal{T}_{[b, c]}^p)^T  \,/\, \sum_p w_{ip})\bm m_b^p[\bm\psi(\bm x_i)] \\
    \end{aligned}
\end{equation}
where the third and fourth equality marks rely on the assumption that $\bm x_p$ is very close to $\bm x_i$. This formulation implies interpolating $\mathcal{T}$ from particles to the grid in a PIC manner. Subsequently, this interpolated $\mathcal{T}$ is utilized to evolve the impulse at $\bm\psi(\bm x_i)$ from time $b$ to time $c$. Here, $\bm\psi(\bm x_i)$ represents the backtracked location at time $b$ for the particle currently at $\bm x_i$.
}

\revv{
If we incorporate the Hessian term, we have the following equations:
\begin{equation}
    \begin{aligned}
        \bm m_i &= \sum_p w_{ip}(\bm m_c^p + \nabla \bm m_c^p (\bm x_i - \bm x_p)) \,/\, \sum_p w_{ip} \\
        &= \sum_p w_{ip}(\bm m_c^p + ((\mathcal{T}_{[b, c]}^p)^T \nabla\bm m_b \mathcal{T}_{[b, c]}^p \\
        &\quad\quad\quad\quad\quad\quad+ (\nabla\mathcal{T}_{[b, c]}^p)^T\bm m_b) (\bm x_i - \bm x_p)) \,/\, \sum_p w_{ip} \\
        &= \sum_p w_{ip}(\bm m_c^p + (\mathcal{T}_{[b, c]}^p)^T \nabla\bm m_b \mathcal{T}_{[b, c]}^p (\bm x_i - \bm x_p)) \,/\, \sum_p w_{ip} \\
        &\quad\quad\quad\quad\quad\quad+ \sum_p w_{ip}(\nabla\mathcal{T}_{[b, c]}^p)^T\bm m_b (\bm x_i - \bm x_p)) \,/\, \sum_p w_{ip} \\
        &= (\sum_p w_{ip}(\mathcal{T}_{[b, c]}^p)^T  \,/\, \sum_p w_{ip})\bm m_b^p[\bm\psi(\bm x_i)] \\
        &\quad\quad\quad\quad\quad\quad+ (\sum_p w_{ip}(\nabla \mathcal{T}_{[b, c]}^p)^T(\bm x_i - \bm x_p)  \,/\, \sum_p w_{ip})\bm m_b^p[\bm\psi(\bm x_i)] \\
        &= (\sum_p w_{ip}((\mathcal{T}_{[b, c]}^p)^T + (\nabla \mathcal{T}_{[b, c]}^p)^T(\bm x_i - \bm x_p))  \,/\, \sum_p w_{ip})\bm m_b^p[\bm\psi(\bm x_i)] \\
    \end{aligned}
\end{equation}
where the fourth equality mark relies on the assumption that $\bm x_p$ is very close to $\bm x_i$. This formula implies that incorporating the Hessian term is equivalent to shifting the interpolation of $\mathcal{T}$ to APIC manner.
}

\section{Interpolation Details}
% \paragraph{MAC Grid} \junwei{TODO}
% \paragraph{Interpolation}
\label{sec:interp_kernel}
When interpolating values between $x_k$ and $x_l$, we apply the MPM \cite{jiang2016material} interpolation scheme using the quadratic kernel:
\begin{equation}
w_{kl}=
\begin{cases}
% \frac{3}{4} - \vert x \vert^2 \ \ \ 0 \leq \vert x \vert < \frac{1}{2}\\
% \frac{1}{2}(\frac{3}{2}-\vert x \vert)^2 \ \ \ 0 \leq \vert x \vert < \frac{1}{2}\\
% x(n-1)
    \frac{3}{4} - \vert x_k - x_l \vert^2 & 0 \leq \vert x_k - x_l \vert < \frac{1}{2},\\
    \frac{1}{2}(\frac{3}{2}-\vert x_k - x_l \vert)^2 & \frac{1}{2} \leq \vert x_k - x_l \vert < \frac{3}{2},\\
    0    & \frac{3}{2} \leq \vert x_k - x_l \vert.
\end{cases}
\end{equation}

\end{document}